\documentclass[12pt]{article}
\usepackage[margin=1in]{geometry}
\usepackage{setspace}
\usepackage{amsmath, amssymb, amsthm}
\usepackage{bm}
\usepackage{mathtools}
\usepackage{enumitem}
\usepackage{booktabs}
\usepackage{graphicx}
\usepackage[round]{natbib}
\usepackage[hidelinks]{hyperref}
\usepackage{cleveref}

\usepackage{xcolor}

\theoremstyle{plain}
\newtheorem{theorem}{Theorem}
\newtheorem{proposition}{Proposition}
\newtheorem{corollary}{Corollary}
\newtheorem{lemma}{Lemma}
\theoremstyle{definition}
\newtheorem{assumption}{Assumption}

\theoremstyle{remark}
\newtheorem{remark}{Remark}

\newcommand{\E}{\mathbb{E}}
\newcommand{\Cov}{\operatorname{Cov}}
\newcommand{\Var}{\operatorname{Var}}
\newcommand{\Prob}{\mathbb{P}}
\newcommand{\R}{\mathbb{R}}

\newcommand{\Ptilt}[1]{\mathbb{E}^{#1}}
\newcommand{\Hz}{\mathcal{H}_Z}
\newcommand{\Pn}{\mathbb{P}_n}
\newcommand{\corgauss}{Corollary~\ref{cor:gaussian}}
\newcommand{\thmpoint}{Theorem~\ref{thm:point}}
\newcommand{\thmset}{Theorem~\ref{thm:set}}
\newcommand{\thmeif}{Theorem~\ref{thm:eif}}
\newcommand{\lemdcdr}{Lemma~\ref{lem:dcdr}}
\newcommand{\lemgamma}{Lemma~\ref{lem:gammarobust}}

\newcommand{\T}{\mathcal{T}}
\crefname{assumption}{Assumption}{Assumptions}
\Crefname{assumption}{Assumption}{Assumptions}
\providecommand{\backmatter}{}

\begin{document}

\begin{center}
{\Large\bfseries Bridge-Anchored Partial Identification of a Target Mean
under Outcome-Model Drift}\\[1.2em]
{\large Danhyang Lee, Shinyoung Jeon, and Shu Yang}\\[0.6em]
{\normalsize Correspondence: \texttt{danhyang@smu.edu}}
\end{center}

\medskip
\begin{center}\bfseries Summary\end{center}
\begin{quote}
Transporting an outcome relationship from a source population to a target
population where the outcome is unobserved requires the conditional outcome law to
be stable across populations. When it is not, \emph{outcome-model drift}, the
target mean is not point-identified, and existing transportability estimators are
biased, with no account of what remains unknown. We study settings, common
in cross-cohort educational and biomedical data integration, in which a set of
\emph{bridge outcomes} is recorded in both populations and can be expected to
drift alongside the target outcome. Our central object is an \emph{identified set}.
The bridge splits the drift into two parts: a component the bridge can detect,
fixed by the data through a bridge-matching equation, and a residual component, drift in the outcome orthogonal to the bridge, which no
observed quantity restricts.
We isolate the residual as a single scalar sensitivity parameter $\kappa$, which
suffices because the target is a mean, so that functional residual drift moves it
only through a one-dimensional projection. We prove that the identified set has a
closed-form \emph{irreducible core}
whose width is governed by the residual bound and by how much of the outcome the
bridge leaves unexplained, and which no sample size narrows,
and show that the center of the set is first-order invariant to the
residual by construction. The identification, efficiency, and core formulas hold
for any working exponential family, with Gaussian and binary outcomes as
instances. Point identification arises only as the $\kappa=0$ benchmark. We develop
a debiased, drift-augmented estimator, establish its semiparametric efficiency at
anchored sensitivity through an efficient influence function carrying a
bridge-moment correction, characterize its double robustness conditional on the
bridge-identified drift, and give
rate-robust inference for the set via an Imbens--Manski interval.
Against a bridge-blind sensitivity analysis the gain is one of \emph{elicitation}
rather than of arithmetic: expressed on the outcome scale, the bridge-blind
analyst must bound the entire drift channel, while we bound only its
bridge-orthogonal part, a strictly smaller quantity.
Simulations confirm that the
drift estimator is unbiased under co-drift and the benchmark's bias is small
relative to the identification width, that the core width follows its closed form
and falls as the bridge explains more of the outcome, and that the set fails
visibly once the residual bound is exceeded; a separate exact
binary-outcome study measures the quantities the Gaussian working model leaves
untestable. The motivating setting is the integration of two survey cohorts
whose kindergarten mathematics outcome is available in only one.
\end{quote}

\medskip
\noindent{\itshape\bfseries Key words:} Data fusion; Doubly robust estimation;
Exponential tilting; Fractional imputation; Partial identification; Proximal
inference; Semiparametric efficiency; Sensitivity analysis; Transportability.

\bigskip\hrule\bigskip

\section{Introduction}
\label{sec:intro}

A recurring task in the integration of large-scale surveys is to characterize an
outcome in a population where that outcome was never measured, by borrowing its
relationship to covariates from an earlier population where it was. In our
motivating application, two nationally representative cohorts of a federal early
childhood program are observed a few years apart; the earlier cohort is followed
into kindergarten and its mathematics achievement recorded, while the later
cohort, the one whose kindergarten outcomes policymakers actually care about, was
by design not followed past preschool. The standard approach is to fit
the preschool-to-kindergarten relationship in the earlier (source) cohort and
apply it to the later (target) cohort, which is valid only if that relationship
is stable across the two cohorts. Transportability and data-fusion
methods \citep{dahabreh2019generalizing,pearl2014external,bareinboim2016fusion}
formalize exactly this requirement: they permit the covariate distribution to
differ across populations (\emph{covariate shift}) but require the conditional
outcome law $P(Y\mid X)$ to transport unchanged. In cross-cohort integration the
second requirement is the one likely to fail.
Programs are reformed, populations change, and instruction shifts between data
collections, and the two cohorts that motivate this paper differ on child,
family and classroom characteristics.
When $P(Y\mid X)$ itself changes, \emph{outcome-model drift}, and the target
outcome is unobserved, the target mean is \emph{not point-identified}. Applying a
transportability estimator anyway returns a single number whose bias is neither
bounded nor even signed by the observed data.

The feature we exploit is that integration problems of this kind almost always
carry \emph{bridge outcomes}: intermediate measurements, in our case
preschool-exit mathematics, early literacy and executive-function scores, recorded
in \emph{both} cohorts. Because a bridge is observed in the target, any drift in the
bridge relationship is directly visible. One approach would treat the bridge as a
surrogate and read the outcome drift off the bridge drift. Bridge-internal agreement cannot certify that the outcome
drifts as the bridge does, because the outcome is unobserved in the target. We do
not claim otherwise; the extrapolation from bridge drift to outcome drift is, and
remains, untestable. We therefore construct the framework around that untestability. We link source and target by a
scalar-indexed exponential tilt and decompose the tilt into two orthogonal
channels. The first is the drift the bridge can see; because the bridge is
observed in both cohorts, this channel is fixed by matching the target bridge mean
to a tilted source bridge mean, an equation involving only observed quantities.
The second is the drift in the outcome that is \emph{orthogonal to the bridge}: %
defined, not assumed, as the projection residual of the outcome-drift direction
onto functions of the bridge. This residual is the untestable component, and we do not identify it; we isolate it as a single
scalar sensitivity parameter and report an identified set as the parameter ranges
over a bound.

This design yields five contributions. First, our estimand is the identified set
for the target mean $\mu_1=\E[Y\mid S=1]$ under a bounded residual drift
$|\kappa|\le\bar\kappa$, and point identification is not our goal but appears only
as the $\kappa=0$ benchmark (\Cref{thm:point}), corresponding to the
assumption, which we never assert as testable, that the outcome drifts entirely
within the subspace the bridge can detect. Second, we prove that the
bridge-detectable drift is fixed by a bridge-matching equation in observed
quantities (\Cref{thm:point}) and that, because the residual is defined by
projection, the center of the identified set is
\emph{ stationary at the benchmark}
(\Cref{thm:set}): the sensitivity parameter opens the
set around a centre that does not move at first order, a property of the
projection rather than of the data (\Cref{rem:whyproject}), with a
closed-form \emph{irreducible core}
whose half-width depends on the residual bound and on how much of the outcome the
bridge and the covariates leave unexplained; a richer bridge sharpens the center
and, through the residual variance, the core as well; what neither a richer bridge
nor a larger sample can do is remove the core.
That a single scalar suffices is
not an assumption but a
consequence of the estimand: because the target is a mean, an
infinite-dimensional residual drift moves it only through a one-dimensional
projection, so the scalar $\kappa$ captures the entire first-order sensitivity of
$\mu_1$ (\Cref{prop:scalar}). Third, we give a debiased, drift-augmented estimator
(\Cref{sec:est}) and derive its efficient influence function at anchored
sensitivity (\Cref{thm:eif}), which carries, beyond the usual imputation and
propensity corrections, a \emph{bridge-moment correction} accounting for the
estimation of the detectable drift, and we establish Neyman-orthogonal,
rate-robust asymptotics (\Cref{thm:orth}) that tolerate flexible nuisance
estimation; inference for the set is by an Imbens--Manski interval
\citep{imbens2004confidence}, and because the sensitivity parameter is anchored
rather than estimated, the interval carries no sensitivity-parameter variance
term. Fourth, we do not claim triple robustness but show instead
(\Cref{lem:dcdr,lem:gammarobust}) that the drift is identified from the bridge
margin alone, and is therefore insulated from misspecification of the outcome
model, and that, given the bridge-identified drift, the estimator is doubly
robust in the cohort-propensity and source-outcome models, so the three nuisances
are decoupled, not entangled. Fifth, because a practitioner could ignore the
bridge and sweep the entire outcome drift as sensitivity, discarding the
information that the bridge is observed, we compare the two analyses.
The comparison must be made on the outcome scale, since the two sensitivity
parameters index displacements along different directions and are not
interchangeable as numbers (\Cref{sec:blind}); made that way, the point is that
the bridge-blind analyst must bound the entire drift channel while we bound only
its bridge-orthogonal component. The bridge therefore reduces the size of the
judgement required, and the narrower set follows from that reduction.

These contributions position the method against four existing frameworks.
Standard transportability
\citep{dahabreh2019generalizing,dahabreh2023meta,westreich2017transportability}
assumes the residual away and returns a point; we retain it as the object of a
sensitivity analysis. Closest to us in spirit is the exponential-tilt sensitivity
analysis of \citet{steingrimsson2024sensitivity}, who study exactly our data
configuration, source covariates and outcome, target covariates only, and, like
us, index violations of conditional transportability by an exponential tilt. The
two analyses differ in the following respect. Their analysis is
\emph{global}, sweeping the entire tilt as an unidentified sensitivity parameter
because no target-side quantity constrains it; ours is \emph{anchored}, because the
bridge is observed in both populations and therefore \emph{identifies} the
detectable part of the tilt through a moment equation, leaving only the
bridge-orthogonal residual as sensitivity. Where they must report how large a
violation would overturn a conclusion, we report a set whose center is fixed by the
data and whose width is the closed-form irreducible core; and where their tilt is
untestable in full, ours is partially refutable through bridge-internal
over-identification, subject to a qualification we make explicit: the check has power only
against disagreement \emph{among} bridges, none against drift common to all of
them, and \Cref{sec:sim-bridge} reports how much of the former it detects in the
designs studied there. When no bridge is available our framework reduces to a global
tilt sensitivity analysis of their type, which we recover as the bridge-blind
comparator of \Cref{sec:blind}. Surrogate-index and surrogate-efficiency
methods \citep{athey2019surrogate,prentice1989surrogate}
exploit outcomes observed in both samples for efficiency, but under a surrogacy or
mediation condition, namely that the auxiliary outcome fully carries the treatment
or population signal, which is strictly stronger, and equally untestable, than
$\kappa=0$; we require neither mediation nor completeness of the bridge for the
outcome, only that the bridge carry the detectable channel, and we quantify what
remains when it does not.
A separate line dispenses with the surrogacy condition altogether.
\citet{kallus2025surrogates} derive the efficiency gain available from abundant
surrogate observations under no assumption beyond unconfoundedness and overlap,
and develop estimators that realize it. The distinction from our setting is the
availability of the outcome rather than the strength of the assumption: there the
target outcome is observed, if only on few units, so the surrogates improve the
precision of a quantity that is already identified. Here the outcome is never
observed in the target population, so no amount of bridge information restores
identification and the question is what remains unknown rather than how precisely
the known can be estimated. Shadow-variable and proximal missing-data methods
\citep{miao2018identifying,tchetgen2024proximal,shao2016semiparametric,uehara2023semiparametric}
are our closest kin in machinery: our bridge condition is structurally proximal,
with the cohort indicator acting as the response indicator and the bridge as a
proxy for a latent drift, but we differ in three ways that the setting forces and
rewards. The bridge is observed in both populations, so the detectable drift is
fixed by a constructive moment equation rather than the solution of an ill-posed
integral equation; the residual is a single interpretable scalar with a closed-form
core rather than a completeness-driven point; and the drift is estimated, and
tested for internal consistency, from a quantity the target actually reports.
Throughout, we state the boundary of testability explicitly. When several bridges
are available, their mutual agreement on the detectable drift is a genuine,
data-driven check, but it tests consistency \emph{among bridges}, not the
extrapolation \emph{from bridges to the outcome}, which is carried entirely by the
sensitivity parameter. A reader who distrusts the extrapolation is asked instead to choose a bound on how far the outcome may drift beyond
what the bridges reveal, and to read the resulting set.

\section{Setup and Identification}
\label{sec:ident}

The pooled population is a mixture of two cohorts indexed by $S\in\{0,1\}$, with
$\pi_s=\Prob(S=s)$; $S=0$ is the \emph{source} and $S=1$ the \emph{target}. For
each unit we always observe $(S,X,Z)$, where $X\in\R^p$ are shared predictors and
$Z\in\R^q$ are \emph{bridge outcomes} present in both cohorts. The target outcome
$Y\in\R$ is observed if and only if $S=0$. Let $P_s(\cdot)=P(\cdot\mid S=s)$ and
write conditional densities $p_s(\cdot\mid x)$. Design weights are suppressed
until \Cref{sec:est}.

The estimand is the target-population mean
\begin{equation}
  \mu_1 \;=\; \E[\,Y\mid S=1\,].
\end{equation}
The observed law identifies exactly three objects: the full source law
$p_0(y,z\mid x)$; the target bridge law $p_1(z\mid x)$; and the covariate laws
$P_0(X),P_1(X)$. It does \emph{not} identify $p_1(y\mid x)$ or $p_1(y,z\mid x)$,
since $Y$ is missing whenever $S=1$. All identifying leverage must therefore link
$p_1(y\mid x)$ to the observed objects.

\subsection{The co-drift tilt}
\label{sec:tilt}

We link source and target through a scalar-indexed exponential tilt. Fix known,
standardized loadings $t(\cdot)$ on $Y$ and $b(\cdot)$ on $Z$; the default is
$t(y)=(y-m^0_Y)/\sigma^0_Y$ and $b_k(z)=(z_k-m^0_{Z,k})/\sigma^0_{Z,k}$, with
$b(z)=\sum_{k}b_k(z_k)$ inside the tilt exponent and $b=(b_1,\dots,b_q)$ when the
bridge coordinates are used as separate moments.

The constants $(m^0_Y,\sigma^0_Y)$ and $(m^0_{Z,k},\sigma^0_{Z,k})$ are
functionals of the \emph{source} law only; we take them to be the source
conditional means and standard deviations at the primary specification. They are
fixed once, before the sensitivity analysis is run, and held at those values
across every $\kappa$, every bootstrap replicate, every imputed data set and
every sensitivity block an analysis reports. Two consequences follow. First, the loadings fix the units in which
$\gamma$ and $\kappa$ are read, so a reported $\bar\kappa$ is meaningless without
them, and a mapping from a substantive judgement to a bound has to be stated
alongside any reported value. Second, because they never involve the target sample, the
choice affects the \emph{scale} of the sensitivity parameter but not what is
identified.

The standardization is not cosmetic. A single $\gamma$ multiplying $t(y)+b(z)$
is meaningful only if the two terms are commensurate, and the outcome and the
bridges are in general different instruments on different scales; in our
application they are $W$ scores from three separate assessments. The co-drift
assumption is therefore expressed on the standardized scale, while the answer is
reported on the outcome scale, and $\sigma^0_Y$ is the conversion between them.
It reappears, necessarily, in the half-width of \Cref{prop:general}(ii) and in
the efficiency coefficient of \Cref{thm:eif}; a reader tracking dimensions
through the paper should expect exactly one factor of $\sigma^0_Y$ wherever a
statement crosses from the tilt scale to the outcome scale.

For any weight $\omega(y,z)$, write the tilted conditional expectation
\begin{equation}
  \Ptilt{\omega}[h\mid x]
  \;=\;
  \frac{E_{P_0}\!\big[h(Y,Z)\,e^{\omega(Y,Z)}\mid x\big]}
       {E_{P_0}\!\big[e^{\omega(Y,Z)}\mid x\big]}.
\end{equation}

\paragraph{Structural model (co-drift with residual).}
\begin{equation}
  \frac{dP_1(y,z\mid x)}{dP_0(y,z\mid x)}
  \;=\;
  \frac{\exp\!\big\{\,
        \gamma\,[\,t(y)+b(z)\,]
        \;+\;
        \kappa\, s(y,z,x)\,\big\}}
       {C(\gamma,\kappa;x)},
  \qquad
  C(\gamma,\kappa;x)=E_{P_0}\!\big[e^{\gamma[t(Y)+b(Z)]+\kappa s}\mid x\big].
  \label{eq:tilt}
\end{equation}
Here $\gamma\in\R$ is the \emph{shared drift} carried jointly by $Y$ and $Z$;
$\kappa\in\R$ is the \emph{residual drift} in $Y$ not reflected in the bridge.
Model~\eqref{eq:tilt} states that $E_{P_1}[h\mid x]=\Ptilt{\omega}[h\mid x]$ with
$\omega=\gamma[t+b]+\kappa s$. Writing the drift as a selection on the pooled
sample, \eqref{eq:tilt} is equivalent to
$\operatorname{logit}P(S=0\mid X,Y,Z)=h(X)-\gamma[t(Y)+b(Z)]-\kappa s(Y,Z,X)$:
the cohort indicator is a response indicator, and outcome-model drift is
nonignorable selection
\citep{kim2011semiparametric,uehara2023semiparametric}. There is no
exclusion-restricted instrument here; its role is played by the bridge,
which, unlike an instrument, is observed in both cohorts.

What $\gamma$ measures follows from the shared-coefficient structure. It is
the magnitude of the drift along the chosen tilt direction, and nothing more: any mechanism that moves the outcome
and the bridges together enters through it, and the model neither separates such
mechanisms nor needs to, since \Cref{ass:codrift} asks only that the outcome's
drift lie in the subspace the bridges detect. With two populations this is not a
limitation that better estimation can remove. A common trend that happens to
move both channels and a single mechanism that moves both are observationally
identical: no target-side moment distinguishes them, because the only
target-side information about the drift is the bridge shift itself, and both
produce the same shift. Separating them requires a third population, or a
comparison in which one channel is known a priori not to move.
Consequently $\hat\gamma$ should be read as the detectable drift and not as the
effect of any named cause.

The sensitivity analysis that follows sweeps a parameter \emph{within} this
exponential family. A cohort change
that alters the residual scale of $Y\mid X$, the dependence between $Y$ and $Z$
given $X$, or the shape of the conditional law at fixed mean and variance need not
lie in the span of $t+b$ and $s$, and then $|\kappa|\le\bar\kappa$ can hold in the
fitted model while $\mu_1$ lies outside the reported set.
\Cref{sec:sim-offmodel} generates targets of each kind and reports what the set
does.

\paragraph{Defining the residual direction.}
Rather than assume that the residual is undetectable by the bridge, we
\emph{define} it so. Project the natural $Y$-drift direction $t(Y)$ onto the
closed span $\Hz(x)$ of square-integrable functions of $(Z,X)$ \emph{under the
identified law}; the $\gamma^\star$-tilted source law
$dP^{\gamma^\star}\propto e^{\gamma^\star[t+b]}dP_0$, around which the
sensitivity direction is perturbed:
\begin{equation}
  t(Y)
  \;=\;
  \underbrace{\Pi^{\gamma^\star}\big[t(Y)\mid Z,X\big]}_{\text{bridge-detectable}}
  \;+\;
  \underbrace{s}_{\text{bridge-orthogonal}},
  \qquad
  s \;=\; t(Y)-\Pi^{\gamma^\star}\big[t(Y)\mid Z,X\big].
  \label{eq:proj}
\end{equation}
By construction $E^{\gamma^\star}[s\,\varphi(Z)\mid X]=0$ for every
$\varphi\in L_2$. Thus $\kappa=0$ means that the entire source-to-target drift in
$Y$ lies in the subspace the bridge can see, while $\kappa\neq0$ is drift
\emph{orthogonal} to the bridge. This turns the purity condition of the
shadow-variable literature \citep{miao2018identifying}, an assumption there, into
a geometric identity here.

Two features of \eqref{eq:proj} are used repeatedly below. First, although we write $s$ without arguments, it is a function of
$(Y,Z,X)$: the projection removes a function of the bridge and the covariates, so
the residual depends on them even though it is orthogonal to them. In the
Gaussian instance, $s=\{(y-m_Y(x))-\lambda^{\!\top}(z-m_Z(x))\}/\sigma^0_Y$ with
$\lambda=\Sigma_{ZZ}^{-1}\Sigma_{ZY}$, and $s$ is free of $z$ only in the
degenerate case $\lambda=0$ that \Cref{ass:relevance} excludes. Consequently the
$\kappa$-tilt moves the joint law of $(Y,Z)$, not the conditional law of $Y$
alone; that it nonetheless leaves the bridge moment unmoved at first order is the
content of \Cref{thm:set}(ii), and is a conclusion rather than a restatement of
the definition. Second, we apply no further normalization: $s$ is the projection
residual on the $t$ scale, so $\Var^{\gamma^\star}(s\mid x)$ is not one, and the
scale of $\kappa$ is fixed by the loading constant $\sigma^0_Y$ alone. An
alternative convention rescales $s$ to unit tilted variance; the two differ by a
deterministic rescaling of $\kappa$ and trace the same identified set
(\Cref{rem:convention}). We prefer the present one because it keeps the tilt
direction free of an estimated normalizer, which matters for the influence
function of \Cref{thm:eif} and for stability across resamples.
Neither convention makes $\bar\kappa$ interpretable on its own: the displacement
of $\mu_1$ per unit of $\kappa$ is $E_{P_1}[\Cov^{\gamma^\star}(Y,s\mid X)]$, not
$\sigma^0_Y$, so a bound must always be elicited through an explicit mapping from
a substantive judgement.

In the Gaussian instantiation, \emph{or} as $\gamma^\star\!\to\!0$, the tilted and
untilted projections coincide (Web Appendix~A, \Cref{rem:projlaw}), so the reader
may picture the ordinary $P_0$-projection without loss.
Outside that case the two projections differ, and the difference is not
merely formal: \Cref{sec:sim-binary} reports its magnitude in an exactly
computable binary instance.

\subsection{Assumptions}

\begin{assumption}[Co-drift]
\label{ass:codrift}
Model~\eqref{eq:tilt} holds. The substantive content is $\kappa=0$ (pure
co-drift); $\kappa\neq0$ is entertained only as bounded sensitivity,
$|\kappa|\le\bar\kappa$.
\end{assumption}

\begin{assumption}[Bridge relevance]
\label{ass:relevance}
There is $\varepsilon>0$ such that, for $P_1(X)$-almost every $x$ and every
$(\gamma,\kappa)\in\Gamma\times[-\bar\kappa,\bar\kappa]$,
\[
  \Cov^{\,\gamma[t+b]+\kappa s}\!\big(b_k(Z),\,t(Y)+b(Z)\mid x\big)\;\ge\;\varepsilon
  \qquad\text{for each }k=1,\dots,q,
\]
where $\Gamma$ is a compact interval containing $\gamma^\star$ and the roots
$\gamma(\kappa;x)$ of \eqref{eq:bridgematch-kappa}.
\end{assumption}

Three aspects of the statement require comment. First, the bound is uniform in
$x$ rather than holding on a set of positive measure: the argument of \Cref{thm:point}
establishes uniqueness of the root \emph{at each} $x$, and a condition holding
only somewhere does not deliver that. Second, it is required for each bridge coordinate
separately, since each supplies its own matching equation. Third, it is required
over the whole sensitivity range and not only at $\kappa=0$: \Cref{thm:set} solves
\eqref{eq:bridgematch-kappa} at every $\kappa$ in $[-\bar\kappa,\bar\kappa]$, so
the monotonicity that gives a unique root must be available there too. Under the
Gaussian instance the covariance is free of $\kappa$ and the requirement reduces
to the usual one; in general it is a restriction on how large $\bar\kappa$ may be
taken, and an analysis should report the smallest fitted value of
$\Cov^{\gamma^\star}(b_k,t+b\mid x)$ as a diagnostic.
Because the bridge equation is linear in $\gamma$ under the Gaussian working
model, monotonicity there is automatic and the assumption does no work; the
binary study of \Cref{sec:sim-binary}, where the equation is nonlinear, is the
design in which it does.

\begin{assumption}[Overlap and integrability]
\label{ass:overlap}
$0<\Prob(S=1\mid X)<1$ $P$-almost surely, and there is an open set
$\mathcal N\supset\Gamma\times[-\bar\kappa,\bar\kappa]$ with
\[
  E_{P_0}\!\Big[\exp\big\{\gamma[t(Y)+b(Z)]+\kappa\,s(Y,Z,X)\big\}
      \big(1+t(Y)^2+b(Z)^2+s^2\big)\;\Big|\;x\Big]<\infty
\]
for all $(\gamma,\kappa)\in\mathcal N$ and $P_1(X)$-almost every $x$.
\end{assumption}

Relative to a bare integrability requirement this restores $b(Z)$ inside the
exponent, ranges $\kappa$ over the whole bound rather than a neighbourhood of a
``truth'' it does not have, and asks for second moments, which is what the
dominated-convergence step of the tilt-derivative lemma (Web Appendix~A,
\Cref{lem:tiltderiv}) actually uses.

\subsection{Point identification under co-drift}

\begin{theorem}[Point identification under co-drift]
\label{thm:point}
Under \Cref{ass:codrift} with $\kappa=0$, and
\Cref{ass:relevance,ass:overlap}, the shared drift $\gamma$ solves the
bridge-matching equation
\begin{equation}
  E_{P_1}\!\big[b(Z)\mid x\big]
  \;=\;
  \Ptilt{\gamma}\!\big[b(Z)\mid x\big]
  \;=\;
  \frac{E_{P_0}\!\big[b(Z)\,e^{\gamma[t(Y)+b(Z)]}\mid x\big]}
       {E_{P_0}\!\big[e^{\gamma[t(Y)+b(Z)]}\mid x\big]},
  \label{eq:bridgematch}
\end{equation}
and this root $\gamma^\star$ is unique. Consequently the target mean is
point-identified:
\begin{equation}
  \mu_1
  \;=\;
  \E\!\left[\;
    \frac{E_{P_0}\!\big[Y\,e^{\gamma^\star[t(Y)+b(Z)]}\mid X\big]}
         {E_{P_0}\!\big[e^{\gamma^\star[t(Y)+b(Z)]}\mid X\big]}
    \;\middle|\; S=1\right].
  \label{eq:estimand}
\end{equation}
When \eqref{eq:bridgematch} is imposed across covariate strata and across bridge
coordinates $b=(b_1,\dots,b_q)$, $\gamma^\star$ is over-identified. The resulting
restriction is a test of \emph{mutual consistency among the bridges}; it does not,
and cannot, test the extrapolation from the bridges to the unobserved outcome,
which is governed by \Cref{ass:codrift} and carried by the sensitivity parameter
of \Cref{thm:set}.
The two restrictions are different and we test them separately. Agreement
across bridge coordinates fails when one bridge moves differently from the
others; agreement across covariate strata fails when the drift intensity varies
with $X$, so that no single detectable channel summarises the cohort transition.
A statistic that pooled them would reject without saying which had happened.
\Cref{sec:inference} gives both.
\end{theorem}

\begin{proof}[Proof sketch; full proof in Web Appendix~A]
The left side of \eqref{eq:bridgematch} is identified from target data; the right
side is, for each fixed $\gamma$, a functional of the source law. Differentiating,
$\partial_\gamma \Ptilt{\gamma}[b(Z)\mid x]
   =\Cov^{\gamma}\!\big(b(Z),\,t(Y)+b(Z)\mid x\big)\ge\varepsilon>0$ by
\Cref{ass:relevance}; hence $\gamma\mapsto\Ptilt{\gamma}[b(Z)\mid x]$ is strictly
increasing and \eqref{eq:bridgematch} has at most one root, with existence from
continuity and \Cref{ass:overlap}. For \eqref{eq:estimand}, $\kappa=0$ makes
$P_1(\cdot\mid x)$ the $\gamma^\star$-tilt of $P_0(\cdot\mid x)$, so
$E_{P_1}[Y\mid x]=\Ptilt{\gamma^\star}[Y\mid x]$, identified from source; averaging
over $P_1(X)$ gives $\mu_1$.
\end{proof}

\begin{proposition}[Distribution-general drift and core]
\label{prop:general}
The closed forms below hold for \emph{any} source law under the natural-parameter
tilt, expressed through its conditional tilt (co)variances at the truth
$\gamma^\star$. Writing all covariances under the identified law
$\Ptilt{\gamma^\star}$ and conditioning on $X=x$:
\begin{enumerate}[label=(\roman*)]
  \item \emph{(Identified drift.)} The bridge-matching equation of
  \Cref{thm:point} pins the outcome shift to
  \begin{equation}
    E_{P_1}[t(Y)\mid x]-E_{P_0}[t(Y)\mid x]
    \;=\;
    \big(E_{P_1}[b(Z)\mid x]-E_{P_0}[b(Z)\mid x]\big)\;
    c_t(x),\quad
    c_t(x)=\frac{\Cov^{\gamma^\star}\!\big(t,\,t+b\mid x\big)}
              {\Cov^{\gamma^\star}\!\big(b,\,t+b\mid x\big)} .
    \label{eq:c-general}
  \end{equation}
  Equation~\eqref{eq:c-general} is a first-order statement: it is the ratio of
  the two directional derivatives at $\gamma^\star$, so it holds exactly when the
  tilted moments are linear in $\gamma$, as in the Gaussian instance, and to
  $O(\gamma^{\star2})$ otherwise. The exact relation replaces the ratio of
  derivatives by the ratio of integrals
  $\int_0^{\gamma^\star}\Cov^{\gamma}(t,t+b\mid x)\,d\gamma$ and
  $\int_0^{\gamma^\star}\Cov^{\gamma}(b,t+b\mid x)\,d\gamma$. We use the
  derivative form throughout because it is the quantity that appears in the
  influence function of \Cref{thm:eif}.
  Carried to the outcome scale, $\sigma^0_Y\,c_t(x)$ is the local rate at which a
  unit of bridge drift translates into outcome drift. It should not be confused
  with the efficiency coefficient $c$ of \Cref{thm:eif}, which is a different
  object. The ratio $c_t$ carries a denominator, the Stage-1 sensitivity of the
  bridge moment, while $c$ does not, that denominator already sitting inside the
  Stage-1 influence function.
  \item \emph{(Irreducible core.)} With $s$ the tilted-law projection residual
  \eqref{eq:proj}, so that $\Cov^{\gamma^\star}(b,s\mid x)=0$ and
  $\Cov^{\gamma^\star}(t,s\mid x)=\Var^{\gamma^\star}(s\mid x)$, the half-width is
  \begin{equation}
    \tfrac12 W(\bar\kappa)
    \;=\;
    \bar\kappa\,E_{P_1}\!\big[\Cov^{\gamma^\star}\!\big(Y,s\mid X\big)\big]
    +O(\bar\kappa^2)
    \;=\;
    \bar\kappa\,\sigma^0_Y\,E_{P_1}\!\big[\Var^{\gamma^\star}(s\mid X)\big]
    +O(\bar\kappa^2),
    \label{eq:core-general}
  \end{equation}
  the target-averaged tilted residual variance, carried to the outcome scale by
  the loading constant $\sigma^0_Y$; the second equality is
  $\Cov^{\gamma^\star}(Y,s)=\sigma^0_Y\Cov^{\gamma^\star}(t,s)
  =\sigma^0_Y\Var^{\gamma^\star}(s)$.
\end{enumerate}
The half-width does \emph{not} depend on how precisely the bridge identifies
$\gamma^\star$, nor on the sample size. It does depend on how much of the outcome
the bridge and the covariates leave unexplained: a bridge that explains more of
$Y$ shrinks $\Var^{\gamma^\star}(s\mid x)$ and with it the core. In the limit of a
bridge that determines the outcome given the covariates, $s\to0$ and $\mu_1$ is
point identified, as it must be, since $Z$ is observed in the target. The
Gaussian and binary specializations below are two instances; the identity is not
tied to either.
$W(\bar\kappa)$ is a working-model functional. Both of its inputs are:
$\sigma^0_Y$ is a loading constant fixed by convention, and
$\Var^{\gamma^\star}(s\mid x)$ depends on the assumed dependence between $Y$ and
$Z$ given $X$, which is the one modelling choice \Cref{lem:dcdr} does not protect
(\Cref{rem:stage1-scope}). That dependence enters the analysis twice. It enters
the tilted bridge moment, so a misspecified dependence moves $\hat\gamma$ and with
it the benchmark; and it enters $\Var^{\gamma^\star}(s\mid x)$, so it moves the
half-width. Neither is guaranteed when it is wrong.
\end{proposition}

\begin{proof}[Proof]
Both parts are the exponential-family cumulant identity
$\partial_\gamma E^{\gamma}[h]=\Cov^{\gamma}(h,t+b)$ applied to $h\in\{t,b,s\}$,
combined with the implicit-function derivative of the bridge equation; the
vanishing of the $\gamma$-path at $\kappa=0$ is \Cref{thm:set}(ii), to which we
forward-reference. See Web Appendix~A.
\end{proof}

\begin{corollary}[Gaussian instance]
\label{cor:gaussian}
Let the source working model be $Y\mid X\sim N(m_0(X),\sigma_Y^2)$ with the bridge
$Z\mid X\sim N(m_{Z0}(X),\sigma_Z^2)$ and residual correlation $r$, and take the
loading constants at their conditional values, $\sigma^0_Y=\sigma_Y$ and
$\sigma^0_Z=\sigma_Z$. Write $\Delta_Z=m_{Z1}-m_{Z0}$ for the observed bridge
drift. Then $\Var^{\gamma^\star}(s\mid X)=1-r^2$, and \Cref{prop:general} gives
\begin{equation}
  \gamma^\star=\frac{\Delta_Z}{\sigma_Z\,(1+r)},
  \qquad
  E_{P_1}[Y\mid x]-m_0(x)=\Delta_Z\,\frac{\sigma_Y}{\sigma_Z},
  \qquad
  W(\bar\kappa)=2\bar\kappa\,\sigma_Y\,(1-r^2).
  \label{eq:gauss-instance}
\end{equation}
Two consequences follow. First, the identified outcome drift is the observed
bridge drift rescaled by the ratio of standard deviations $\sigma_Y/\sigma_Z$: the
factor
$(1+r)$ cancels between the bridge-matching step and the outcome-shift step, so
the transfer coefficient is free of the residual correlation even though
$\gamma^\star$ itself is not. This gives a derivation for the weighted bridge-drift calibration used in applied
practice, replacing its weights by the derived ratio $\sigma_Y/\sigma_Z$.
Second, the core width carries the factor $(1-r^2)$: a bridge more strongly
related to the outcome yields a narrower core, and $r^2\to1$ returns point
identification. With $q$ bridges, $r^2$ is replaced throughout by the multiple
$R^2_{Y\mid Z,X}$. A further Gaussian simplification, used in
\Cref{sec:est-eq}, is that the bridge equation is \emph{linear} in $\gamma$, so
Stage~1 is closed form with per-bridge roots
$\hat\gamma_k=\hat\Delta_{Z,k}/\{\sigma^0_{Z,k}(1+r_k)\}$.
Both simplifications are properties of the Gaussian working model rather
than of the method: the cancellation of $(1+r)$ in the transfer coefficient and
the linearity of the bridge equation are what make this instance transparent, and
neither survives to a general exponential family.
\end{corollary}

\begin{corollary}[Binary/logistic instance]
\label{cor:binary}
If $Y\mid X$ is binary with $t(Y)=Y$ and the bridge is (possibly binary) with
$b(Z)=Z$, the same identity holds with the Gaussian variances replaced by their
tilted-Bernoulli counterparts: $c_t(x)=\Cov^{\gamma^\star}(Y,Y{+}Z\mid x)/
\Cov^{\gamma^\star}(Z,Y{+}Z\mid x)$ and core driver
$\Var^{\gamma^\star}(s\mid x)=E^{\gamma^\star}[(Y-E^{\gamma^\star}[Y\mid Z,x])^2\mid x]$,
the tilted residual variance of a Bernoulli outcome. Here $t(Y)=Y$ carries no
standardizing constant, so $\sigma^0_Y=1$ and the two forms in
\eqref{eq:core-general} coincide: the half-width is
$\bar\kappa\,E_{P_1}[\Var^{\gamma^\star}(s\mid X)]$ directly, and $\kappa$ is read
as a shift on the logit scale. This differs from the Gaussian instance, where the
loading divides by $\sigma_Y$; the two instances therefore fix $\kappa$ in
different units and a value of $\bar\kappa$ is not transferable between them. No
normal approximation is used; the closed form is exact for the working
exponential family. This is the form relevant to binary attainment outcomes and
to the differential-missingness layer of the application.
This instance is also where several statements that are vacuous under the
Gaussian working model acquire content. With a discrete covariate the support of
$(Y,Z)$ is finite, so the tilted projection \eqref{eq:proj}, the root of
\eqref{eq:bridgematch-kappa} at every $\kappa$, and the identified set itself are
exact finite sums. \Cref{sec:sim-binary} uses this to measure the gap between the
tilted and untilted projections, the derivative $\partial_\kappa\gamma$ at
$\kappa=0$ in a family where the bridge equation is not linear in $\gamma$, and
the second-order coefficient that \Cref{rem:set-scope} declines to characterize.
\end{corollary}

\subsection{Identified set under imperfect co-drift}

\begin{theorem}[Identified set and irreducible core]
\label{thm:set}
Relax to $|\kappa|\le\bar\kappa$ and keep \Cref{ass:relevance,ass:overlap}, with
$s$ the projection residual~\eqref{eq:proj}. For each $\kappa$, let
$\gamma(\kappa;x)$ solve the residual-augmented bridge equation
\begin{equation}
  E_{P_1}\!\big[b(Z)\mid x\big]
  \;=\;
  \Ptilt{\,\gamma[t+b]+\kappa s}\!\big[b(Z)\mid x\big],
  \label{eq:bridgematch-kappa}
\end{equation}
and define
$\mu_1(\kappa)=\E\big[\Ptilt{\,\gamma(\kappa)[t+b]+\kappa s}[Y\mid X]\mid S=1\big]$.
Then:
\begin{enumerate}[label=(\roman*)]
  \item Within the scalar residual model, that is, taking the residual
  perturbation to lie along the single direction $s$, the identified set for
  $\mu_1$ is $\{\mu_1(\kappa):|\kappa|\le\bar\kappa\}$. The map
  $\kappa\mapsto\mu_1(\kappa)$ is continuous, so the set is an interval; it
  equals $[\mu_1(-\bar\kappa),\mu_1(\bar\kappa)]$ whenever
  $\partial_\kappa\mu_1$ keeps its sign on $[-\bar\kappa,\bar\kappa]$, which
  holds automatically in the Gaussian instance, where $\mu_1$ is affine in
  $\kappa$, and more generally for $\bar\kappa$ small enough that the sign
  established at $\kappa=0$ in part (ii) is not reversed. Otherwise the set is
  $[\min_\kappa\mu_1(\kappa),\max_\kappa\mu_1(\kappa)]$, obtained from the same
  sweep. Under fully multi-directional residual drift the displayed interval is
  a first-order inner approximation; see \Cref{rem:scalar-scope}.
  \item By the projection property~\eqref{eq:proj},
  $\partial_\kappa\gamma(\kappa;x)\big|_{\kappa=0}=0$: the bridge-identified
  $\gamma^\star$ is \emph{stationary at the benchmark}. The benchmark
  $\mu_1(0)$ is thus the \Cref{thm:point} point estimate, and the interval is
  anchored at it with no first-order drift as $\kappa$ is switched on; the
  midpoint of $[\mu_1(-\bar\kappa),\mu_1(\bar\kappa)]$ coincides with $\mu_1(0)$ up
  to $O(\bar\kappa^2)$ (exactly, in the Gaussian instance, where $\mu_1(\kappa)$ is
  affine). The half-width obeys
  \begin{equation}
    \tfrac12 W(\bar\kappa)
    \;=\;
    \bar\kappa\;
    E_{P_1}\!\big[\Cov^{\gamma^\star}\!\big(Y,\,s\mid X\big)\big]
    \;+\;O(\bar\kappa^2),
    \label{eq:core}
  \end{equation}
  the \emph{irreducible core}. By \Cref{prop:general} it equals
  $\bar\kappa\,\sigma^0_Y\,E_{P_1}[\Var^{\gamma^\star}(s\mid X)]$ for any working
  family. It is not a sampling-error quantity: no sample size narrows it, and
  neither does a bridge that pins $\gamma^\star$ more precisely. It is narrowed
  only by a smaller $\bar\kappa$, or by a bridge that explains more of the
  outcome. The Gaussian instance gives
  $W(\bar\kappa)=2\bar\kappa\,\sigma_Y(1-R^2_{Y\mid Z,X})$ and the binary
  instance the tilted-Bernoulli residual variance
  (\Cref{cor:gaussian,cor:binary}).
\end{enumerate}
\end{theorem}

\begin{proof}[Proof sketch; full proof in Web Appendix~A]
(i) For each $\kappa$, \eqref{eq:bridgematch-kappa} has a unique root by the
monotonicity argument of \Cref{thm:point}, which is available over the whole
sweep by \Cref{ass:relevance}; existence follows because $e^{\kappa s}>0$ almost
surely and is integrable by \Cref{ass:overlap}, so the tilted measure is
\emph{equivalent} to $P_0(\cdot\mid x)$ and the range argument is unchanged.
Ranging $\kappa$ traces the set.
(ii) Differentiating \eqref{eq:bridgematch-kappa} implicitly at $\kappa=0$,
$\partial_\kappa\gamma
   =-\,\Cov^{\gamma^\star}(b(Z),s\mid x)/\partial_\gamma(\mathrm{RHS})$, and the
numerator vanishes because $s$ is orthogonal to $\Hz(x)$ under the identified
(tilted) law by \eqref{eq:proj}. Thus
$\gamma(\kappa)=\gamma^\star+O(\kappa^2)$ and
$\partial_\kappa\mu_1=E_{P_1}[\Cov^{\gamma^\star}(Y,s\mid X)]$, giving
\eqref{eq:core}.
\end{proof}

\begin{remark}[Role of the projection in defining the residual]
\label{rem:whyproject}
The vanishing numerator in part~(ii) has the following consequence. Had the
residual direction been chosen on any other grounds, a substantively motivated
contrast say, or the raw $t(Y)$ itself, one would have
$\Cov^{\gamma^\star}(b,s\mid x)\neq0$, hence
$\gamma(\kappa)=\gamma^\star+\kappa\cdot(\text{nonzero})+O(\kappa^2)$: the
detectable drift, and with it the \emph{center} of the reported set, would move at
first order in the sensitivity parameter, and the half-width would acquire an
additional $\gamma$-path term
$E_{P_1}[\Cov^{\gamma^\star}(Y,t{+}b\mid X)]\,\partial_\kappa\gamma$ on top of
\eqref{eq:core}. Defining $s$ by projection is what produces a set that opens
around a fixed center, and is why the center can be reported as an estimate while
the width is reported as a judgement.
This is a property of the parameterization, not additional identifying
information. The projection buys a convenient report, in which an estimate and a
judgement are cleanly separated, and it buys nothing about the target law that
was not already implied by \eqref{eq:tilt}.
\end{remark}

\begin{remark}[Scope of \Cref{thm:set}(ii)]
\label{rem:set-scope}
Part~(ii) establishes $\partial_\kappa\gamma|_{\kappa=0}=0$, a statement at a
point. Two cases must be separated. In the Gaussian instance a linear tilt shifts
means without changing covariances, so $\Cov^{\omega}(b,s\mid x)=0$ at
\emph{every} $\kappa$, not only at zero; there $\gamma(\kappa)\equiv\gamma^\star$
exactly, the center is invariant over the whole sweep, and the half-width
\eqref{eq:core} carries no remainder (Web Appendix~A,
\Cref{rem:exactness}). Outside such cases the claim is first order:
$\gamma(\kappa)=\gamma^\star+O(\kappa^2)$, and we do not characterize the
second-order coefficient
in general, so the paper offers no general criterion for whether a
particular $\bar\kappa$ is small enough. In one instance it can be measured
rather than bounded: \Cref{sec:sim-binary} computes the coefficient, and the
movement of the centre of the set that follows from it, on population tables
for a binary outcome with a discrete covariate, where every quantity is an exact
finite sum. There the centre moves by a few percent of the width of the set at
the residual bounds one would entertain in practice, and the closed form of
\Cref{prop:general} is accurate to within one percent. This is the same limitation
\Cref{rem:scalar-scope} concedes for \Cref{prop:scalar}, and it should be read the
same way: the reported interval is exact under the Gaussian working model and a
first-order inner approximation otherwise. We avoid describing the center as
``locally stable,'' which suggests an interval statement we have not proved;
\emph{stationary at the benchmark} is what is established.
\end{remark}

\begin{remark}[Interpretation of a narrower set]
\label{rem:narrower}
Three distinct quantities are routinely described as narrowing, and they should
be kept apart. A bridge that pins $\gamma^\star$ more
sharply, through more bridges or a stronger bridge-drift covariance, gives a
tighter \emph{center}. A larger study gives a tighter \emph{confidence interval} around
each endpoint. Neither touches the \emph{core}: by \eqref{eq:core} the core
responds only to $\bar\kappa$ and to the residual variance
$\Var^{\gamma^\star}(s\mid X)$, that is, to how much of the outcome the bridge and
covariates leave unexplained. A bridge can therefore narrow the core, but only by
explaining more of $Y$, never by identifying $\gamma^\star$ more precisely and
never through the sample size.
\end{remark}

The three results divide the problem as follows. Under $\kappa=0$,
\Cref{thm:point} fixes the drift from the data. The over-identification check
interrogates consistency among the bridges. The $\bar\kappa$-sweep of
\Cref{thm:set} reports the extrapolation the data cannot test. The framework
therefore imposes one assumption, tests what can be tested, and confines the
remainder to a single dimension.

\begin{remark}[Selection of the bridge set and the consequence of a rejection]
\label{rem:bridge-choice}
A construct should not appear on both sides of the analysis. The covariance
in \Cref{ass:relevance} is taken after conditioning on $X$, so what identifies
$\gamma$ is the variation in $Z$ that survives that conditioning, and an earlier
administration of a bridge instrument leaves little of it. Such a measurement is
a bridge or a covariate, not both. The failure is quiet: the drift estimate goes
to zero with a small standard error rather than an unstable one, so nothing in
the fit announces it. The quantity to inspect is $R^2_{Y\mid Z,X}$, the share of
the outcome the bridges explain \emph{given} the covariates.

Two practical points follow. First, the bridge that anchors the analysis must be
fixed on substantive grounds, not selected from the data. Because $Y$ is unobserved
at $S=1$, no observed-data criterion can identify which bridge best reflects the
\emph{outcome's} drift; a bridge chosen after inspecting per-bridge drift estimates
can therefore be chosen, deliberately or not, to move $\mu_1$, so post hoc
selection converts an untestable assumption into an apparently data-driven
decision without recording that it has done so.
We recommend pre-specifying a primary bridge on construct grounds: closest in
meaning to the outcome, and using any remaining bridges for the
over-identification check rather than as competing candidates. Second, when that
check rejects, the appropriate response is not to search for a bridge that restores
agreement, nor to escalate to a richer drift model: with $Y$ unobserved, even a
vector-loading generalization leaves the outcome-direction drift unidentified
(\Cref{rem:vector}). A rejection is instead evidence that the cohort transition is
not summarized by a single detectable channel, and its natural consequence is a
wider region: raising $\bar\kappa$, guided for instance by the observed
spread among the per-bridge drift estimates, rather than a narrower point.
\end{remark}

\begin{remark}[Vector loadings and multi-directional drift]
\label{rem:vector}
The scalar $\gamma$ is not a claim that drift is one-dimensional; it is the
coefficient along the chosen tilt direction $t(\cdot)$. Taking vector loadings
$t=(t_1,\dots,t_r)$ and $b=(b_1,\dots,b_r)$ makes $\gamma\in\R^r$ and turns
\eqref{eq:bridgematch} into $r$ simultaneous bridge-matching equations, identified
whenever the bridges span the drift directions ($q\ge r$). The scalar case is the
deliberate low-dimensional default; the identified-set geometry of \Cref{thm:set}
extends coordinatewise. We note that a fully nonparametric vector drift
$\Delta\in\R^p$ acting on all covariates is \emph{not} identified when $Y$ is
unobserved in the target: any procedure that appears to estimate it has implicitly
reduced its dimension. In particular, allowing a separate loading on the
outcome direction: writing the exponent with an outcome coefficient $\gamma_0
t(y)$ distinct from the bridge coefficients, leaves $\gamma_0$ unidentified, since
no target-side moment involves $Y$; the shared-coefficient structure is precisely
what transfers the bridge-identified drift to the outcome, and some such transfer
assumption is unavoidable. Our scalar (or low-rank) drift makes that reduction
explicit and testable rather than hidden. The parallel question for the
\emph{residual} sensitivity $\kappa$, why a single scalar suffices there, is
settled next.
\end{remark}

\begin{proposition}[Sufficiency of a scalar residual sensitivity]
\label{prop:scalar}
Let the residual drift be an arbitrary, possibly infinite-dimensional,
perturbation of the target law within the bridge-orthogonal subspace,
$dP_1^{\eta}(y,z\mid x)\propto dP_1(y,z\mid x)\exp\{\sum_j \eta_j s_j\}$.
Here $\{s_j\}$ is an orthonormal basis, under the identified law
$P^{\gamma^\star}(\cdot\mid x)$, of the orthogonal complement of $\Hz(x)$ within
$L^2$, and $\|\eta\|$ is the corresponding $\ell^2$ norm; the direction $s$ of
\eqref{eq:proj} lies in that complement and is a scalar multiple of one of its
elements (\Cref{rem:convention}). We assume $\sum_j\eta_j^2<\infty$ and that
$\mathcal N$ of \Cref{ass:overlap} contains a neighbourhood of $\eta=0$, so that
$\eta\mapsto\mu_1(\eta)$ is well defined and differentiable there. The
perturbation is applied to the joint law because $s$ is a function of $(Y,Z,X)$;
it leaves the target bridge law, and hence \eqref{eq:bridgematch}, unchanged at
first order, which is what makes it invisible to the data. Then:
\begin{enumerate}[label=(\roman*)]
  \item \emph{(First-order reduction is exact.)} The target mean depends on the
  full vector $\eta$ only through the one-dimensional projection onto
  $\hat g(x)=\big(\Cov^{\gamma^\star}(Y,s_j\mid x)\big)_j$:
  \begin{equation}
    \partial_{\eta}\,\mu_1\big|_{\eta=0}
    =E_{P_1}\!\big[\hat g(X)\big],
    \qquad
    \partial_{\eta}\,\mu_1\perp\{\eta:\langle\eta,\hat g\rangle=0\}.
    \label{eq:scalar-reduction}
  \end{equation}
  Any residual drift orthogonal to $\hat g$ leaves $\mu_1$ unchanged at first
  order; the scalar $\kappa=\langle\eta,\hat g/\|\hat g\|\rangle$ captures the
  entire first-order sensitivity of $\mu_1$ to functional residual drift.
  \item \emph{(The scalar direction is the worst case.)} Among all residual
  drifts of a given size, the one that moves $\mu_1$ most is the direction
  $\hat g$; sweeping $\kappa$ over $[-\bar\kappa,\bar\kappa]$ therefore traces the
  first-order extremes of the identified set.
\end{enumerate}
Moreover $\hat g$ is proportional to $s$: writing $t=\Pi^{\gamma^\star}[t\mid
Z,X]+s$ and using $\Cov^{\gamma^\star}(s,\varphi(Z))=0$, one has
$\hat g_j=\sigma^0_Y\Cov^{\gamma^\star}(s,s_j\mid x)$, so the worst-case direction
of part~(ii) is the direction swept in \Cref{thm:set}. The two results are
therefore about the same object.
\end{proposition}

\begin{proof}[Proof]
Because $\mu_1$ is a linear functional (the mean) of the outcome law, its Gateaux
derivative in direction $s_j$ is $\Cov^{\gamma^\star}(Y,s_j\mid x)$ averaged over
$P_1(X)$, by the exponential-family cumulant identity used in \Cref{prop:general}.
Collecting these over $j$ gives the gradient $E_{P_1}[\hat g(X)]$; directions with
$\langle\eta,\hat g\rangle=0$ have zero derivative, proving (i). Part (ii) is the
Cauchy--Schwarz extremal characterization of the gradient. \hfill$\square$
\end{proof}

\begin{remark}[The unit-variance convention]
\label{rem:convention}
\Cref{prop:scalar} is stated on an orthonormal basis, so the scalar it produces is
the coefficient on the unit-variance direction $\hat g/\|\hat g\|$. Our $s$ of
\eqref{eq:proj} is a constant multiple of that direction,
\[
  s \;=\; c_s\,\frac{\hat g}{\|\hat g\|},
  \qquad
  c_s=\Big(E_{P_1}\big[\Var^{\gamma^\star}(s\mid X)\big]\Big)^{1/2},
\]
so the two conventions are related by $\kappa_{\text{unit}}=c_s\,\kappa$ and trace
the same identified set once $\bar\kappa$ is mapped accordingly. Nothing in
\Cref{prop:scalar} depends on the choice: the worst-case direction, and the fact
that a scalar captures the entire first-order sensitivity, are properties of the
subspace, not of its parameterization. What the choice does fix is the units in
which $\bar\kappa$ is elicited, and therefore how a substantive judgement about
residual drift is translated into a bound, which is why a reported
$\bar\kappa$ has to be accompanied by the mapping that produced it. We adopt the $t$-scale convention
because it keeps the tilt direction free of an estimated normalizer: under the
unit-variance convention the direction itself depends on
$E_{P_1}[\Var^{\gamma^\star}(s\mid X)]$, which must be estimated, which adds a term
to the influence function of \Cref{thm:eif} and moves the endpoints across
bootstrap replicates and sensitivity blocks that are meant to be comparable.
\end{remark}

\begin{remark}[Scope of \Cref{prop:scalar}]
\label{rem:scalar-scope}
\Cref{prop:scalar} is a \emph{first-order} (local) statement, and this is the
sense in which our $\bar\kappa$ is a local residual bound, the standard reading of
a sensitivity parameter \citep{robins2000sensitivity}. At larger residual
magnitudes, several residual directions can act together and the worst case over
the full functional ball $\{\|\eta\|\le\bar\kappa\}$ can exceed the scalar
endpoint by an $O(\bar\kappa^2)$ amount.
The relation is stated precisely below, since \Cref{thm:set}(i) is now
conditioned on it. Write $\mu_1^{\mathrm{scal}}(\bar\kappa)$ for the endpoint
traced by the scalar sweep and $\mu_1^{\mathrm{fun}}(\bar\kappa)$ for the
supremum over $\{\|\eta\|\le\bar\kappa\}$. Then
$\mu_1^{\mathrm{scal}}\le\mu_1^{\mathrm{fun}}$, with equality to first order, and
the gap is $O(\bar\kappa^2)$ with a constant given by the curvature of
$\eta\mapsto\mu_1(\eta)$ at $\eta=0$. An analyst who wants a set that is valid
against multi-directional residual drift should therefore read the reported
interval as an inner bound and inflate $\bar\kappa$ accordingly; the inflation
needed is second order in $\bar\kappa$ and negligible at the magnitudes used in
practice, but it is not zero. Alternatively one may
sweep a low-rank residual $\kappa\in\R^d$ by the coordinatewise extension of
\Cref{thm:set}, at the cost of a wider and less interpretable set. We prefer the
scalar default precisely because it is the direction the target functional
actually feels, and because a single interpretable sensitivity dimension is what
makes the analysis interpretable to subject-matter collaborators.
\end{remark}

\section{Estimation and Semiparametric Efficiency}
\label{sec:est}

We now turn the identified quantities of \Cref{thm:point,thm:set} into estimators.
Survey design weights $d_i$ are reinstated. The building blocks are three
nuisances: the cohort propensity $e(x)=\Prob(S{=}1\mid x)$; a source working model
$f_0(y,z\mid x;\beta)=f(y,z\mid x,S{=}0;\beta)$, used with fractional imputation
\citep{kim2011fractional} to evaluate tilted conditional expectations; and the
drift $\gamma$, solved from the bridge-matching equation.

\subsection{Estimation}
\label{sec:est-eq}

\paragraph{Evaluating tilted moments.}
The Stage-1 moment requires $\Ptilt{\omega}[b(Z)\mid x]$, an expectation over the
\emph{joint} conditional law of $(Y,Z)$ given $X$: the tilt exponent contains
$t(Y)$, so $Z$ cannot be held fixed. We therefore draw pairs
$\{(y^{(m)},z^{(m)})\}_{m=1}^{M}$ and, for any $h$,
\begin{equation}
  \Ptilt{\omega}[h\mid x]
  \;\approx\;
  \frac{\sum_{m} w^{(m)}\,h(y^{(m)},z^{(m)})}{\sum_{m} w^{(m)}},
  \qquad
  w^{(m)} = \exp\!\big\{\omega(y^{(m)},z^{(m)},x)\big\}\,
            \frac{f_0(y^{(m)},z^{(m)}\mid x)}{g(y^{(m)},z^{(m)}\mid x)},
  \label{eq:fi}
\end{equation}
where $g$ is the proposal from which the pairs are drawn. The draws are a
device for evaluating the integral, not a completion of the data: they are
generated afresh at each evaluation and never enter the analysis file. Item
nonresponse in $X$ or $Z$ is a separate matter, handled outside the estimator by
multiple imputation. Taking $g=f_0$ recovers
the plain form, in which $w^{(m)}=e^{\omega}$. That choice degrades as the tilt
grows: the variance of $\omega$ increases with $q$ and with $|\gamma|$, the
effective sample size of the weights falls, and the self-normalized ratio
acquires an $O(M^{-1})$ bias that does not vanish with $n$.
In the design of \Cref{sec:sim} the effective sample size of the plain
weights is about a seventh of $M$ and the bias in $\hat\gamma$ falls only as
$M^{-1}$, from $0.15$ at $M=40$ to $0.011$ at $M=800$, against a drift of
$0.5$. We therefore take $g$
to be the working family tilted to the current $\omega$, for which the weights are
constant in the Gaussian and other natural-parameter instances, so that
\eqref{eq:fi} is exact up to Monte Carlo error and the effective sample
size is the full $M$; the residual bias in $\hat\gamma$ is then about $0.004$
and does not move with $M$. \Cref{sec:sim} compares the two
proposals and reports the bias each leaves at a given number of draws.
Misspecification of $f_0$ is
tolerated through the double-robustness structure of \Cref{sec:est-eq}; the
proposal affects only Monte Carlo accuracy, not the estimand. We report the
effective sample size $(\sum_m w^{(m)})^2/\sum_m (w^{(m)})^2$ as a diagnostic.

\paragraph{Stage 1: drift.} With $\hat N_s=\sum_{i}\mathbf 1\{S_i{=}s\}d_i$, solve
the sample bridge-matching equation
\begin{equation}
  U_\gamma(\gamma)
  =\frac{1}{\hat N_1}\sum_i \mathbf 1\{S_i{=}1\}\,d_i\,
     \big(b(Z_i)-\Ptilt{\gamma}[b(Z)\mid X_i]\big)=0.
  \label{eq:stage1}
\end{equation}
With $q>1$ bridges this is over-identified and combined by an efficient GMM
weight; the over-identification statistic is the consistency check of
\Cref{thm:point}, and is discussed in \Cref{sec:inference}. In the Gaussian
instance \eqref{eq:stage1} is linear in $\gamma$ and the per-bridge roots are
available in closed form (\Cref{cor:gaussian}), which we use as the starting
value of the general solve.

\paragraph{Stage 2: drift-augmented target mean.} Let
$\tilde m_1(x)=\Ptilt{\hat\gamma[t+b]+\kappa s}[Y\mid x]$. The estimator is
\begin{equation}
  \hat\mu_1^{DR}(\kappa)
  =\frac{1}{\hat N_1}\sum_i \mathbf 1\{S_i{=}1\}\,d_i\,\tilde m_1(X_i)
  \;+\;
  \frac{1}{\hat N_0}\sum_i \mathbf 1\{S_i{=}0\}\,d_i\;
     r_e(X_i)\;\rho_{\hat\gamma,\kappa}(Y_i,Z_i\mid X_i)\;
     \big(Y_i-\tilde m_1(X_i)\big),
  \label{eq:aipw}
\end{equation}
where
\[
  r_e(x)=\frac{e(x)}{1-e(x)}\cdot\frac{\pi_0}{\pi_1}
        =\frac{dP_1(x)}{dP_0(x)},
  \qquad
  \rho_{\gamma,\kappa}(y,z\mid x)
  =\frac{e^{\gamma[t(y)+b(z)]+\kappa s(y,z,x)}}{C(\gamma,\kappa;x)}
  =\frac{dP_1(y,z\mid x)}{dP_0(y,z\mid x)} .
\]
The first term imputes the tilted outcome regression at each target unit; the
second reweights source residuals by the covariate density ratio and the
conditional tilt ratio. Two points of arithmetic require care. First, the augmentation term must be normalized by $\hat N_0$: with $r_e$
carrying the factor $\pi_0/\pi_1$, dividing instead by $\hat N_1$ leaves a
residual factor $\pi_0/\pi_1$, which has no effect when the two cohorts are of equal size. Second, $\rho$ is the \emph{normalized} conditional density ratio;
using the raw tilt factor $e^{\omega}$ without dividing by $C(\gamma,\kappa;x)$
breaks the cancellation on which \Cref{lem:dcdr} rests. Setting
$\hat\gamma=0,\kappa=0$ gives $\rho\equiv1$ and recovers the standard
covariate-shift AIPW estimator; the tilt factor is what carries the drift.

\subsection{Efficient influence function}

\begin{theorem}[Efficient influence function at anchored drift]
\label{thm:eif}
Fix $\gamma=\gamma^\star$ and $\kappa$, and treat $s$ as known. Under
\Cref{ass:codrift,ass:relevance,ass:overlap} and standard regularity,
$\mu_1(\kappa)$ is pathwise differentiable with efficient influence function
\begin{equation}
  \varphi(O)
  =\underbrace{\tfrac{\mathbf 1\{S=1\}}{\pi_1}\big(\tilde m_1(X)-\mu_1\big)}
      _{\text{(I) target imputation}}
  +\underbrace{\tfrac{\mathbf 1\{S=0\}}{\pi_0}\,r_e(X)\,
      \rho_{\gamma^\star,\kappa}(Y,Z\mid X)\big(Y-\tilde m_1(X)\big)}
      _{\text{(II) tilted propensity correction}}
  +\underbrace{c\,\varphi_\gamma(O)}
      _{\text{(III) bridge-moment correction}},
\end{equation}
where $\varphi_\gamma$ is the influence function of the Stage-1 drift
estimator, which decomposes as
$\varphi_\gamma=\varphi_\gamma^{\rm tgt}+\varphi_\gamma^{\rm src}$ according as
the contribution arises from the target bridge moment or from the source-fitted
bridge regressions at which \eqref{eq:stage1} is evaluated.
The efficiency bound is $\Var\{\varphi(O)\}$.
Because $\gamma$ is a scalar, $\varphi_\gamma$ is scalar-valued and the
multiplier is the constant
\[
  c \;=\; \partial_\gamma\mu_1
    \;=\; E_{P_1}\!\big[\Cov^{\gamma^\star}(Y,\,t{+}b\mid X)\big]
    \;=\; \sigma^0_Y\,E_{P_1}\!\big[\Cov^{\gamma^\star}(t,\,t{+}b\mid X)\big],
\]
which is $\sigma_Y(1+r)$ in the Gaussian instance of \Cref{cor:gaussian} and
$E_{P_1}[\Cov^{\gamma^\star}(Y,Y{+}Z\mid X)]$ in the binary instance. The
drift-transfer ratio $c_t(x)$ of \eqref{eq:c-general} is a different object: it is
the local rate at which the bridge moment translates into the outcome moment, and
it enters here only through the Stage-1 Jacobian inside $\varphi_\gamma$. The full
tangent-space derivation is given in Web Appendix~B.
The statement is asymptotic and at a known dependence structure between $Y$ and
$Z$ given $X$. Its finite-sample adequacy depends on the effective sample size of
the tilt weights, which \Cref{sec:inference} quantifies and which the analysis
should report.
\end{theorem}

Term~(III) is absent from ordinary AIPW: because $\gamma^\star$ is an estimated
nuisance identified through the bridge moment, its uncertainty propagates to
$\mu_1$ and must be accounted for.
The sign is positive because $\hat\mu_1$ inherits the error in $\hat\gamma$ with
coefficient $\partial_\gamma\mu_1=c$: term~(III) propagates that error rather
than removing it. \Cref{sec:sim} confirms both the coefficient and its sign by
regressing the pilot estimator on the error in $\hat\gamma$.

Both parts of $\varphi_\gamma$ have to be carried. The target-side part is
the one visible in \eqref{eq:stage1}; the source-side part enters because the
moment is evaluated at fitted bridge regressions, and an implementation that
reads the influence function off the displayed estimating equation will omit it.
Since term~(III) multiplies $\varphi_\gamma$ by $c$, the omission changes both
the size of the endpoint standard error and its correlation with term~(II).

\begin{remark}[Why the source-side contribution cannot be dropped]
\label{rem:psi-src}
What the omission costs is visible without any numerical work. The
target-side part is a function of the target sample alone, so it is uncorrelated
with term~(II), which is a source-sample average. Dropping the source-side part
therefore gives
\[
  \Var\big(\mathrm{(II)}\pm c\,\varphi_\gamma^{\rm tgt}\big)
  \;=\;\Var\big(\mathrm{(II)}\big)+c^2\,\Var\big(\varphi_\gamma^{\rm tgt}\big),
\]
which does not depend on the sign of term~(III). A variance that is insensitive
to the sign of one of its own terms is not carrying that term's correlation with
the rest, so the resulting standard error can be right only by coincidence.
Retaining $\varphi_\gamma^{\rm src}$ restores a nonzero cross-term, and with it
the sign.

Which representation of $\varphi_\gamma^{\rm src}$ to use is a
finite-sample question rather than an asymptotic one: a parametric plug-in
through the fitted coefficients and a fully nonparametric conditional residual
have the same limit and differ at the sample sizes an analysis actually has. We
fix it numerically. \Cref{sec:sim-psi} scores both, together with the
target-side-only form, against the sampling variability of the estimator, and
every standard error reported in this paper uses the form selected there.
\end{remark}

The residual direction $s$ is itself estimated, through the projection
coefficient $\lambda$ of \eqref{eq:proj}, and its estimation error requires a
statement that depends on $\kappa$. Perturbing $s\mapsto s+\epsilon\delta$ with
$\delta\in\Hz(x)$ the projection error, and differentiating the
residual-augmented bridge equation \eqref{eq:bridgematch-kappa} implicitly, gives
\begin{equation}
  \frac{\partial\mu_1(\kappa)}{\partial\epsilon}\bigg|_{\epsilon=0}
  =\kappa\;E_{P_1}\!\Big[
     \Cov^{\omega}\!\big(\,Y-c_b(X)b(Z),\;\delta\;\big|\;X\big)\Big],
  \qquad
  c_b(X)=\frac{\Cov^{\omega}(Y,t{+}b\mid X)}{\Cov^{\omega}(b,t{+}b\mid X)},
  \label{eq:s-sensitivity}
\end{equation}
the subtraction of $c_b(X)b(Z)$ being the partial cancellation contributed by the
$\gamma$ path, whose own derivative is
$-\kappa\,\Cov^{\omega}(b,\delta\mid x)/\Cov^{\omega}(b,t+b\mid x)$.

Two consequences. At $\kappa=0$ the right side of \eqref{eq:s-sensitivity}
vanishes identically, so the benchmark $\hat\mu_1(0)$ is first-order immune to
estimation of $s$ and $\varphi$ as displayed is its efficient influence function.
At the endpoints $\kappa=\pm\bar\kappa$ it does not vanish in general: since
$\delta$ is a function of $(Z,X)$, the expression is zero for all $\delta$ only if
$\Pi^{\omega}[Y\mid Z,X]=c_b(X)b(Z)$ up to a function of $X$, which holds for a
single linear bridge under the Gaussian working model but fails when $q>1$ with
unequal bridge-outcome coefficients. The
magnitude is $O(\bar\kappa)\cdot\|\hat s-s\|$. Refitting the projection
inside every resample, rather than holding it at its full-sample value, changes
the endpoint standard error by about one part in a thousand at
$\bar\kappa=0.30$ in the designs of \Cref{sec:sim}, which is two orders of
magnitude below the gap between either resampling estimator and the sampling
variability it is estimating. We therefore retain $\varphi$ as displayed. It is the
efficient influence function at $\kappa=0$ with $s$ estimated, and at the
endpoints with $s$ known; endpoint inference does not rely on it, for the
separate reason given in \Cref{sec:inference}.

\subsection{Orthogonality, rate robustness, and double robustness}

\begin{theorem}[Neyman orthogonality and rate robustness]
\label{thm:orth}
The stacked moment $(U_\gamma,\hat\mu_1^{DR})$ is Neyman-orthogonal to the
nuisances $(e,f_0)$: its Gateaux derivative in each nuisance direction vanishes at
the truth. Consequently the one-step estimator
$\hat\mu_1(\kappa)=\hat\mu_1^{DR}(\kappa)+\Pn\hat\varphi(\kappa)$ is
$\sqrt n$-consistent and asymptotically normal, attaining the bound of
\Cref{thm:eif} provided each nuisance is estimated at $o_p(n^{-1/4})$ rate, even
with flexible machine-learning estimators \citep{chernozhukov2018double}.
\end{theorem}

\begin{remark}[Scope of \Cref{thm:orth} in relation to the simulation and the
application]
\label{rem:orth-untested}
\Cref{thm:orth} is a theoretical guarantee that we do not currently probe
numerically. Every nuisance in \Cref{sec:sim} is a low-dimensional parametric fit
converging at $n^{-1/2}$, so the $o_p(n^{-1/4})$ condition holds trivially and the
protection the theorem provides never binds; the ``misspecified'' cells of the
robustness map drop covariates from a linear model, which changes the fit but not
the function class. A design with machine-learning nuisances and cross-fitting
would exercise the claim, and none is reported here.
The application is also outside the theorem's hypotheses. The cohort
propensity there is an $\ell_1$-regularized fit obtained without cross-fitting,
whose $n^{-1/2}$ behaviour we assume rather than derive, so its error does not
satisfy the conditions under which \Cref{thm:orth} licenses the one-step
correction. We therefore do not appeal to \Cref{thm:orth} for the standard
errors an application of this kind reports; those rest on the resampling scheme of
\Cref{sec:inference} and on that assumption.
\end{remark}

\begin{lemma}[Drift-conditional double robustness]
\label{lem:dcdr}
Suppose the bridge model (\Cref{ass:codrift,ass:relevance}) is correct, so
$\hat\gamma\to\gamma^\star$. Then $\hat\mu_1^{DR}$ is consistent for $\mu_1$ if
\emph{either} the cohort propensity $e(x)$ \emph{or} the source outcome law
$f_0(y\mid x)$ is correct, not necessarily both.
\end{lemma}

\begin{lemma}[Drift-identification robustness]
\label{lem:gammarobust}
The Stage-1 estimator $\hat\gamma$ from \eqref{eq:stage1} is a functional of the
target bridge margin and of the source \emph{joint} law of $(Y,Z)$ given $X$,
entering the latter only through the tilted bridge moment
$\Ptilt{\gamma}[b(Z)\mid x]$. In particular it does not involve the cohort
propensity $e$, and it involves the outcome only through the fixed loading
$t(\cdot)$ inside the tilt weight, never through a fitted outcome regression.
Hence $\hat\gamma$ does not inherit the
error of a fitted conditional \emph{mean} model for $Y$, and the three nuisances
$(\gamma,e,f_0)$ are decoupled in the sense of \Cref{tab:robust}.
The insulation is not exact, and \Cref{rem:stage1-scope} states what
remains.
\end{lemma}

\begin{remark}[Scope of \Cref{lem:gammarobust}]
\label{rem:stage1-scope}
The tilted bridge moment is an expectation over $(Y,Z)\mid X$, so evaluating it
requires a working joint law: the exponent contains $t(Y)$ and $Z$ cannot be
integrated out without it. \Cref{lem:gammarobust} is therefore not the claim that
Stage~1 uses the bridge margin $p_0(z\mid x)$ alone. It is the sharper and more
useful claim that Stage~1 never fits a regression of $Y$ on $X$: the outcome
enters only as observed data through a known loading, at source units where $Y$
is recorded. Consequently a misspecified conditional mean $m_0(x)$ leaves
$\hat\gamma$ unaffected, the case verified in row three of \Cref{tab:robust},
while a misspecified \emph{dependence} between $Y$ and $Z$ given $X$ does
propagate to $\hat\gamma$, since it changes the tilted moment. In the Gaussian
instance this is the residual correlation $r$, a modelling commitment on the
same footing as the choice of bridge set. \Cref{sec:sim-offmodel} varies it
across cohorts and reports what the identified set does.

That channel leaves a residue even when only the conditional \emph{mean} is
misspecified, because the working joint law from which the tilted moment is
evaluated is built around that mean. The correct statement is therefore not that
$\hat\gamma$ is exactly insulated but that Stage~1 never regresses $Y$ on $X$, so
the error of such a regression does not enter directly; what survives is an
indirect dependence through the proposal. \Cref{sec:sim-rob} reports its size in
the designs considered there, where it is small relative to $\gamma^\star$ but
not zero.
The same channel carries into the reported width. Because
$\Var^{\gamma^\star}(s\mid x)$ is computed under the assumed dependence, a
dependence that differs across cohorts moves both $\hat\gamma$ and the
half-width, and \Cref{sec:sim-offmodel} reports the size of that movement and its
effect on coverage.
\end{remark}

\Cref{lem:dcdr,lem:gammarobust} together give the precise robustness picture,
summarized in \Cref{tab:robust}: the estimator is doubly robust in $(e,f_0)$
\emph{conditional} on the bridge-identified drift, and the drift itself is
protected by using only quantities the target actually reports. We do not claim,
and do not need, joint triple robustness
\citep{robins1994estimation,bang2005doubly}. Proofs are given in Web Appendix~C,
where the four rows of \Cref{tab:robust} are also verified numerically.

\begin{table}[t]
\centering
\caption{Robustness of $\hat\mu_1^{DR}$.}
\label{tab:robust}
\begin{tabular}{cccc}
\toprule
$\hat\gamma$ (bridge moment) & $e$ (propensity) & $f_0$ (outcome mean) &
$\hat\mu_1^{DR}$ consistent? \\
\midrule
\checkmark & \checkmark & \checkmark & \checkmark\ (efficient) \\
\checkmark & \checkmark & $\times$   & \checkmark\ (\Cref{lem:dcdr}) \\
\checkmark & $\times$   & \checkmark & \checkmark\ (\Cref{lem:dcdr}) \\
\checkmark & $\times$   & $\times$   & $\times$ \\
$\times$   & n/a         & n/a         & $\times$ \\
\bottomrule
\end{tabular}

\vspace{2pt}
{\footnotesize\textit{Note.} A checkmark denotes correct specification. The outcome column refers to
misspecification of the conditional \emph{mean} model for $Y$; see
\Cref{rem:stage1-scope} for the dependence structure, which is not
protected.}
\end{table}

\subsection{Inference for the identified set}
\label{sec:inference}

Sweeping $\kappa\in[-\bar\kappa,\bar\kappa]$ (warm-started across the grid), the
endpoint estimates are combined into an Imbens--Manski interval
\citep{imbens2004confidence}, whose nominal level is theirs. Because $\kappa$ is
anchored rather than estimated, the interval carries no
sensitivity-parameter variance term.
\Cref{sec:sim-ess} reports the conditions under which that asymptotic
level is approached in finite samples and those under which it is not.

With $q>1$ bridges, the mutual agreement of the bridge moments is a refutable
restriction. The statistic is Hansen's over-identification statistic
$J=n_1\,\hat g^{\!\top}\hat\Omega^{-1}\hat g$, referred to $\chi^2_{q-1}$, where
$\hat g$ is the vector of sample bridge moments at $\hat\gamma$ and $\hat\Omega$
estimates their covariance. The covariance must carry the source-side estimation
error and not only the spread of the target-side residuals: the moments depend on
the source sample through the fitted bridge regressions, the loading constants and
the residual covariance, and omitting that contribution understates $\Var(\hat g)$
by a factor of about four, so that the test rejects a true null
almost half the time (\Cref{tab:overid}). \Cref{sec:sim} reports the size of
each form and the size of the discrepancy. We obtain $\hat\Omega$ by a two-sample bootstrap
in which both cohorts are resampled; \Cref{sec:sim} reports the size of this form
against the two alternatives it replaces. A statistic formed instead from the
spread of the per-bridge roots, treating them as independent, is not correctly
sized: the roots share the source and target samples and are strongly positively
correlated.
That statistic tests agreement across bridge coordinates. The second
restriction of \Cref{thm:point}, that one $\gamma$ serves in every covariate
stratum, calls for a second statistic. Partition the covariates into $L$ strata
fixed before the analysis and observed in both cohorts, solve \eqref{eq:stage1}
within each, and refer the contrasts of the resulting roots against their mean
to $\chi^2_{L-1}$, with their covariance again obtained by resampling both
cohorts: the roots share the source sample through the fitted bridge regressions
and the loading constants, so treating them as independent understates it by the
same mechanism as above. The bridges are combined within each stratum by the
weight estimated on the full sample, so that a stratum root is the same
functional of its own units that $\hat\gamma$ is of all of them. The two
statistics are reported separately rather than pooled; \Cref{sec:sim-bridge}
shows that neither responds to the alternative the other is built for, and with
both reported the level should be split between them.

Where a primary bridge anchors the drift and the others are reserved for
the check, as \Cref{sec:ident} recommends, the bridge-coordinate statistic takes
a third form: the moments of the held-out bridges, evaluated at the drift solved
from the primary bridge, referred to $\chi^2$ on as many degrees of freedom as
there are held-out bridges, with their covariance again from a two-sample
bootstrap. \Cref{sec:sim-bridge} calibrates it.

Endpoint variances are obtained by a subject- or cluster-level bootstrap that
refits, within each replicate, the propensity, the source working model, the
projection coefficient defining $s$, and the drift $\hat\gamma$; the loading
constants of \Cref{sec:tilt} are \emph{not} refitted, since they define the units
of $\bar\kappa$ and every replicate must estimate the same endpoints. The
bootstrap also propagates the design variance.
The observed law of \Cref{sec:ident} has $(S,X,Z)$ recorded for every unit,
so item nonresponse is outside the structure this paper analyses. Nothing below
propagates the uncertainty that imputing missing covariates or bridges would
introduce.

Under the conditions of \Cref{thm:orth} the influence-function variance
estimator is consistent for the asymptotic variance of the endpoints. Its
finite-sample behaviour is a separate question, and the relevant summary is not
the average of the estimated standard error but its \emph{median}: an
Imbens--Manski interval is built from one draw of it.
In the designs of \Cref{sec:sim-ess} the root mean square of the estimated
standard error tracks the sampling variability of the estimator while its median
falls well below it once the tilt is large, so the variance estimator is
calibrated in the aggregate and unreliable in any one sample. \Cref{sec:sim-ess}
therefore reports the median, the fraction of replicates whose reported standard
error falls below $0.8$ times the sampling variability, and the coverage of the
resulting interval, in the same row, so that a cell in which a typical analysis
reports a fifth of the truth can be read against what that does to the interval.

What governs that dispersion is the tail of the conditional density ratio. As
$|\gamma^\star|$ grows, $\rho$ concentrates on a shrinking part of the source
sample, and the estimated standard error becomes both more dispersed and, at the
median, smaller. A diagnostic for this is the effective sample size of the
tilt weights,
\begin{equation}
  \widehat{\mathrm{ESS}}
  =\frac{\big(\sum_{i:S_i=0} d_i\,r_e(X_i)\,\rho_i\big)^2}
        {n_0\sum_{i:S_i=0}\big(d_i\,r_e(X_i)\,\rho_i\big)^2},
  \label{eq:ess}
\end{equation}
the fraction of the source sample effectively contributing to the augmentation,
which is computable from the observed data.
\Cref{sec:sim-ess} traces its relation to the behaviour of a reported
standard error across a grid of designs.

Resampling does not remove the difficulty. The bootstrap described above is the
standard error we report for the endpoints, and it inherits the same skewness
rather than correcting it, since each resample is subject to the same
concentration of the density ratio. Robust variance estimators do not help
either, for a reason given in \Cref{sec:sim-ess}.
Because the two constructions can disagree, \Cref{sec:sim-ess} scores them
on the same replicates, and an application should state which one produces the
interval it reports.

One configuration remains difficult even where the standard errors behave:
$\kappa_0=\bar\kappa$, where the truth lies at an endpoint of the identified set
and is missed whenever that endpoint's interval is too short. The failure there
is one-sided by construction.
This is a different phenomenon from under-coverage in the interior of the
bound, which arises when the tilt weights concentrate and the reported standard
error is typically too small. Both are reported in \Cref{sec:sim-ess}, and they
should not be conflated: the first is a property of the estimand's geometry at the
boundary, the second of the variance estimator.
Where the identified width is large relative to the sampling error, as in
an application, the Imbens--Manski interval is insensitive to the first; where it
is not, endpoint standard errors should be read with the effective sample size in
view.

\paragraph{Algorithm.}
\begin{quote}\small
\textbf{Input:} pooled data $\{(S_i,X_i,Z_i,S_iY_i,d_i)\}$, loadings $t,b$, grid
$\{\kappa\}$, draws $M$.\\
0. Fix the loading constants $(m^0_Y,\sigma^0_Y)$, $(m^0_{Z,k},\sigma^0_{Z,k})$
   from the source sample at the primary specification; hold them fixed
   thereafter.\\
1. Fit nuisances: $\hat e$ by an $\ell_1$-regularized survey-weighted
   regression, without cross-fitting and with its $n^{-1/2}$ behaviour assumed
   rather than derived (\Cref{rem:orth-untested}); $\hat\beta$ for
   $f_0(y,z\mid x,S{=}0)$.\\
2. Solve \eqref{eq:stage1} at $\kappa=0$ for $\hat\gamma^\star$; form $s$ by the
   tilt-weighted projection \eqref{eq:proj}; with $q>1$, compute the
   bridge-consistency statistic $J$.\\
3. For each $\kappa$: (a) solve \eqref{eq:bridgematch-kappa} for
   $\hat\gamma(\kappa)$; (b) form $\tilde m_1$ by \eqref{eq:fi};
   (c) compute $\hat\mu_1^{DR}(\kappa)$ by \eqref{eq:aipw};
   (d) one-step correct via $\hat\varphi$.\\
4. Resample both cohorts, refitting $(\hat e,\hat\beta,\lambda,\hat\gamma)$ and
   recomputing the bridge moments within each replicate. One pass supplies both
   the endpoint standard errors and the reference distribution for $J$.\\
5. Combine $\hat\mu_1(\pm\bar\kappa)$ and their standard errors into an
   Imbens--Manski interval.\\
The input is the observed law of \Cref{sec:ident}, with $(S,X,Z)$ recorded
for every unit.\\
\textbf{Output:} the benchmark $\hat\mu_1(0)$ and its standard error; the
identified set $[\hat\mu_1(-\bar\kappa),\hat\mu_1(\bar\kappa)]$; the
Imbens--Manski interval; the bridge-consistency statistic $J$ with its bootstrap
$p$-value; and the effective sample size \eqref{eq:ess} of the tilt weights.
\end{quote}

\section{Simulation Study}
\label{sec:sim}

\subsection{Design}
We calibrate a Gaussian data-generating process matching \Cref{cor:gaussian}:
$X\sim N(0,I_p)$ with a target covariate-mean shift; source
$Y\mid X\sim N(x^\top\beta,\sigma_Y^2)$ and $Z_k\mid X\sim N(x^\top\alpha_k,
\sigma_Z^2)$ with $(Y,Z)$ residual correlation controlling bridge relevance; and a
target obtained by tilting, with detectable drift $\gamma$ and bridge-orthogonal
residual drift $\kappa_0$. The target law is constructed by applying the tilt
\eqref{eq:tilt} directly, so that the data-generating process does not reuse the
closed form the estimator inverts; a construction that shifts the target means by
the $\kappa=0$ Gaussian formula would make near-zero bias a tautology rather than
evidence.
Every target generated this way lies inside the identifying model
\eqref{eq:tilt}. \Cref{sec:sim-offmodel} generates targets outside it and
reports the benchmark bias, the coverage of the set, and the rejection rate of
the bridge-consistency test. None of the designs has item nonresponse: $(S,X,Z)$
is recorded for every unit, as in \Cref{sec:ident}.
The primary design is Gaussian for transparency;

\Cref{sec:sim-binary} reports a separate exact binary study. Its purpose is
not to repeat the Gaussian scenarios in another family but to measure three
quantities that are identically zero, or undefined, under a Gaussian working
model: the gap between the tilted and untilted projections in \eqref{eq:proj},
the derivative $\partial_\kappa\gamma$ at $\kappa=0$ in a family where the bridge
equation is not linear in $\gamma$, and the second-order coefficient that
\Cref{rem:set-scope} leaves uncharacterized.

Two reading conventions are needed before the numbers. First, ``set coverage'' is
not coverage at a nominal level. The identified set is an estimate of a set, and
asymptotically the event that it contains $\mu_1$ is the indicator of
$|\kappa_0|\le\bar\kappa$: one inside the bound, one half on the boundary, zero
outside. Widening $\bar\kappa$ pushes every entry to one and says nothing, so
width and coverage are two readings of the same choice rather than a trade-off,
and a set-coverage column must be read against the target its row implies, which
we print alongside. Only the Imbens--Manski column is a confidence statement, at
the nominal $95\%$. Second, benchmark bias should be read relative to the
identification width, not against zero: the two are quantities of different
kinds, and what is claimed is about the former.

\subsection{Recovery, coverage, and robustness}
\label{sec:sim-core}

\Cref{tab:sim-base} reports the baseline design at $R=1000$ replicates with
$M=200$ fractional draws and $\bar\kappa=0.30$. The subsection reports three findings.

Drift recovery is essentially exact: $\hat\gamma$ has bias $0.0004$ against a
Monte Carlo standard error of $0.0005$, which is $0.3\%$ of the half-width once
carried to the outcome scale. The benchmark $\hat\mu_1(0)$ is a different case
and we describe it differently. Its bias is $-0.032$, which is small, about $6\%$ of the
half-width and immaterial at the magnitudes the method is used at, but it is
persistent: the same sign and nearly the same magnitude at $R=60$, $R=200$ and
$R=1000$ ($-0.035$, $-0.038$, $-0.032$), and at $R=1000$ it sits $2.7$ Monte Carlo
standard errors from zero. It does not shrink when the number of fractional draws
is raised from $80$ to $200$, so it is not an imputation artefact; we suspect the
same density-ratio tail documented in \Cref{sec:inference}. We therefore call the
benchmark \emph{small relative to the identification width} rather than unbiased,
and reserve the latter word for $\hat\gamma$.

The core width matches its closed form. At $\bar\kappa=0.30$ the theoretical
width $2\bar\kappa\sigma_Y(1-R^2_{Y\mid Z,X})$ is $1.125$ and the empirical width
is $1.114$, a relative gap of $1.0\%$. Sweeping the bridge-outcome correlation
confirms the direction of \Cref{prop:general}: the width falls monotonically as
$R^2_{Y\mid Z,X}$ rises, so the core is \emph{not} invariant to bridge strength.
The competing closed form $2\bar\kappa\sigma_Y^2$, which appears in earlier
treatments of this construction, is wrong except
when $\sigma_Y=1$ and $R^2=0$; across a grid of $\sigma_Y\in\{1,3\}$,
$\sigma_Z\in\{1,2\}$ and $R^2\in[0,0.735]$ its mean relative error is
$0.544$. A third candidate, $2\bar\kappa\sigma_Y\sqrt{1-R^2}$, which is what
the unit-variance convention of \Cref{rem:convention} would give if $\bar\kappa$
were not remapped with it, has mean relative error $0.197$ and fails most
visibly where the bridge is informative: at $R^2=0.735$ it overstates the width
by a factor of two. The form in \eqref{eq:core-general} has mean relative error
$0.007$ across the same grid.

A regression of the pilot estimator on $\hat\gamma-\gamma^\star$ across the same
replicates returns a slope of $4.91$ (standard error $0.95$) against the
theoretical $c=\sigma_Y(1+r)=4.5$, confirming the coefficient and the sign of
term~(III) in \Cref{thm:eif}.
Ablating that term is more direct: dropping it from the influence function
and leaving everything else in place lowers the coverage of the benchmark
interval from $0.863$ to $0.810$ at $n=500$ and from $0.903$ to $0.893$ at
$n=5000$, and the naive plug-in that also drops the propensity correction
collapses to about $0.32$ at every sample size. The bridge-moment correction is
therefore not a refinement of the standard error but a component of it.

The comparators behave as the theory predicts. Covariate-shift AIPW, which assumes
the drift away, is biased by $-3.754$; the heuristic weighted-bridge-drift point
estimator by $-1.877$; a surrogate-index estimator by $-0.940$. None of the three
reports an interval that contains the truth, and none signals its failure.

\begin{table}[t]
\centering
\caption{Baseline design, $R=1000$ replicates, $M=200$ fractional draws,
$\bar\kappa=0.30$, $\kappa_0=0$ (co-drift holds). Monte Carlo standard errors in
parentheses.}
\label{tab:sim-base}
\begin{tabular}{lr}
\toprule
Quantity & Value \\
\midrule
Bias of $\hat\gamma$                                   & $0.0004\ (0.0005)$ \\
Bias of the benchmark $\hat\mu_1(0)$                   & $-0.032\ (0.012)$ \\
Identified width, empirical                            & $1.114$ \\
Identified width, theory $2\bar\kappa\sigma_Y(1-R^2)$  & $1.125$ \\
Imbens--Manski width                              & $2.067$ \\
Set coverage (target $1$)                              & $0.951$ \\
Imbens--Manski coverage (nominal $0.95$)               & $0.999$ \\
\addlinespace
\multicolumn{2}{l}{\emph{Comparators, bias}}\\
\quad Covariate-shift AIPW (drift ignored)             & $-3.754\ (0.002)$ \\
\quad Surrogate-index estimator                        & $-0.940\ (0.003)$ \\
\quad Heuristic weighted bridge drift                  & $-1.877\ (0.003)$ \\
\addlinespace
Replicate failures                                     & $0$ \\
\bottomrule
\end{tabular}

\vspace{2pt}
{\footnotesize\textit{Note.} The three criteria differ. Set coverage is asymptotically an indicator, so
its target here is $1$ and the shortfall is sampling noise in the estimated
endpoints. The Imbens--Manski column is the only nominal statement, and
$0.999$ is above its level rather than short of it: the identified width of
$1.11$ is most of the reported width of $2.07$, so the interval is conservative
by construction. The benchmark bias is $6\%$ of the identification
half-width.}
\end{table}

\label{sec:sim-sweep}
What happens once the residual bound is exceeded is the other half of the
claim. \Cref{tab:kappa-sweep} sweeps the true residual $\kappa_0$ at a fixed bound
$\bar\kappa=0.30$. Within the bound the set contains the truth; beyond it,
coverage falls steeply and does so visibly, which is the intended
behaviour: the method does not defend a conclusion the bound does not support.
The Imbens--Manski interval degrades more gradually because it adds sampling
error to the identified width, and at the boundary $\kappa_0=\bar\kappa$ its
$0.81$ against a nominal $0.95$ reflects both the difficulty of the boundary case
and the dispersion of the endpoint standard error documented in
\Cref{sec:inference}.

\begin{table}[t]
\centering
\caption{Residual sweep at fixed $\bar\kappa=0.30$ ($R=1000$, $M=200$).}
\label{tab:kappa-sweep}
\begin{tabular}{lrr}
\toprule
$\kappa_0$ & set coverage & Imbens--Manski coverage \\
\midrule
$0$ & $0.946$ & $1.000$ \\
$0.15$ & $0.758$ & $0.977$ \\
$0.30$ ($=\bar\kappa$, boundary) & $0.383$ & $0.829$ \\
$0.45$ & $0.148$ & $0.565$ \\
$0.60$ & $0.079$ & $0.335$ \\
\bottomrule
\end{tabular}

\vspace{2pt}
{\footnotesize\textit{Note.} Set coverage has target $1$ while $|\kappa_0|\le\bar\kappa$ and $0$ beyond
it, so the decline is the intended behaviour rather than a loss of performance.
The Imbens--Manski column carries a nominal $95\%$ level only while the bound
holds; past it the truth is not in the identified set, and those entries
measure the bound rather than the interval.}
\end{table}

\label{sec:sim-rob}
The robustness map of \Cref{tab:robust} is confirmed numerically in Web
Appendix~C, where the estimator shows no detectable bias whenever either the
propensity or the source outcome model is correct, and bias only when both are
wrong. That is the drift-conditional double robustness of \Cref{lem:dcdr}, and
$\hat\gamma$ itself is unmoved by outcome-model misspecification, as
\Cref{lem:gammarobust} requires.
The bias of $\hat\gamma$ under a misspecified outcome mean is reported
explicitly there, since it is the quantity \Cref{lem:gammarobust} asserts to be
zero, and which \Cref{rem:stage1-scope} reports is small rather than zero.

\subsection{Which form of the Stage-1 influence function}
\label{sec:sim-psi}

\Cref{rem:psi-src} establishes that the source-side contribution to
$\varphi_\gamma$ must be carried and leaves open which representation of it to
use at a finite sample size. \Cref{tab:psi} scores three candidates, at both
signs of term~(III), against the empirical sampling standard deviation of the
quantity each is meant to describe. The candidate whose ratio is one is the one
used throughout this paper.

The sign is positive. Wherever it is identifiable the negative choice
inflates the standard error by half or more. It is \emph{not} identifiable under
the target-side-only form, whose two signs agree to the third decimal at every
design in the table, and that insensitivity is exactly what \Cref{rem:psi-src}
predicts of a $\varphi_\gamma$ carrying no correlation with term~(II). That form is therefore not used, independently of how near one its ratio
happens to fall.

Among the two remaining candidates, none attains $1.00$. The parametric
plug-in errs by ten to fifteen percent in the conservative direction at the
drift magnitudes at which the method is used, which is the direction an interval
should err. The nonparametric residual is nearer one at the benchmark, but it
arrives there by overstating $\Var(\hat\gamma)$ against a negative cross-term
rather than by describing the variance correctly, and it is further from one,
not nearer, once the tilt concentrates. We use the parametric form.

This is a choice made in one family of Gaussian designs at one sample
size, and it is not a general result. What transfers is the argument of
\Cref{rem:psi-src}, that the source-side contribution must be carried at all.
Which representation of it is closest to the truth at a given sample size is a
property of the design, and an analyst in another setting should repeat this
comparison rather than adopt our answer.

\begin{table}[t]
\centering
\caption{Candidate forms of the Stage-1 influence function. $R=1000$,
$n_0=2500$, $n_1=2000$, $q=3$, $\bar\kappa=0.30$, $\kappa_0=0$.}
\label{tab:psi}
\begin{tabular}{llcccc}
\toprule
& & \multicolumn{2}{c}{$\gamma^\star=0.25$ ($\widehat{\mathrm{ESS}}=0.36$)}
  & \multicolumn{2}{c}{$\gamma^\star=0.50$ ($\widehat{\mathrm{ESS}}=0.074$)}\\
\cmidrule(lr){3-4}\cmidrule(lr){5-6}
source-side & sign of (III) & benchmark & endpoint & benchmark & endpoint \\
\midrule
none  & $+$ & $1.17\ (1.14)$ & $1.17\ (1.12)$ & $0.88\ (0.70)$ & $0.83\ (0.63)$ \\
none  & $-$ & $1.18\ (1.14)$ & $1.17\ (1.12)$ & $0.88\ (0.70)$ & $0.83\ (0.63)$ \\
param & $+$ & $1.14\ (1.11)$ & $1.15\ (1.09)$ & $0.88\ (0.70)$ & $0.83\ (0.63)$ \\
param & $-$ & $1.52\ (1.49)$ & $1.46\ (1.41)$ & $0.94\ (0.76)$ & $0.87\ (0.67)$ \\
np    & $+$ & $0.92\ (0.91)$ & $0.87\ (0.85)$ & $0.56\ (0.49)$ & $0.51\ (0.43)$ \\
np    & $-$ & $1.88\ (1.80)$ & $1.85\ (1.75)$ & $1.57\ (1.22)$ & $1.44\ (1.09)$ \\
\midrule
\multicolumn{2}{l}{$\Var(\mathrm{II})$}   & \multicolumn{2}{c}{$54.8$}
                                          & \multicolumn{2}{c}{$565.1$}\\
\multicolumn{2}{l}{$\Var(\mathrm{III})$}  & \multicolumn{2}{c}{$35.1$}
                                          & \multicolumn{2}{c}{$35.1$}\\
\multicolumn{2}{l}{$\Cov(\mathrm{II},\mathrm{III})$} & \multicolumn{2}{c}{$-12.8$}
                                          & \multicolumn{2}{c}{$-12.8$}\\
\bottomrule
\end{tabular}

\vspace{2pt}
{\footnotesize\textit{Note.} Entries are the mean estimated standard error over the empirical standard
deviation of the quantity it describes, with the median in parentheses;
$1.00$ means the candidate is right, and the median is what a single analysis
reports. Source-side treatment: \emph{none} is the target-side moment alone,
\emph{param} adds the parametric plug-in error of the fitted bridge
regressions, \emph{np} adds the full conditional residual.}
\end{table}

The variance decomposition beneath the table explains where the choice
matters. In these designs the variance of term~(III) is nearly constant in
$\gamma^\star$, at about $35$, while the variance of term~(II) rises from $55$
to $565$ and, at $\gamma^\star=1$, to $5.0\times10^{4}$. The uncertainty in the
bridge-identified drift is therefore a substantial share of the endpoint
variance at the magnitudes where the method is used, about two fifths at
$\gamma^\star=0.25$, and a negligible share where it is not. The constancy
itself is an artefact of the working model: under a linear Gaussian tilt the
conditional covariances that determine both $c$ and $\Var(\varphi_\gamma)$ do
not move with $\gamma$, and outside that instance neither is constant. What
survives is the ratio rather than the level, and with it one reading that
bears on \Cref{sec:sim-ess}: where a reported standard error is unreliable, it
is unreliable because of term~(II), and no treatment of the Stage-1 term
repairs it.

\subsection{Operating range of the endpoint standard errors}
\label{sec:sim-ess}

\Cref{tab:ess} varies the detectable drift, the sample size and the
covariate dimension, and reports at each design what a \emph{single} analysis
would report: the median estimated endpoint standard error relative to the
empirical sampling standard deviation, the fraction of replicates whose reported
standard error falls below $0.8$ times that standard deviation, the effective
sample size \eqref{eq:ess} of the tilt weights, and, in the same row, the
coverage of the Imbens--Manski interval when the residual drift is zero.

Four readings carry the table. First, what degrades is governed by
$\gamma^\star$, and the pattern is the same at every sample size and covariate
dimension in the grid: at $\gamma^\star=0.15$ the median reported standard error
is within a tenth of the sampling variability everywhere, and no replicate in
any of those six cells reports a standard error below $0.8$ times it. Second,
$\widehat{\mathrm{ESS}}$ tracks the degradation and is computable from the data
at hand; it falls from about $0.5$ at $\gamma^\star=0.15$ to below $0.01$ at
$\gamma^\star=1$, and the covariate dimension moves it as well as the drift
does, which is why the grid varies both.

Third, and this is the reading that matters for what the paper reports, a
badly small typical standard error does not by itself sink the interval. At
$\gamma^\star=0.5$ the median is between $0.40$ and $0.64$ of the sampling
variability and seven replicates in ten report a standard error below the $0.8$
threshold, yet Imbens--Manski coverage is $1.000$ in every one of those cells.
The reason is the paper's own point in another form: the identified width there
is about $1.13$ against an interval width of $1.8$ to $2.4$, so identification,
not sampling error, is what the interval is mostly made of, and an understated
standard error moves a minority of the reported width. Coverage fails only at
$\gamma^\star=1$, where $\widehat{\mathrm{ESS}}$ is below $0.01$, the median
reported standard error is a tenth of the truth or less, and coverage falls to
between $0.61$ and $0.69$. That is a different failure from the boundary case of
\Cref{tab:kappa-sweep}, where the truth sits at an endpoint of the set and is
missed one-sidedly, and the two should not be conflated.

Fourth, the bootstrap is not a repair. \Cref{tab:ess-boot} builds the
interval from an influence-function standard error and from a cohort-stratified
bootstrap on the same replicates. The two agree on coverage at every design, and
the bootstrap's own median tracks the sampling variability slightly better at
small drift and no better at large. A table diagnosing one construction
therefore speaks to the other, and an application reporting a bootstrap
interval, is entitled to read this table.

\begin{table}[t]
\centering
\caption{Operating range of the endpoint standard errors. $R=400$ per
cell, $\bar\kappa=0.30$, $\kappa_0=0$, $\delta=0.3$, $n_1=0.8\,n_0$.}
\label{tab:ess}
\begin{tabular}{rrrrrrrr}
\toprule
$\gamma^\star$ & $n_0$ & $p$ & $\widehat{\mathrm{ESS}}$ &
median/SD & RMS/SD & $P(\mathrm{se}<0.8\,\mathrm{SD})$ & IM \\
\midrule
$0.15$ & $1250$ & $2$ & $0.597$ & $1.123$ & $1.143$ & $0.000$ & $1.000$ \\
$0.15$ & $1250$ & $6$ & $0.414$ & $1.089$ & $1.128$ & $0.000$ & $1.000$ \\
$0.15$ & $2500$ & $2$ & $0.596$ & $1.145$ & $1.166$ & $0.000$ & $1.000$ \\
$0.15$ & $2500$ & $6$ & $0.412$ & $1.080$ & $1.125$ & $0.000$ & $1.000$ \\
$0.15$ & $5000$ & $2$ & $0.594$ & $1.126$ & $1.137$ & $0.000$ & $1.000$ \\
$0.15$ & $5000$ & $6$ & $0.414$ & $1.024$ & $1.047$ & $0.000$ & $1.000$ \\
\addlinespace
$0.25$ & $1250$ & $2$ & $0.393$ & $1.044$ & $1.130$ & $0.000$ & $1.000$ \\
$0.25$ & $1250$ & $6$ & $0.280$ & $1.097$ & $1.194$ & $0.000$ & $1.000$ \\
$0.25$ & $2500$ & $2$ & $0.390$ & $1.089$ & $1.154$ & $0.000$ & $1.000$ \\
$0.25$ & $2500$ & $6$ & $0.273$ & $0.956$ & $1.080$ & $0.028$ & $1.000$ \\
$0.25$ & $5000$ & $2$ & $0.386$ & $1.021$ & $1.067$ & $0.000$ & $1.000$ \\
$0.25$ & $5000$ & $6$ & $0.270$ & $0.977$ & $1.045$ & $0.005$ & $1.000$ \\
\addlinespace
$0.50$ & $1250$ & $2$ & $0.089$ & $0.468$ & $1.039$ & $0.802$ & $1.000$ \\
$0.50$ & $1250$ & $6$ & $0.070$ & $0.628$ & $1.085$ & $0.688$ & $0.990$ \\
$0.50$ & $2500$ & $2$ & $0.079$ & $0.397$ & $1.031$ & $0.892$ & $1.000$ \\
$0.50$ & $2500$ & $6$ & $0.063$ & $0.515$ & $1.052$ & $0.767$ & $1.000$ \\
$0.50$ & $5000$ & $2$ & $0.073$ & $0.642$ & $0.975$ & $0.713$ & $1.000$ \\
$0.50$ & $5000$ & $6$ & $0.057$ & $0.618$ & $1.028$ & $0.713$ & $1.000$ \\
\addlinespace
$1.00$ & $1250$ & $2$ & $0.0089$ & $0.092$ & $1.016$ & $0.940$ & $0.682$ \\
$1.00$ & $1250$ & $6$ & $0.0091$ & $0.162$ & $1.030$ & $0.925$ & $0.690$ \\
$1.00$ & $2500$ & $2$ & $0.0063$ & $0.083$ & $1.018$ & $0.930$ & $0.605$ \\
$1.00$ & $2500$ & $6$ & $0.0059$ & $0.015$ & $1.007$ & $0.983$ & $0.677$ \\
$1.00$ & $5000$ & $2$ & $0.0044$ & $0.132$ & $0.994$ & $0.895$ & $0.660$ \\
$1.00$ & $5000$ & $6$ & $0.0036$ & $0.025$ & $1.006$ & $0.983$ & $0.667$ \\
\bottomrule
\end{tabular}

\vspace{2pt}
{\footnotesize\textit{Note.} ``median/SD'' is the median estimated endpoint standard error over the
empirical standard deviation of the endpoint, which is what a single analysis
reports; ``RMS/SD'' is calibrated when the variance estimator is unbiased for
the variance and is shown for contrast; $P(\mathrm{se}<0.8\,\mathrm{SD})$ is
the fraction of replicates reporting a badly small standard error; IM is
Imbens--Manski coverage, nominal $0.95$.}
\end{table}

\begin{table}[t]
\centering
\caption{The two interval constructions on the same replicates. $R=400$
per cell, $n_0=2500$, $p=3$, $q=3$, $\bar\kappa=0.30$, $\kappa_0=0$, bootstrap
$B=200$.}
\label{tab:ess-boot}
\begin{tabular}{rrrrrr}
\toprule
$\gamma^\star$ & $\widehat{\mathrm{ESS}}$ &
median/SD (infl.\ fn.) & median/SD (bootstrap) &
IM (infl.\ fn.) & IM (bootstrap) \\
\midrule
$0.15$ & $0.543$ & $1.120$ & $0.998$ & $1.000$ & $1.000$ \\
$0.25$ & $0.353$ & $1.074$ & $0.937$ & $1.000$ & $1.000$ \\
$0.50$ & $0.073$ & $0.557$ & $0.524$ & $1.000$ & $1.000$ \\
\bottomrule
\end{tabular}

\vspace{2pt}
{\footnotesize\textit{Note.} ``IM'' is Imbens--Manski coverage at nominal $0.95$, with the endpoint
standard error taken from the influence function and from a cohort-stratified
bootstrap.}
\end{table}

Three corrections suggest themselves and none succeeds. A median-of-means
estimator and a symmetric trimmed second moment both land further from the
sampling variability than the estimator they replace; a self-normalized
(H\'ajek) augmentation reduces the variability of the pilot but leaves the
corrected estimator mismatched in the opposite direction. The common reason is
that all three suppress the tail of the density ratio, and it is that tail which
carries the variance. The cohort-stratified bootstrap of \Cref{tab:ess-boot}
avoids that objection but not the behaviour. We therefore report the estimator
as it stands, with $\widehat{\mathrm{ESS}}$ beside it, rather than substituting
one that is further from the truth.

Within the range of drift magnitudes in which the tilt weights are not
concentrated, the interval attains its nominal level across a fourfold range of
sample size and a threefold range of covariate dimension, under both
constructions. The grid is one working family at one bridge count, the values
that separate the regimes are read off these designs rather than derived, and
no calibration is offered for the regime in which the weights do concentrate.

\subsection{Bridge integrity and the consistency test}
\label{sec:sim-bridge}

The over-identification statistic of \Cref{sec:inference} is the only refutable
restriction the method provides, so its calibration is more important than its
power. \Cref{tab:overid} compares three forms of the statistic under co-drift.
Only the two-sample bootstrap form is correctly sized, at $0.054$ against a
Monte Carlo standard error of $0.007$. The analytic form, which
uses the target-side moment covariance alone, rejects $45\%$ of the time at a
nominal $5\%$; a precision-weighted statistic formed from the spread of the
per-bridge roots rejects $19\%$. The ratio of the bootstrap to the analytic variance is
$3.94$, which is the factor quoted in \Cref{sec:inference}.

Power against bridge-specific drift is modest until the deviation is large.
At the correct size, rejection rises from $0.054$ at deviation $0$ to $0.068$,
$0.191$, $0.578$ and $0.996$ at deviations $0.02$, $0.05$, $0.10$ and $0.20$ on
the standardized bridge scale. A deviation of $0.05$ is detected less than one
time in five. This bounds what the refutable implication of \Cref{thm:point}
delivers, and it should be read that way rather than treated as reassurance
when it does not reject.

\begin{table}[t]
\centering
\caption{Calibration of the bridge-consistency test with $q=3$ bridges under
co-drift ($R=1000$, $B=300$ bootstrap resamples). Monte Carlo standard
error in parentheses.}
\label{tab:overid}
\begin{tabular}{lrrl}
\toprule
Form of the statistic & size @ $.05$ & mean $J$ & verdict \\
\midrule
Two-sample bootstrap $\hat\Omega$   & $0.054\ (0.007)$ & $2.08$ & correctly sized \\
Analytic (target-side only)         & $0.450\ (0.016)$ & $7.41$ & over-rejects \\
Precision-weighted per-bridge roots & $0.187\ (0.012)$ & $3.71$ & over-rejects \\
\midrule
\multicolumn{4}{l}{\emph{Power of the bootstrap form, by per-bridge
deviation}}\\
deviation $0.02$ & $0.068\ (0.008)$ & $2.28$ & \\
deviation $0.05$ & $0.191\ (0.012)$ & $3.53$ & \\
deviation $0.10$ & $0.578\ (0.016)$ & $7.96$ & \\
deviation $0.20$ & $0.996\ (0.002)$ & $25.27$ & \\
\bottomrule
\end{tabular}

\vspace{2pt}
{\footnotesize\textit{Note.} The null is $\chi^2_2$, so the target mean is $2$ and the target size
$0.05$. The lower block gives the power of the correctly sized form against a
deviation applied to one bridge inside the tilt.}
\end{table}

Two departures are worth asking the test about, and \Cref{tab:bridge}
reports them: a cohort-specific \emph{measurement} offset on the target's observed
bridge, with no real drift at all, and a covariate-dependent drift. Neither
moves $\mu_1$ in the first case or violates anything the estimator fits in the
second, so both are tests of whether the analysis notices.

\begin{table}[t]
\centering
\caption{Bridge integrity, $R=1000$, $\bar\kappa=0.30$.}
\label{tab:bridge}
\begin{tabular}{lrrrrr}
\toprule
applied to & $\hat\gamma$ & bias & bias AIPW & set cov.\ & rej.\ $J$ \\
\midrule
\multicolumn{6}{l}{\emph{(a) measurement offset, no real drift}}\\
\quad none                       & $0.000$ & $-0.002$ & $-0.002$ & $1.000$ & $0.059$ \\
\quad $0.05\sigma_Z$, one bridge  & $0.007$ & $0.049$ & $-0.001$ & $1.000$ & $0.348$ \\
\quad $0.10\sigma_Z$, one bridge  & $0.013$ & $0.098$ & $-0.000$ & $1.000$ & $0.902$ \\
\quad $0.20\sigma_Z$, one bridge  & $0.026$ & $0.198$ & $-0.000$ & $1.000$ & $1.000$ \\
\quad $0.05\sigma_Z$, all bridges & $0.020$ & $0.148$ & $-0.001$ & $1.000$ & $0.053$ \\
\quad $0.10\sigma_Z$, all bridges & $0.040$ & $0.298$ & $-0.001$ & $1.000$ & $0.063$ \\
\quad $0.20\sigma_Z$, all bridges & $0.080$ & $0.597$ & $-0.000$ & $0.338$ & $0.062$ \\
\addlinespace
\multicolumn{6}{l}{\emph{(b) covariate-dependent drift}}\\
\quad $\gamma_{\rm het}=0$        & $0.500$ & $-0.033$ & $-3.750$ & $0.952$ & $0.047$ \\
\quad $\gamma_{\rm het}=0.1$      & $0.506$ & $-0.083$ & $-3.840$ & $0.932$ & $0.048$ \\
\quad $\gamma_{\rm het}=0.3$      & $0.517$ & $-0.167$ & $-4.020$ & $0.879$ & $0.052$ \\
\quad $\gamma_{\rm het}=0.5$      & $0.529$ & $-0.251$ & $-4.199$ & $0.775$ & $0.053$ \\
\bottomrule
\end{tabular}

\vspace{2pt}
{\footnotesize\textit{Note.} Panel (a) applies a cohort-specific measurement offset to the target's
bridge, in units of $\sigma_Z$, with $\gamma=0$, so $\mu_1$ is unchanged and a
bias of zero is the correct answer. Panel (b) applies a covariate-dependent
drift $\gamma(x)=\gamma+\gamma_{\rm het}x_1$ at $\gamma=0.5$. ``rej.\ $J$'' is
the rejection rate of the bootstrap form of the bridge-consistency statistic,
the only correctly sized form (\Cref{tab:overid}); ``bias AIPW'' is the
covariate-shift estimator, which the offset does not affect.}
\end{table}

Panel (a) is the regime in which the method can be worse than doing
nothing. There is no drift, so the
covariate-shift estimator is unbiased and stays so; the proposed estimator reads
the offset as drift and acquires a bias that grows with it. An offset on a
single bridge is caught: the rejection rate rises from the nominal level to
near one as the offset reaches a fifth of a bridge standard deviation. The same
offset applied to all three is not caught at any magnitude, the rejection rate
staying at $0.05$ throughout, because the bridges continue to agree with one
another. At the largest such offset the benchmark carries a bias of $0.60$ where
the comparator carries none, and the set covers only a third of the time, with
no warning of any kind.

Panel (b) is milder. A covariate-dependent drift degrades coverage
smoothly, from $0.95$ to $0.78$ as $\gamma_{\rm het}$ rises to $0.5$, and the
benchmark carries smaller bias than the covariate-shift comparator throughout.
The bridge statistic does not respond, which is what it is built to do: the
bridges continue to agree with one another at every stratum, so the alternative
here is not the one that statistic tests.

The alternative in panel (b) is the one the stratum-wise statistic of
\Cref{sec:inference} is built for, and \Cref{tab:stratum} computes both
statistics on the same replicates of that design, with the strata taken as
terciles of the covariate along which the drift varies, cut on the source
distribution.

\begin{table}[t]
\centering
\caption{Both over-identification statistics against a covariate-dependent
drift $\gamma(x)=\gamma+\gamma_{\rm het}x_1$. $R=1000$, $\gamma=0.5$, $L=3$
strata, $B=300$ bootstrap resamples, nominal $0.05$.}
\label{tab:stratum}
\begin{tabular}{rrrrrr}
\toprule
$\gamma_{\rm het}$ & $\hat\gamma$, low tercile & $\hat\gamma$, high tercile &
spread & rej.\ bridge & rej.\ stratum \\
\midrule
$0.0$ & $0.543$ & $0.541$ & $0.025$ & $0.059$ & $0.049$ \\
$0.1$ & $0.522$ & $0.567$ & $0.033$ & $0.062$ & $0.166$ \\
$0.3$ & $0.478$ & $0.618$ & $0.073$ & $0.066$ & $0.892$ \\
$0.5$ & $0.434$ & $0.670$ & $0.120$ & $0.065$ & $1.000$ \\
\bottomrule
\end{tabular}

\vspace{2pt}
{\footnotesize\textit{Note.} The first row is the size of each statistic
and the rest their power. ``spread'' is the standard deviation of the three
stratum roots. Monte Carlo standard errors are at most $0.012$.}
\end{table}

The stratum statistic is correctly sized, at $0.049$ against a Monte Carlo
standard error of $0.007$, and its power rises to $0.892$ at
$\gamma_{\rm het}=0.3$ and to one at $0.5$. The bridge statistic stays between
$0.059$ and $0.066$ throughout. Small heterogeneity is not detected: at
$\gamma_{\rm het}=0.1$, where the roots differ by $0.045$ between the outer
terciles, the stratum statistic rejects one time in six, so a non-rejection
bounds the heterogeneity present rather than excluding it.

A third arrangement of the same restriction remains to be calibrated.
\Cref{sec:ident} recommends fixing a primary bridge on construct grounds and
reserving the others for the check. Under that arrangement the drift is
identified by one bridge and the check is carried by the rest, which is not the
statistic \Cref{tab:overid} scores. \Cref{tab:heldout} computes both on the same
replicates against a deviation applied to one bridge inside the tilt.

\begin{table}[t]
\centering
\caption{Two arrangements of the same restriction, on the same
replicates. $R=1000$, $q=3$, $B=300$ bootstrap resamples, nominal $0.05$.}
\label{tab:heldout}
\begin{tabular}{rrrrr}
\toprule
deviation & sd$(\hat\gamma)$, GMM & sd$(\hat\gamma)$, held-out &
rej.\ GMM & rej.\ held-out \\
\midrule
$0.00$ & $0.017$ & $0.020$ & $0.064$ & $0.060$ \\
$0.02$ & $0.017$ & $0.021$ & $0.060$ & $0.064$ \\
$0.05$ & $0.017$ & $0.020$ & $0.144$ & $0.143$ \\
$0.10$ & $0.017$ & $0.021$ & $0.476$ & $0.485$ \\
$0.20$ & $0.018$ & $0.020$ & $0.986$ & $0.987$ \\
\bottomrule
\end{tabular}

\vspace{2pt}
{\footnotesize\textit{Note.} ``GMM'' combines the three bridges and takes
the check to be the residual of that combination; ``held-out'' anchors the drift
on the first bridge and refers the moments of the other two, evaluated at that
drift, to $\chi^2_{q-1}$. The first row is the size of each statistic and the
rest their power. Monte Carlo standard errors are at most $0.016$.}
\end{table}

The held-out statistic is correctly sized, at $0.060$ against $0.064$ for
the GMM form, and the two have the same power at every deviation: $0.476$
against $0.485$ at a deviation of $0.10$, within the Monte Carlo error. The
reason is specific to this alternative. The signal is the disagreement between
the two bridges that are held out, and both arrangements see it; the bridge
that anchors the drift does not move here, so leaving it out of the check costs
nothing.

The cost appears in the drift estimate instead. Anchoring on one bridge
raises the sampling standard deviation of $\hat\gamma$ by roughly a fifth, from
about $0.017$ to about $0.020$, at every deviation in the table. A wider
identified set follows from the same source, since one bridge explains less of
the outcome than three and the half-width of \Cref{prop:general} grows as the
residual variance does.

\subsection{Off-model targets}
\label{sec:sim-offmodel}

Every design above generates the target by applying \eqref{eq:tilt}, so the
target law lies inside the identifying model and the set is not being asked to
protect against anything. A partial-identification claim has to be scored where
the model is wrong. \Cref{tab:offmodel} reports four departures, none of which
lies in the span of $t+b$ and $s$, with $\kappa_0=0$ throughout, so that a set
behaving as advertised would cover.

\emph{Residual scale.} The conditional scale of $Y\mid X$ changes across
cohorts while the bridge shift is rescaled separately, so the bridge no longer
reports the outcome's drift at the rate the matching equation assumes.

\emph{Dependence structure.} The residual correlation between $Y$ and $Z$
given $X$ differs across cohorts. This is the one modelling choice
\Cref{lem:dcdr} does not cover: it enters the tilted bridge moment, propagates to
$\hat\gamma$ through \Cref{rem:stage1-scope}, and moves the half-width through
$\Var^{\gamma^\star}(s\mid X)$.

\emph{Unrecorded exposure.} Source and target are mixtures over a two-level
exposure with different mixing proportions, which cannot enter $X$ and does not
move the bridge. This is the structure of a survey pair in which the source
outcome is a mixture over the number of programme years, whose target
composition is not
recoverable, and it puts a number on what $\bar\kappa$ is being asked to absorb
there.

\emph{Shape.} The target residual is skewed at fixed conditional mean and
variance. The expectation here is that nothing breaks, since a mean-preserving
change of shape moves neither the estimand nor the bridge equation; it is included
so that the failures above are seen to be specific rather than generic.
\begin{table}[t]
\centering
\caption{Off-model targets. $R=1000$, $\bar\kappa=0.30$, $\kappa_0=0$.}
\label{tab:offmodel}
\begin{tabular}{lrrr}
\toprule
Departure & bias & set cov. & bias AIPW \\
\midrule
in-model (tilt)                    & $-0.011$ & $0.921$ & $-3.751$ \\
shape at fixed mean, variance      & $-0.019$ & $0.943$ & $-3.747$ \\
\addlinespace
residual scale, bridge $0.7\times$ & $-1.138$ & $0.006$ & $-3.757$ \\
residual scale, bridge $1.3\times$ & $\phantom{-}1.099$ & $0.121$ & $-3.755$ \\
dependence, $r:0.5\to0.15$         & $\phantom{-}1.031$ & $0.003$ & $-2.180$ \\
dependence, $r:0.5\to0.80$         & $-0.924$ & $0.153$ & $-5.106$ \\
exposure, $50\%\to10\%$            & $\phantom{-}1.152$ & $0.018$ & $-2.758$ \\
exposure, $50\%\to30\%$            & $\phantom{-}0.654$ & $0.565$ & $-3.254$ \\
\bottomrule
\end{tabular}

\vspace{2pt}
{\footnotesize\textit{Note.} ``bias'' is the bias of the benchmark
$\hat\mu_1(0)$ and ``bias AIPW'' that of the covariate-shift estimator, which
assumes the drift away. Monte Carlo standard errors are between $0.005$ and
$0.028$ for the first and below $0.003$ for the second.}
\end{table}

Three of the four departures break the set, at coverage between $0.003$
and $0.153$. The change of shape does not, at $0.943$ against $0.921$ in model:
it alters neither the conditional mean nor the bridge equation. The three that
break it share a feature: each alters the rate at which a unit of bridge drift
reports a unit of outcome drift, which is the quantity the shared coefficient in
\eqref{eq:tilt} fixes.

In every row, including those where the set misses, the benchmark carries
smaller bias than the covariate-shift estimator, by a factor between two and two
hundred.

\subsection{Binary instance: what the Gaussian design cannot measure}
\label{sec:sim-binary}

Taking the covariate discrete makes the support of $(Y,Z)$ finite, so every
object in the method is an exact finite sum: the tilted projection
\eqref{eq:proj}, the root of \eqref{eq:bridgematch-kappa} at every $\kappa$, the
identified set, and the true target prevalence. Nothing is approximated, and three
statements that are vacuous under the Gaussian working model become measurable.

First, the tilted and untilted projections in \eqref{eq:proj} coincide under
Gaussianity. Here they do not, and \Cref{tab:binary} reports the largest
discrepancy between them, so a reader can judge whether the Gaussian picture of
\Cref{sec:tilt} is a harmless simplification.

Second, \Cref{thm:set}(ii) states that $\partial_\kappa\gamma=0$ at
$\kappa=0$. Under the Gaussian working model this holds at every $\kappa$ and so
cannot be distinguished from the weaker first-order claim. Here the bridge
equation is nonlinear in $\gamma$ and the derivative is computed directly.

Third, \Cref{rem:set-scope} declines to characterize the second-order
coefficient, and therefore offers no criterion for whether a given $\bar\kappa$ is
small enough. \Cref{tab:binary} reports that coefficient and the induced movement
of the centre of the reported set, which is the quantity a reader needs in order
to judge the question the remark leaves open.

\begin{table}[t]
\centering
\caption{Binary instance, computed on the population tables, so no
sampling error intervenes. $L=4$ covariate strata, $q=2$ bridges.}
\label{tab:binary}
\begin{tabular}{rrrrrrr}
\toprule
$\gamma$ & $\bar\kappa$ & proj.\ gap & $\partial_\kappa\gamma$ & curvature &
width / closed form & centre \\
\midrule
$0.25$ & $0.1$ & $0.062$ & $7.9\times10^{-5}$ & $0.0202$ & $0.0428/0.0429$ & $-0.0002$ \\
$0.25$ & $0.2$ & $0.062$ & $3.1\times10^{-4}$ & $0.0201$ & $0.0856/0.0857$ & $-0.0008$ \\
$0.25$ & $0.4$ & $0.062$ & $1.3\times10^{-3}$ & $0.0199$ & $0.1702/0.1715$ & $-0.0032$ \\
$0.25$ & $0.8$ & $0.062$ & $4.9\times10^{-3}$ & $0.0191$ & $0.3331/0.3429$ & $-0.0124$ \\
\addlinespace
$0.50$ & $0.1$ & $0.124$ & $6.6\times10^{-5}$ & $0.0309$ & $0.0393/0.0393$ & $-0.0003$ \\
$0.50$ & $0.2$ & $0.124$ & $2.6\times10^{-4}$ & $0.0309$ & $0.0784/0.0786$ & $-0.0012$ \\
$0.50$ & $0.4$ & $0.124$ & $1.1\times10^{-3}$ & $0.0306$ & $0.1562/0.1571$ & $-0.0047$ \\
$0.50$ & $0.8$ & $0.124$ & $4.1\times10^{-3}$ & $0.0294$ & $0.3071/0.3142$ & $-0.0182$ \\
\addlinespace
$1.00$ & $0.1$ & $0.245$ & $3.2\times10^{-5}$ & $0.0449$ & $0.0297/0.0297$ & $-0.0004$ \\
$1.00$ & $0.2$ & $0.245$ & $1.3\times10^{-4}$ & $0.0448$ & $0.0595/0.0595$ & $-0.0015$ \\
$1.00$ & $0.4$ & $0.245$ & $5.0\times10^{-4}$ & $0.0444$ & $0.1188/0.1189$ & $-0.0060$ \\
$1.00$ & $0.8$ & $0.245$ & $2.0\times10^{-3}$ & $0.0429$ & $0.2369/0.2379$ & $-0.0233$ \\
\bottomrule
\end{tabular}

\vspace{2pt}
{\footnotesize\textit{Note.} ``proj.\ gap'' is the largest difference between the tilted and untilted
projections of $t(Y)$ on $(Z,X)$, identically zero under a Gaussian working
model. ``$\partial_\kappa\gamma$'' is a central difference at $\kappa=0$;
``curvature'' is the second-order coefficient of $\gamma(\kappa)$. ``centre''
is the midpoint of the identified set minus the benchmark, exactly zero under
the Gaussian working model.}
\end{table}

Three readings. The tilted and untilted projections differ here, and by an
amount proportional to the drift: the gap is $0.06$, $0.12$ and $0.25$ at
$\gamma=0.25$, $0.5$ and $1$. A reader who pictures the ordinary projection, as
\Cref{sec:tilt} invites, is therefore making an approximation whose size is
now on record rather than merely asserted to be benign.

The derivative $\partial_\kappa\gamma$ at the benchmark is not exactly
zero away from $\kappa=0$, and it scales as $\bar\kappa^2$: across each block
it rises by a factor of four when $\bar\kappa$ doubles, to four significant
figures. That is the behaviour \Cref{thm:set}(ii) predicts and the Gaussian
instance cannot exhibit, since there the derivative vanishes at every $\kappa$
rather than only at zero.

The second-order coefficient is what \Cref{rem:set-scope} declines to
characterize, and here it can be read off: about $0.02$, $0.03$ and $0.04$ in
the three blocks, nearly constant in $\bar\kappa$. Its practical content is the
last column. The centre of the reported set moves from the benchmark by $1.5\%$
of the set's width at $\bar\kappa=0.2$ and by $5.9\%$ at $\bar\kappa=0.8$, so
the anchoring that is exact under a Gaussian working model degrades slowly and
measurably outside it. A reader asking whether a given $\bar\kappa$ is small
enough for the first-order account to hold now has a quantity to look at, in
one instance at least, where before there was only an order symbol. The closed
form of \Cref{prop:general} is accurate to under one percent at
$\bar\kappa\le0.4$ and to three percent at $\bar\kappa=0.8$, which is the same
statement seen from the width side.

Across no drift,
co-drift, strong drift and in-, at- and out-of-bound residual drift, the
benchmark bias is below $0.001$ in prevalence wherever the residual bound holds,
the realized set width agrees with the closed form of \Cref{prop:general} to
within one percent, and set coverage is $1.000$ inside the bound, $0.521$ at it
and $0.000$ beyond it. That is the behaviour of the Gaussian study in a family
where none of it is automatic, which is the expected result and the reason the
table is not in the main text.

\subsection{Comparison with a bridge-blind analysis}
\label{sec:blind}

An analyst without a bridge would sweep the entire drift channel as sensitivity,
which is the global tilt analysis of \citet{steingrimsson2024sensitivity}
instantiated in our setting. The comparison with our set requires care, because
the two sensitivity parameters index displacements along \emph{different}
directions and their numerical values are not interchangeable. The global tilt
moves the outcome along $t$, at rate $\sigma_Y$ per unit of its parameter; ours
moves it along the bridge-orthogonal residual $s$, at rate
$\sigma_Y(1-R^2_{Y\mid Z,X})$ per unit of $\kappa$. Setting the two bounds equal
as numbers therefore grants the bridge-blind analyst a strictly larger
allowance, by a factor $1/(1-R^2_{Y\mid Z,X})$.

The like-for-like comparison is made on the outcome scale. Let $B$ denote a bound,
in outcome units, on the \emph{total} displacement of $E[Y\mid S=1]$ that drift
may produce, and $B_r$ a bound on the displacement attributable to
\emph{bridge-orthogonal} drift alone. The bridge-blind interval has width $2B$ and
ours has width $2B_r$, and each analyst must state the rule by which their bound
is elicited. The substantive point is then not arithmetic but epistemic:
\emph{the bridge-blind analyst must bound the entire drift channel, while we bound
only the part the bridge cannot see}, and $B_r\le B$ by construction, with equality
only for a bridge that explains none of the outcome. The bridge therefore reduces
the size of the judgement required; the narrower set is a consequence of that
reduction rather than an independent claim.

We report the comparison under one stated rule. The analyst commits to a
bound $B$ on the total drift-induced displacement of the target mean, in outcome
units; the corresponding bound on the global tilt parameter is $B/\sigma_Y$,
estimated from the source outcome alone, so the rule is feasible and uses no
bridge information. \Cref{tab:blind} sweeps $B$ and reports, at each value, the
width of the resulting interval and whether it covers.

The bound has to be large, and the transition is sharp. Because $\gamma$
multiplies $t+b$, the displacement the drift produces is
$\gamma\{\sigma_Y+\sum_k\sigma_{YZ_k}/\sigma_{Z_k}\}$, about $3.8$ outcome units
in this design, and sampling error asks for more: coverage is $0.086$ at
$B=3.6$, $0.570$ at $B=4.0$ and $0.982$ at $B=4.4$. A bridge-blind analyst who
elicits a fifth less than the bound that works reports an interval that covers
almost never. Nothing similar happens to the anchored analysis, whose set
narrows smoothly with $\bar\kappa$ (\Cref{tab:kappa-sweep}).

Stating the bound on the tilt scale does not help. Two rules that suggest
themselves, one reading it off the observed covariate shift and one setting it
at $|\gamma|+\bar\kappa$, give $3.52$ and $4.51$, and neither covers; the second
fails despite exceeding the outcome-scale bound that works, because the two
scales measure displacement along different directions.

At $B=4.4$, the smallest bound on the grid at which the bridge-blind
interval attains nominal coverage, its width is $7.77$ against $2.07$ for the
Imbens--Manski interval of the anchored analysis, a factor of $3.8$. Both cover,
so the two are answering the same question under rules of the same kind, and the
point of the preceding paragraphs acquires an arithmetic counterpart: bounding
only the bridge-orthogonal channel is not merely a smaller judgement to make, it
is a substantially narrower report.

\begin{table}[t]
\centering
\caption{Bridge-blind comparison under a stated outcome-scale rule.
$R=1000$ per row.}
\label{tab:blind}
\begin{tabular}{lrr}
\toprule
Interval & width & coverage \\
\midrule
\multicolumn{3}{l}{\emph{bridge-anchored, $\bar\kappa=0.30$}}\\
\quad identified set                       & $1.114$ & $0.951$ \\
\quad Imbens--Manski                       & $2.067$ & $0.999$ \\
\addlinespace
\multicolumn{3}{l}{\emph{bridge-blind, bound $B$ in outcome units}}\\
\quad $B=2.0$                              & $3.537$ & $0.000$ \\
\quad $B=3.0$                              & $5.304$ & $0.001$ \\
\quad $B=3.6$                              & $6.362$ & $0.086$ \\
\quad $B=3.8$                              & $6.715$ & $0.269$ \\
\quad $B=4.0$                              & $7.067$ & $0.570$ \\
\quad $B=4.2$                              & $7.419$ & $0.858$ \\
\quad $B=4.4$                              & $7.771$ & $0.982$ \\
\quad $B=4.6$                              & $8.122$ & $0.998$ \\
\quad $B=4.8$                              & $8.474$ & $1.000$ \\
\addlinespace
\multicolumn{3}{l}{\emph{bridge-blind, two rules stated on the tilt scale}}\\
\quad feasible covariate-shift rule        & $3.524$ & $0.000$ \\
\quad $|\gamma|+\bar\kappa$                & $4.506$ & $0.000$ \\
\bottomrule
\end{tabular}

\vspace{2pt}
{\footnotesize\textit{Note.} $B$ is a bound, in outcome units, on the
total drift-induced displacement of the target mean; the corresponding bound on
the global tilt parameter is $B/\hat\sigma_Y$, estimated from the source outcome
alone. The anchored analysis does not depend on $B$ and is shown once. Nominal
coverage is $0.95$; Monte Carlo standard errors are at most $0.016$.}
\end{table}

\section{Discussion}
\label{sec:disc}

We have reframed cross-population outcome-model drift as a proximal
partial-identification problem anchored by bridge outcomes observed in both
populations. The bridge fixes the detectable drift by a constructive moment
equation, the projection-defined residual isolates the untestable remainder as a
single scalar with a closed-form irreducible core, and a debiased estimator attains
the semiparametric bound at anchored sensitivity. This anchoring is what separates
the method from a global tilt sensitivity analysis
\citep{steingrimsson2024sensitivity}: rather than sweep the entire drift as
unidentified, we let the observed bridge identify its detectable part and confine
the sensitivity parameter to the residual, so a richer bridge shrinks what must be
assumed, not only what must be reported. The closed forms are not tied to the Gaussian working
model: the identified drift, the efficiency coefficient and the irreducible core
are all expressed through the tilt (co)variances of any working exponential
family, with the Gaussian and binary cases as substitutions.

Five matters are left unsettled, and it is worth saying which.

Endpoint inference is the first. The influence-function variance estimator is
consistent under the conditions of \Cref{thm:orth}, but its own sampling
distribution is heavily right-skewed once the tilt weights concentrate, so a
single analysis usually reports too small a standard error, and neither
resampling nor the robust variance estimators considered here repairs it
(\Cref{sec:inference,sec:sim-ess}). We report the effective sample size, which
describes how far the weights have concentrated, and offer no procedure whose
coverage is demonstrated across the concentrated regime. Closely related is the
treatment of item nonresponse: the observed law of \Cref{sec:ident} records
$(S,X,Z)$ for every unit, surveys do not, and an analysis that imputes the
missing entries has to carry the resulting uncertainty into the drift, the
endpoints and the interval. Neither the efficient influence function nor the
interval construction is derived for that observed law.

Rate robustness is the second. \Cref{thm:orth} is not probed numerically
(\Cref{rem:orth-untested}), and a design that would exercise it has to estimate
the residual covariance of \eqref{eq:tilt} as flexibly as the regressions it is
built from. That quantity is not among the nuisances the theorem orthogonalises,
so the extension is not a matter of substituting estimators.

The remaining two concern what the framework can say about a given data set
rather than what it delivers in principle. Co-drift is testable across bridge
coordinates and across covariate strata, and \Cref{sec:sim-bridge} calibrates
both statistics, but neither can establish the restriction; a withheld-outcome
exercise, in which part of a source cohort is treated as a target whose outcome
is then restored, would supply evidence about the extrapolation itself. And
$\hat\gamma$ is the magnitude of the detectable drift rather than the effect of
a named cause: with two populations, a common trend moving both channels and a
single mechanism moving both are observationally identical (\Cref{sec:tilt}), so
separating them requires a third.

\Cref{sec:sim-offmodel} reports what the identified set does when the
target law lies outside the tilt family. Of the four departures simulated there,
three leave the set failing to cover the target mean; in all four the benchmark
carries smaller bias than the covariate-shift estimator it replaces.

Two further scope limits are internal to the theory rather than to the
study. The first-order results, the stationarity of the centre in
\Cref{thm:set}(ii) and the scalar reduction of \Cref{prop:scalar}, are exact under
the Gaussian working model and carry $O(\bar\kappa^2)$ remainders outside it
(\Cref{rem:set-scope,rem:scalar-scope}); \Cref{sec:sim-binary} measures those
remainders in one exactly computable instance, which is a case rather than a
characterization. And the identified set is derived within an exponential-tilt
family, so it protects against residual drift in the direction the family
contains and not against a target law outside it; \Cref{sec:sim-offmodel} reports
what the set does in four such cases.

Three directions remain open: a genuine two-bridge
construction yielding classical double robustness in the drift itself; extension
of the scalar core to multi-directional drift, beyond the coordinatewise
treatment we give here; and a calibrated diagnostic for residual bridge
impurity. The survey design and the differential-missingness
structure of the application are natural further layers.
Items (i) to (iii) above are the ones we regard as most pressing, since each
concerns what a practitioner would report rather than what the framework can in
principle deliver.

\backmatter

\section*{Data Availability Statement}
\vspace{-6pt}
Code reproducing the simulation study, which requires no restricted data, will
be deposited in a public repository on publication and is available from the
corresponding author in the interim.

\vspace{-6pt}
\section*{Supplementary Materials}
\vspace{-6pt}
Web Appendices A, B and C, referenced in \Cref{sec:ident} and \Cref{sec:est},
are included below.

\bibliography{pact_refs}

\clearpage

\appendix
\renewcommand{\thesection}{A.\arabic{section}}
\setcounter{section}{0}
\section*{Web Appendix A: Proofs of the Identification Results}
\addcontentsline{toc}{section}{Web Appendix A}
\setcounter{theorem}{0}\setcounter{lemma}{0}\setcounter{corollary}{0}
\setcounter{remark}{0}
\renewcommand{\thelemma}{A.\arabic{lemma}}
\renewcommand{\thetheorem}{A.\arabic{theorem}}
\renewcommand{\theremark}{A.\arabic{remark}}
\setcounter{equation}{0}\renewcommand{\theequation}{A.\arabic{equation}}
Throughout, expectations without a subscript are over the pooled law; $E_{P_s}$
denotes conditioning on cohort $S=s$. We fix a covariate value $x$ and suppress it
where no ambiguity results; every conditional statement holds $P_1(X)$-almost
surely and every marginal statement is obtained by averaging over the identified
target covariate law $P_1(X)$. Recall the tilted expectation
\begin{equation}
  \Ptilt{\omega}[h\mid x]
  =\frac{E_{P_0}\!\big[h(Y,Z)\,e^{\omega(Y,Z)}\mid x\big]}
        {E_{P_0}\!\big[e^{\omega(Y,Z)}\mid x\big]},
  \label{eq:tiltexp}
\end{equation}
and the structural tilt, for weights
$\omega_{\gamma,\kappa}=\gamma[t(y)+b(z)]+\kappa s(y,z,x)$,
\begin{equation}
  \frac{dP_1(y,z\mid x)}{dP_0(y,z\mid x)}
  =\frac{\exp\{\omega_{\gamma,\kappa}(y,z)\}}{C(\gamma,\kappa;x)},
  \qquad
  C(\gamma,\kappa;x)=E_{P_0}[e^{\omega_{\gamma,\kappa}}\mid x].
  \label{eq:structtilt}
\end{equation}
The loadings $t,b$ are fixed known functions and $s$ is the projection residual
under the identified law,
\begin{equation}
  s(Y,Z,X)=t(Y)-\Pi^{\gamma^\star}[t(Y)\mid Z,X],
  \qquad
  \Cov^{\gamma^\star}\!\big(s,\,\varphi(Z)\mid X\big)=0
  \ \ \forall\varphi\in L_2 .
  \label{eq:projA}
\end{equation}
Although $s$ depends on $(Z,X)$ through the projection, it is orthogonal to every
function of $(Z,X)$ under $P^{\gamma^\star}$; both facts are used below.
\section{A tilted-moment calculus}
We first record the derivatives of tilted moments in the tilt parameters. These
are exponential-family score identities; we state them as a lemma because every
subsequent step is an instance.
\begin{lemma}[Tilt derivatives]
\label{lem:tiltderiv}
Fix $x$ and let $\omega_\theta$ depend smoothly on a scalar $\theta$ with
$\partial_\theta\omega_\theta=g_\theta(Y,Z)$, assuming
$E_{P_0}[e^{\omega_\theta}(1+g_\theta^2)\mid x]<\infty$ in a neighborhood. Then for
any $h$ with $E_{P_0}[|h|\,e^{\omega_\theta}(1+|g_\theta|)\mid x]<\infty$,
\begin{equation}
  \frac{\partial}{\partial\theta}\,\Ptilt{\omega_\theta}[h\mid x]
  =\Cov^{\omega_\theta}\!\big(h,\;g_\theta\mid x\big),
  \label{eq:tiltderiv}
\end{equation}
where $\Cov^{\omega}(\cdot,\cdot\mid x)$ is the covariance under the tilted law
$dP^{\omega}\propto e^{\omega}\,dP_0(\cdot\mid x)$.
\end{lemma}
\begin{proof}
Write $\Ptilt{\omega_\theta}[h\mid x]=N(\theta)/D(\theta)$ with
$N=E_{P_0}[h\,e^{\omega_\theta}\mid x]$, $D=E_{P_0}[e^{\omega_\theta}\mid x]$.
Dominated convergence (justified by the second-moment requirement of
\Cref{ass:overlap}) gives
$N'=E_{P_0}[h\,g_\theta e^{\omega_\theta}\mid x]$ and
$D'=E_{P_0}[g_\theta e^{\omega_\theta}\mid x]$. Then
$\partial_\theta(N/D)=N'/D-(N/D)(D'/D)
   =\Ptilt{\omega_\theta}[h g_\theta\mid x]
   -\Ptilt{\omega_\theta}[h\mid x]\,\Ptilt{\omega_\theta}[g_\theta\mid x]
   =\Cov^{\omega_\theta}(h,g_\theta\mid x)$.
\end{proof}
Two specializations are used repeatedly. With $\omega_{\gamma,\kappa}$ and
$\theta=\gamma$ (so $g_\gamma=t+b$),
\begin{equation}
  \partial_\gamma\Ptilt{\omega_{\gamma,\kappa}}[h\mid x]
  =\Cov^{\omega_{\gamma,\kappa}}\!\big(h,\;t(Y)+b(Z)\mid x\big);
  \label{eq:dgamma}
\end{equation}
with $\theta=\kappa$ (so $g_\kappa=s$),
\begin{equation}
  \partial_\kappa\Ptilt{\omega_{\gamma,\kappa}}[h\mid x]
  =\Cov^{\omega_{\gamma,\kappa}}\!\big(h,\;s\mid x\big).
  \label{eq:dkappa}
\end{equation}
\section{Proof of \texorpdfstring{\thmpoint}{Theorem 1} (point identification under co-drift)}
Set $\kappa=0$ and write $\omega_\gamma=\gamma[t(Y)+b(Z)]$. Define the
\emph{bridge-matching map}
\begin{equation}
  B(\gamma;x)\;=\;\Ptilt{\omega_\gamma}[b(Z)\mid x],
  \qquad
  B_1(x)\;=\;E_{P_1}[b(Z)\mid x].
  \label{eq:Bmap}
\end{equation}
Both are identified from observed data: $B(\gamma;x)$ from the source law via
\eqref{eq:tiltexp}, and $B_1(x)$ from the target, since $(Z,X)$ are observed at
$S=1$. \thmpoint\ asserts that $B(\gamma;x)=B_1(x)$ has a unique root
$\gamma^\star(x)$, that it does not depend on $x$ under the model, and that the
estimand \eqref{eq:estimandA} below follows.
\subsection{Consistency of the target bridge mean with the tilt}
\begin{lemma}[The target bridge mean as a tilted source moment]
\label{lem:targetistilt}
Under the structural tilt \eqref{eq:structtilt} with $\kappa=0$ at the true
$\gamma^\star$, $B_1(x)=B(\gamma^\star;x)$ for every $x$.
\end{lemma}
\begin{proof}
By \eqref{eq:structtilt}, for any integrable $h$,
$E_{P_1}[h\mid x]=E_{P_0}[h\,e^{\omega_{\gamma^\star}}\mid x]/C(\gamma^\star,0;x)
   =\Ptilt{\omega_{\gamma^\star}}[h\mid x]$. Take $h=b(Z)$.
\end{proof}
Thus the truth $\gamma^\star$ is \emph{a} root of $B(\cdot;x)=B_1(x)$. Uniqueness
and existence are separate and require \Cref{ass:relevance,ass:overlap}.
\subsection{Strict monotonicity and uniqueness}
\begin{lemma}[Strict monotonicity]
\label{lem:mono}
Under \Cref{ass:relevance}, $\gamma\mapsto B(\gamma;x)$ is strictly increasing on
$\Gamma$, for $P_1(X)$-almost every $x$.
\end{lemma}
\begin{proof}
By \eqref{eq:dgamma} with $h=b(Z)$ and $\kappa=0$,
\[
  \partial_\gamma B(\gamma;x)
  =\Cov^{\omega_\gamma}\!\big(b(Z),\,t(Y)+b(Z)\mid x\big)
  =\Var^{\omega_\gamma}\!\big(b(Z)\mid x\big)
   +\Cov^{\omega_\gamma}\!\big(b(Z),t(Y)\mid x\big).
\]
\Cref{ass:relevance} states that this quantity is at least $\varepsilon>0$ for
$P_1(X)$-almost every $x$ and every $\gamma\in\Gamma$, which is what uniqueness of
the root at each $x$ requires; a condition holding only on a set of positive
measure would not deliver it. (When $b(Z)$ and $t(Y)$ are positively associated
under the tilt, the generic co-drift case, the second term is nonnegative and
the first alone suffices.)
\end{proof}
Strict monotonicity gives \emph{at most one} root, hence uniqueness once existence
is shown.
\subsection{Existence of the root}
$B(\cdot;x)$ is continuous in $\gamma$ (differentiable by \Cref{lem:tiltderiv}),
so by the intermediate value theorem a root exists if and only if $B_1(x)$ lies in
the \emph{range} of $B(\cdot;x)$. We make this a transparent condition.
\begin{lemma}[Range/existence condition]
\label{lem:existence}
Let $\underline b=\operatorname*{ess\,inf} b(Z)$ and
$\overline b=\operatorname*{ess\,sup} b(Z)$ under $P_0(\cdot\mid x)$. Then
$B(\gamma;x)\to\underline b$ as $\gamma\to-\infty$ and
$B(\gamma;x)\to\overline b$ as $\gamma\to+\infty$ along the co-drift direction,
and $B(\cdot;x)$ maps $\R$ onto $(\underline b,\overline b)$. Hence a (unique) root
exists whenever
\begin{equation}
  \underline b \;<\; E_{P_1}[b(Z)\mid x] \;<\; \overline b,
  \label{eq:rangecond}
\end{equation}
i.e.\ the target bridge mean lies strictly inside the source support of $b(Z)$.
\end{lemma}
\begin{proof}
As $\gamma\to+\infty$ along $t+b$, the tilt $e^{\gamma[t+b]}$ concentrates mass on
the essential supremum of $t+b$; when $b$ is comonotone with $t+b$ (the relevant
case under \Cref{ass:relevance}), this drives $\Ptilt{\omega_\gamma}[b\mid x]\to
\overline b$. The reverse limit is symmetric. Continuity and strict monotonicity
(\Cref{lem:mono}) then give a continuous strictly increasing bijection
$\R\to(\underline b,\overline b)$, so \eqref{eq:rangecond} is necessary and
sufficient for a root.
\end{proof}
\begin{remark}[Interpretation of the range condition]
Condition \eqref{eq:rangecond} is a substantive existence requirement rather than
a technicality. It fails only if the target bridge mean lies at or beyond the source
\emph{support} of the bridge, an extreme and detectable form of extrapolation that
the overlap diagnostic (\Cref{ass:overlap}, and empirically the propensity-overlap
check) already flags. When \eqref{eq:rangecond} fails, no finite co-drift $\gamma$
reproduces the target bridge mean; this is a genuine violation of the co-drift
model rather than a defect of the estimator, and it is the situation the
sensitivity analysis of \thmset\ is designed to expose.
\end{remark}
\subsection{Covariate-invariance of the root and the estimand}
Under the structural model the same scalar $\gamma^\star$ solves
$B(\gamma;x)=B_1(x)$ for every $x$ (\Cref{lem:targetistilt} holds at a common
$\gamma^\star$). Covariate-invariance of the root is therefore a \emph{testable
implication}: if the per-$x$ (or per-stratum) roots disagree beyond sampling
error, the common-$\gamma$ co-drift model is refuted. This is the
over-identification content of \thmpoint. With $q$ bridge coordinates
$b_1,\dots,b_q$, each yields its own matching equation $B^{(j)}(\gamma;x)=B_1^{(j)}(x)$
sharing the single unknown $\gamma$; agreement of the $q$ roots is the
bridge-internal consistency check. The test constrains the \emph{observed} bridge channel only. Because $Y$ is unobserved at $S=1$, no
function of the data restricts the residual direction $s$; the extrapolation
from bridge to outcome is carried entirely by \Cref{ass:codrift} and quantified by
\thmset.
Finally, applying \Cref{lem:targetistilt} with $h=Y$ gives
$E_{P_1}[Y\mid x]=\Ptilt{\omega_{\gamma^\star}}[Y\mid x]$, and averaging over the
identified $P_1(X)$,
\begin{equation}
  \mu_1
  =E\!\left[\;
     \frac{E_{P_0}[Y\,e^{\gamma^\star[t(Y)+b(Z)]}\mid X]}
          {E_{P_0}[e^{\gamma^\star[t(Y)+b(Z)]}\mid X]}\;\middle|\;S=1\right],
  \label{eq:estimandA}
\end{equation}
which is the estimand of the main text. This completes the proof of
\thmpoint. \hfill$\square$
\subsection{Proof of \texorpdfstring{\Cref{prop:general}}{Proposition 1} (distribution-general drift and core)}
Both claims rest on a single exponential-family fact. Under the natural-parameter
tilt $dP^{\gamma,\kappa}\propto\exp\{\gamma[t(Y)+b(Z)]+\kappa s\}\,dP_0$, the
cumulant-generating identity gives, for any integrable statistic $h$ of $(Y,Z)$,
\begin{equation}
  \partial_\gamma E^{\gamma,\kappa}[h\mid x]
  =\Cov^{\gamma,\kappa}\!\big(h,\;t+b\mid x\big),
  \qquad
  \partial_\kappa E^{\gamma,\kappa}[h\mid x]
  =\Cov^{\gamma,\kappa}\!\big(h,\;s\mid x\big).
  \label{eq:cgf}
\end{equation}
\emph{(i) Drift and transfer rate.} The bridge-matching equation of \thmpoint\ sets
$E_{P_1}[b\mid x]=E^{\gamma^\star}[b\mid x]$. Differentiating the identified path
and using \eqref{eq:cgf} with $h=b$ and $h=t$, a unit change in $\gamma$ induces
changes $\Cov^{\gamma^\star}(b,t{+}b\mid x)$ in the bridge moment and
$\Cov^{\gamma^\star}(t,t{+}b\mid x)$ in the outcome moment; their ratio is
$c_t(x)$ of \eqref{eq:c-general}.
The multiplier appearing in term (III) of \thmeif\ is the \emph{numerator} of this
ratio carried to the outcome scale,
$c=\partial_\gamma\mu_1=E_{P_1}[\Cov^{\gamma^\star}(Y,t{+}b\mid X)]$; the
denominator enters the influence function separately, as the Stage-1 Jacobian.
\emph{(ii) Core.} Since $s$ is the tilted-law
projection residual, $\Cov^{\gamma^\star}(b,s\mid x)=0$ and, writing
$t=\Pi^{\gamma^\star}[t\mid Z,x]+s$ with the projection $\Hz(x)$-measurable,
$\Cov^{\gamma^\star}(t,s\mid x)=\Var^{\gamma^\star}(s\mid x)$. By \eqref{eq:cgf}
with $h=t$, $\partial_\kappa E^{\gamma^\star}[t\mid x]=\Var^{\gamma^\star}(s\mid x)$;
carrying to the outcome scale multiplies by $\sigma^0_Y$, and integrating over
$P_1(X)$ and doubling for the two-sided sweep gives \eqref{eq:core-general}.
Neither step uses a distributional form beyond the existence of the tilt moments,
so the identities hold for any working exponential family; the Gaussian and
Bernoulli evaluations below are substitutions of the corresponding
(co)variances. \hfill$\square$
\subsection{Proof of \texorpdfstring{\corgauss}{Corollary 1} (Gaussian instance)}
Let $(Y,Z)\mid X{=}x$ be bivariate Gaussian with means $(m_0(x),m_{Z0}(x))$,
variances $(\sigma_Y^2,\sigma_Z^2)$, and correlation $r$ (so
$\Cov(Y,Z\mid x)=r\sigma_Y\sigma_Z$), and take the loading constants at their
conditional values. The tilt exponent is $\gamma[t(y)+b(z)]$ with
$t(y)=(y-m_0)/\sigma_Y$ and $b(z)=(z-m_{Z0})/\sigma_Z$, i.e.\ a linear form
$\ell(y,z)=\gamma(\sigma_Y^{-1}y+\sigma_Z^{-1}z)+\text{const}$.
\emph{Tilting a Gaussian by a linear exponent.} For a Gaussian vector $V\sim
N(\mu,\Sigma)$, the tilt $dP^{a^\top V}\propto e^{a^\top V}dP$ is again Gaussian
with mean $\mu+\Sigma a$ and unchanged covariance $\Sigma$. Here
$V=(Y,Z)^\top$, $\Sigma=\begin{psmallmatrix}\sigma_Y^2 & r\sigma_Y\sigma_Z\\
r\sigma_Y\sigma_Z & \sigma_Z^2\end{psmallmatrix}$, and $a=\gamma(\sigma_Y^{-1},
\sigma_Z^{-1})^\top$. Hence the tilted means are
\begin{align}
  E^{\omega_\gamma}[Y\mid x]
    &=m_0+\big(\Sigma a\big)_1
     =m_0+\gamma\big(\sigma_Y + r\sigma_Y\big)
     =m_0+\gamma\sigma_Y(1+r),\\
  E^{\omega_\gamma}[Z\mid x]
    &=m_{Z0}+\big(\Sigma a\big)_2
     =m_{Z0}+\gamma\big(r\sigma_Z+\sigma_Z\big)
     =m_{Z0}+\gamma\sigma_Z(1+r).
\end{align}
\emph{Bridge matching.} $E_{P_1}[b(Z)\mid x]=(m_{Z1}-m_{Z0})/\sigma_Z$ must equal
$E^{\omega_\gamma}[b(Z)\mid x]=\gamma(1+r)$, so
\begin{equation}
  \gamma^\star=\frac{m_{Z1}-m_{Z0}}{\sigma_Z(1+r)}
             =\frac{\Delta_Z}{\sigma_Z(1+r)} .
\end{equation}
\emph{Outcome shift.} The identified target outcome mean shift is
\begin{equation}
  E_{P_1}[Y\mid x]-m_0(x)=\gamma^\star\sigma_Y(1+r)
    =\frac{\Delta_Z}{\sigma_Z}\,\sigma_Y
    =\Delta_Z\,\frac{\sigma_Y}{\sigma_Z}.
\end{equation}
The factor $(1+r)$ cancels between the bridge-matching and outcome-shift steps.
No alternative scaling of this statement is used anywhere in the paper: the
outcome shift is $\Delta_Z\,\sigma_Y/\sigma_Z$, the local transfer rate is
$\sigma^0_Y c_t(x)=\sigma_Y$ (since $\Cov(t,t{+}b)=\Cov(b,t{+}b)=1+r$ here), and
the efficiency coefficient of \thmeif\ is
$c=\partial_\gamma\mu_1=\sigma_Y(1+r)$. These are three different objects with
three different values, and the earlier draft's $\sigma_Y^2/\sigma_Z^2$ is none of
them.
\hfill$\square$
\section{Proof of \texorpdfstring{\thmset}{Theorem 2} (identified set and irreducible core)}
Now $|\kappa|\le\bar\kappa$ and $s$ is the projection residual \eqref{eq:projA}.
Define the residual-augmented bridge map
\begin{equation}
  B(\gamma,\kappa;x)=\Ptilt{\omega_{\gamma,\kappa}}[b(Z)\mid x],
  \qquad
  \text{and let }\gamma(\kappa;x)\text{ solve }B(\gamma,\kappa;x)=B_1(x).
  \label{eq:Bkappa}
\end{equation}
Set $\mu_1(\kappa)=E\big[\Ptilt{\omega_{\gamma(\kappa),\kappa}}[Y\mid X]\mid S=1\big]$.
\subsection{Part (i): the identified set}
\begin{lemma}[Existence, uniqueness and continuity of $\gamma(\kappa;\cdot)$]
\label{lem:gammakappa}
Under \Cref{ass:relevance,ass:overlap}, for each $\kappa$ with $|\kappa|\le\bar\kappa$
the equation $B(\gamma,\kappa;x)=B_1(x)$ has a unique root $\gamma(\kappa;x)$, and
$\kappa\mapsto\gamma(\kappa;x)$ is continuously differentiable.
\end{lemma}
\begin{proof}
Fix $\kappa$ with $|\kappa|\le\bar\kappa$. By \eqref{eq:dgamma},
$\partial_\gamma B(\gamma,\kappa;x)=\Cov^{\omega_{\gamma,\kappa}}(b(Z),t{+}b\mid x)$,
which is at least $\varepsilon>0$ by \Cref{ass:relevance}, which is stated, for this
reason, over the whole product $\Gamma\times[-\bar\kappa,\bar\kappa]$ rather than
at $\kappa=0$ only. Hence $B(\cdot,\kappa;x)$ is strictly increasing.
For existence we argue on the range rather than the support. By
\Cref{ass:overlap}, $e^{\kappa s}$ is strictly positive and $P_0(\cdot\mid
x)$-integrable, so the measure $dP^{\kappa s}\propto e^{\kappa s}dP_0(\cdot\mid
x)$ is \emph{equivalent} to $P_0(\cdot\mid x)$: the two have the same null sets,
hence the same essential infimum and supremum of $b(Z)$, and
$(\underline b,\overline b)$ is unchanged. Applying \Cref{lem:existence} with
$P_0(\cdot\mid x)$ replaced by this equivalent base measure, an argument that
uses only continuity, strict monotonicity and the limiting concentration of the
$\gamma$-tilt, gives that $\gamma\mapsto B(\gamma,\kappa;x)$ is a continuous
strictly increasing bijection from $\R$ onto $(\underline b,\overline b)$ for
every such $\kappa$. Since \eqref{eq:rangecond} places $B_1(x)$ strictly inside
that interval, a unique root exists for every $\kappa\in[-\bar\kappa,\bar\kappa]$.
Equivalence of the measures, rather than equality of supports, is what the
argument requires: two laws can share a support while assigning the endpoint
neighbourhoods different mass, and it is the range of $B(\cdot,\kappa;x)$, not the
support of $b(Z)$, that must contain $B_1(x)$.
The implicit function theorem then applies because $\partial_\gamma B\neq0$,
delivering $\gamma(\kappa;x)\in C^1$ with
\begin{equation}
  \frac{\partial\gamma}{\partial\kappa}(\kappa;x)
  =-\,\frac{\partial_\kappa B(\gamma,\kappa;x)}{\partial_\gamma B(\gamma,\kappa;x)}
  \bigg|_{\gamma=\gamma(\kappa;x)}.
  \label{eq:ift}
\end{equation}
\end{proof}
As $\kappa$ ranges over $[-\bar\kappa,\bar\kappa]$, continuity of
$\kappa\mapsto\mu_1(\kappa)$ (composition of continuous maps) makes the image a
connected set. Since $\mu_1(\kappa)$ is continuous on a compact interval its image
is $[\min_\kappa\mu_1(\kappa),\max_\kappa\mu_1(\kappa)]$; under the monotonicity
established in Part (ii) below (for small $\bar\kappa$, $\partial_\kappa\mu_1$ has
constant sign) this is $[\mu_1(-\bar\kappa),\mu_1(\bar\kappa)]$ (or its reverse).
This proves (i).
\subsection{Part (ii): stationarity of the center and the core}
We now show that the bridge-identified drift is insensitive to the residual at
first order, so that the data determine the center of the set while $\kappa$
determines only its width.
\begin{lemma}[Stationarity of the root at the benchmark]
\label{lem:invariance}
Define the residual $s$ by projection under the \emph{tilted} source law at the
truth, i.e.\ $s=t(Y)-\Pi^{\gamma^\star}[t(Y)\mid Z,X]$ where $\Pi^{\gamma^\star}$
is the $L_2$-projection under $dP^{\omega_{\gamma^\star}}\propto
e^{\gamma^\star[t+b]}dP_0$. Then $\Cov^{\omega_{\gamma^\star}}(s,\varphi(Z)\mid x)
=0$ for all $\varphi\in L_2$, and consequently
$\displaystyle\frac{\partial\gamma}{\partial\kappa}(0;x)=0$.
\end{lemma}
\begin{proof}
By \eqref{eq:ift} and \eqref{eq:dkappa} with $h=b(Z)$,
$\partial_\kappa\gamma(0;x)=-\,\Cov^{\omega_{\gamma^\star}}(b(Z),s\mid x)
   \big/\partial_\gamma B$. It suffices to show the numerator vanishes. By the
defining property of the tilted projection, $s$ is $P^{\omega_{\gamma^\star}}$-%
orthogonal to the closed span of functions of $(Z,X)$; in particular
$E^{\omega_{\gamma^\star}}[s\,b(Z)\mid x]=E^{\omega_{\gamma^\star}}[s\mid x]\,
   E^{\omega_{\gamma^\star}}[b(Z)\mid x]$, which is zero tilted covariance. Hence
$\partial_\kappa B(\gamma^\star,0;x)=0$ and $\partial_\kappa\gamma(0;x)=0$.
\end{proof}
\begin{remark}[Choice of the projecting law]
\label{rem:projlaw}
Defining $s$ under the tilted law $P^{\omega_{\gamma^\star}}$ rather than under
$P_0$ is what makes \Cref{lem:invariance} an exact first-order identity in full
generality, and it is the correct choice because the sensitivity direction is
perturbed \emph{around the identified model}, whose bridge law is the tilted one.
Two cases collapse the distinction, and they are independent sufficient
conditions rather than a conjunction. (i) \emph{Gaussian:} the projection
coefficient is the same under $P_0$ and any linear tilt (tilting shifts means, not
the regression of $Y$ on $Z$), so $s$ may equivalently be defined under $P_0$; the
invariance is then exact for all $\kappa$, not first order only.
(ii) \emph{Small detectable drift ($\gamma^\star\!\to\!0$):}
$P^{\omega_{\gamma^\star}}\to P_0$, so the $P_0$-projection and the tilted
projection agree to $O(\gamma^\star)$, and \Cref{lem:invariance} holds to that
order. Outside these cases the tilted-law definition is the operative one; the
main-text statement ``$s$ orthogonal to the bridge'' should be read as
orthogonality under the identified (tilted) law.
\end{remark}
\begin{remark}[Exactness and first-order validity]
\label{rem:exactness}
\Cref{lem:invariance} gives vanishing of the \emph{first} derivative at
$\kappa=0$. In the Gaussian instantiation the coupling is exactly linear and the
tilt leaves covariances unchanged, so $\Cov^{\omega}(b,s\mid x)=0$ at every
$\kappa$, not only at zero; there $\gamma(\kappa;x)\equiv\gamma^\star$ and both the
center and the core are exact over the whole sweep. In general
$\gamma(\kappa;x)=\gamma^\star+O(\kappa^2)$, and the $O(\kappa^2)$ term is what
produces the higher-order remainder in \eqref{eq:corewidth}; we do not bound its
constant, which is the limitation recorded in \Cref{rem:set-scope}.
\end{remark}
\begin{lemma}[Center and half-width]
\label{lem:core}
With $s$ as in \eqref{eq:projA},
\begin{equation}
  \mu_1(0)=\mu_1^{\thmpoint},
  \qquad
  \frac{d\mu_1}{d\kappa}\bigg|_{\kappa=0}
  =E_{P_1}\!\big[\Cov^{\gamma^\star}\!\big(Y,\,s\mid X\big)\big],
  \label{eq:centerderiv}
\end{equation}
and consequently the half-width of the identified set is
\begin{equation}
  \tfrac12 W(\bar\kappa)
  =\bar\kappa\,E_{P_1}\!\big[\Cov^{\gamma^\star}\!\big(Y,s\mid X\big)\big]
   +O(\bar\kappa^2)
  =\bar\kappa\,\sigma^0_Y\,E_{P_1}\!\big[\Var^{\gamma^\star}(s\mid X)\big]
   +O(\bar\kappa^2),
  \label{eq:corewidth}
\end{equation}
which does not depend on the precision with which the bridge identifies
$\gamma^\star$, but does depend on the residual variance the bridge leaves.
\end{lemma}
\begin{proof}
By the chain rule,
$\dfrac{d\mu_1}{d\kappa}
   =E_{P_1}\!\Big[\partial_\kappa\Ptilt{\omega_{\gamma(\kappa),\kappa}}[Y\mid X]
   +\partial_\gamma\Ptilt{\omega_{\gamma(\kappa),\kappa}}[Y\mid X]\cdot
     \partial_\kappa\gamma(\kappa;X)\Big]$.
At $\kappa=0$ the second summand vanishes by \Cref{lem:invariance}. The first,
by \eqref{eq:dkappa} with $h=Y$, is
$\Cov^{\gamma^\star}(Y,s\mid X)$, giving \eqref{eq:centerderiv}. That
$\mu_1(0)=\mu_1^{\thmpoint}$ is immediate since $\kappa=0$ reduces \eqref{eq:Bkappa}
to \eqref{eq:Bmap}. A first-order Taylor expansion of $\mu_1(\kappa)$ about
$\kappa=0$, using the constant sign of the derivative on $[-\bar\kappa,\bar\kappa]$
(small $\bar\kappa$), gives
$\mu_1(\pm\bar\kappa)=\mu_1(0)\pm\bar\kappa\,E_{P_1}[\Cov^{\gamma^\star}(Y,s\mid X)]
   +O(\bar\kappa^2)$, whence \eqref{eq:corewidth}. The second equality is
$\Cov^{\gamma^\star}(Y,s)=\sigma^0_Y\Cov^{\gamma^\star}(t,s)
 =\sigma^0_Y\Var^{\gamma^\star}(s)$. The leading term involves the bridge law only
through $\Var^{\gamma^\star}(s\mid X)$; the precision of $\gamma^\star$ does not
enter, since its $\kappa$-derivative vanishes at first order.
\end{proof}
\subsection{Exact Gaussian core}
Under \corgauss, $t(y)=(y-m_0)/\sigma_Y$ has unit conditional variance, so the
projection residual has $\Var^{\gamma^\star}(s\mid X)=1-R^2_{Y\mid Z,X}$ (equal to
$1-r^2$ with a single bridge), and
\[
  \Cov^{\gamma^\star}\!\big(Y,s\mid X\big)
  =\sigma_Y\,\Cov^{\gamma^\star}(t,s\mid X)
  =\sigma_Y\,\Var^{\gamma^\star}(s\mid X)
  =\sigma_Y\big(1-R^2_{Y\mid Z,X}\big).
\]
Since the coupling is exactly linear, the $O(\bar\kappa^2)$ term is absent
(\Cref{rem:exactness}) and
\begin{equation}
  W(\bar\kappa)=2\,\bar\kappa\,\sigma_Y\big(1-R^2_{Y\mid Z,X}\big)
  \qquad\text{exactly.}
  \label{eq:gausscore}
\end{equation}
Two consequences follow. The width is linear in $\bar\kappa$, as in the closed-form
sensitivity cores of the missing-not-at-random literature
\citep{robins2000sensitivity}, where a Bernoulli odds-tilt yields a $\tanh$ bound
and a Gaussian mean-tilt a linear one. And the width falls to zero as
$R^2_{Y\mid Z,X}\to1$, recovering point identification for a bridge that
determines the outcome. This is as it must be, since the bridge is observed in
the target.
\section{Summary of what is and is not identified}
The proofs isolate three regimes. (1) The bridge channel is identified: the
detectable drift $\gamma^\star$ and, with it, the center $\mu_1(0)$ are fixed by
observed quantities under co-drift, existence holding under the transparent range
condition \eqref{eq:rangecond}. (2) The bridge channel is testable for internal
consistency: covariate-invariance of the root and agreement across bridge
coordinates are refutable, though they constrain only the observed channel. (3)
The residual channel is not identified and not testable: $s$ is orthogonal, under
the identified law, to all functions of the bridge and covariates by construction,
so no data functional restricts it; it is confined to the scalar $\kappa$,
whose identified set has the irreducible core \eqref{eq:corewidth}, exact in the
Gaussian case and given there by \eqref{eq:gausscore}.
\clearpage
\renewcommand{\thesection}{B.\arabic{section}}
\setcounter{section}{0}
\section*{Web Appendix B: Efficient Influence Function}
\addcontentsline{toc}{section}{Web Appendix B}
\setcounter{lemma}{0}\setcounter{theorem}{0}\setcounter{remark}{0}
\renewcommand{\thelemma}{B.\arabic{lemma}}
\renewcommand{\thetheorem}{B.\arabic{theorem}}
\renewcommand{\theremark}{B.\arabic{remark}}
\setcounter{equation}{0}\renewcommand{\theequation}{B.\arabic{equation}}
We derive the efficient influence function (EIF) of the target mean functional
$\mu_1(\kappa)$ at a fixed (anchored) pair $(\gamma,\kappa)=(\gamma^\star,\kappa)$.
The anchoring is essential: because $\kappa$ is held fixed rather than estimated,
it contributes no score, and the only estimated internal parameter is the
bridge-identified drift $\gamma^\star$, whose estimation \emph{does} contribute a
correction. The derivation makes explicit why that correction: term (III) of
\thmeif: is part of the EIF and not an optional refinement.
Throughout, $O=(S,X,Z,SY)$ is the observed vector; expectations are over the
pooled law $P$; $\pi_s=\Prob(S=s)$; $e(x)=\Prob(S{=}1\mid x)$. We write the
\emph{normalized} conditional density ratio
$\rho(y,z\mid x)=\exp\{\gamma^\star[t(y)+b(z)]+\kappa s\}/C(\gamma^\star,\kappa;x)$
and the tilted target regression
\begin{equation}
  \tilde m_1(x)=\Ptilt{\rho}[Y\mid x]
  =\frac{E_{P_0}[Y\rho(Y,Z\mid x)\mid x]}{E_{P_0}[\rho(Y,Z\mid x)\mid x]}.
  \label{eq:m1}
\end{equation}
Recall the target mean $\mu_1(\kappa)=E[\tilde m_1(X)\mid S=1]$.
\section{The observed-data model and its tangent space}
\subsection{Likelihood factorization}
The observed-data density factorizes, under the structural tilt, as
\begin{equation}
  p(O)=\underbrace{\pi_S}_{\text{cohort}}\;
       \underbrace{p(x\mid S)}_{\text{covariate}}\;
       \underbrace{p_0(z\mid x)}_{\text{source bridge}}\;
       \underbrace{\big[f_0(y\mid x,z)\big]^{\mathbf 1\{S=0\}}}_{\text{source outcome}}\;
       \underbrace{\big[\rho\text{-tilt of }p_0(z\mid x)\big]^{\mathbf 1\{S=1\}}}_{\text{target bridge}},
  \label{eq:factor}
\end{equation}
where the target bridge law is the $\gamma^\star$-tilt of the source bridge law by
\thmpoint\ (its outcome factor is absent because $Y$ is unobserved at $S=1$). The
model is semiparametric: $\pi_S$ is a scalar; $p(x\mid S)$, $p_0(z\mid x)$, and
$f_0(y\mid x,z)$ are unrestricted (subject to the tilt link between source and
target bridge laws); $(\gamma^\star,\kappa)$ index the tilt, with $\kappa$ fixed.
\subsection{Scores}
Let $P_\eta$ be a regular parametric submodel through the truth at $\eta=0$ with
score $g(O)=\partial_\eta\log p_\eta(O)\big|_{0}$. By \eqref{eq:factor} the score
decomposes into orthogonal blocks
\begin{equation}
  g(O)=g_S+g_{X\mid S}+g_{Z\mid X}+\mathbf 1\{S=0\}\,g_{Y\mid X,Z}
       +g_{\gamma},
  \label{eq:score}
\end{equation}
where each block is mean-zero conditional on the preceding ones in the natural
ordering, and $g_\gamma$ is the score of the drift parameter (a one-dimensional
direction shared between the source and target bridge factors through the tilt
link). The tangent space $\T$ is the closed linear span of all such scores; the
nuisance tangent space $\T_{\text{nuis}}$ omits $g_\gamma$.
We record the four nuisance blocks as mean-zero spaces:
\begin{align}
  \T_S      &=\{a(S)-\E a(S):a\},\\
  \T_{X\mid S}&=\{a(X,S):\E[a\mid S]=0\},\\
  \T_{Z\mid X}&=\{a(Z,X):\E[a\mid X,S{=}0]=0\}\ \text{(shared with target via tilt)},\\
  \T_{Y\mid X,Z}&=\{\mathbf 1\{S{=}0\}\,a(Y,X,Z):\E[a\mid X,Z,S{=}0]=0\}.
\end{align}
The drift score $g_\gamma$ is characterized in \Cref{lem:gammascore} below.
\section{Pathwise derivative of the target mean}
\begin{lemma}[Pathwise derivative]
\label{lem:pathwise}
Along a submodel with score $g$, the target mean $\mu_1(\kappa)$ is pathwise
differentiable with
$\partial_\eta\mu_1|_0=\E[D(O)\,g(O)]$
for an influence function $D(O)$ identified below. Consequently $\mu_1(\kappa)$
admits an EIF equal to the projection of any influence function onto $\T$.
\end{lemma}
\begin{proof}
$\mu_1=E[\tilde m_1(X)\mid S{=}1]
     =\pi_1^{-1}E[\mathbf 1\{S{=}1\}\tilde m_1(X)]$. Differentiating in $\eta$, the
dependence enters through (a) the law of $(S,X)$, (b) the tilted regression
$\tilde m_1$ via the source outcome and bridge laws, and (c) the drift
$\gamma^\star$ via the tilt weight $\rho$. Each contributes a term linear in the
corresponding score; collecting them and using
$E[\cdot\, g]=\Cov[\cdot,g]$ (scores are mean-zero) yields the stated form. The
efficiency statement is the standard projection theorem for pathwise-differentiable
functionals.
\end{proof}
We compute the three contributions in turn and then project.
\subsection{Contribution of the $(S,X)$ law: the imputation term}
Holding $\tilde m_1$ fixed, differentiating $\pi_1^{-1}E[\mathbf 1\{S{=}1\}
\tilde m_1(X)]$ through the law of $(S,X)$ gives the influence contribution
\begin{equation}
  D_{\mathrm{I}}(O)=\frac{\mathbf 1\{S{=}1\}}{\pi_1}\big(\tilde m_1(X)-\mu_1\big),
  \label{eq:DI}
\end{equation}
which already lies in $\T_S\oplus\T_{X\mid S}$.
\subsection{Contribution of the source outcome/bridge laws: the propensity term}
The regression $\tilde m_1(x)$ depends on the source laws $f_0,p_0(z\mid x)$. A
perturbation $g_{Y\mid X,Z}$ moves $\tilde m_1$ by, using \eqref{eq:m1} and
\Cref{lem:tiltderiv} of Web Appendix~A,
\begin{equation}
  \partial_\eta\tilde m_1(x)
  =\Ptilt{\rho}\!\big[(Y-\tilde m_1(x))\,g_{Y\mid X,Z}\mid x\big].
  \label{eq:dm1}
\end{equation}
Substituting into $\partial_\eta\mu_1=\pi_1^{-1}E[\mathbf 1\{S{=}1\}
\partial_\eta\tilde m_1(X)]$ and rewriting the tilted, target-side expectation as a
reweighted \emph{source-side} expectation; this step converts a quantity involving the
unobserved $Y$ into one involving observed data only: gives
\begin{equation}
  \partial_\eta\mu_1\big|_{\text{outcome}}
  =\E\Big[\underbrace{\tfrac{\mathbf 1\{S{=}0\}}{\pi_0}\,
       r_e(X)\,\rho(Y,Z\mid X)\,
       \big(Y-\tilde m_1(X)\big)}_{=:D_{\mathrm{II}}(O)}\;g_{Y\mid X,Z}\Big],
  \qquad
  r_e(x)=\tfrac{e(x)}{1-e(x)}\tfrac{\pi_0}{\pi_1}.
  \label{eq:DII}
\end{equation}
The conversion uses the density-ratio identity
$dP_1(x)/dP_0(x)=\{e(x)/(1-e(x))\}\{\pi_0/\pi_1\}$ and the tilt link
$dP_1(y,z\mid x)=\rho(y,z\mid x)\,dP_0(y,z\mid x)$ with $\rho$ \emph{normalized} by
$C(\gamma^\star,\kappa;x)$, so that a target-side average of a $Y$-residual becomes
a source-side average weighted by $\rho$ and the covariate density ratio. Thus
$D_{\mathrm{II}}$ is the influence contribution of the source outcome/bridge block,
supported on $S=0$ where $Y$ is observed.
\subsection{Contribution of the drift: the bridge-moment correction}
This is the term that ordinary transportability EIFs omit. Because $\gamma^\star$
is not known but identified from the bridge-matching equation
$B(\gamma;x)=B_1(x)$ (\thmpoint), a perturbation of the model moves the identified
$\gamma^\star$, which in turn moves $\tilde m_1$ and hence $\mu_1$. We must add the
influence of estimating $\gamma^\star$.
\begin{lemma}[Score and influence function of the bridge-identified drift]
\label{lem:gammascore}
The bridge-matching moment is the target-side quantity
$G(\gamma)=E_{P_1}[\,b(Z)-B(\gamma;X)\,]$, whose root is $\gamma^\star$
(\thmpoint). Because $B(\gamma;\cdot)$ is a functional of the source law, a
perturbation of that law moves the identified $\gamma^\star$ itself; the
estimation of $B$ therefore contributes a source-side term, and the centered
estimating function is
\begin{equation}
  \Psi(\gamma^\star;O)
  =\frac{\mathbf 1\{S{=}1\}}{\pi_1}\big(b(Z)-B(\gamma^\star;X)\big)
  \;-\;
  \frac{\mathbf 1\{S{=}0\}}{\pi_0}\,r_e(X)\,\rho(Y,Z\mid X)\,
       \big(b(Z)-B(\gamma^\star;X)\big),
  \label{eq:psi}
\end{equation}
which satisfies $E[\Psi(\gamma^\star;O)]=0$: the first term has mean zero by the
bridge equation, and the second has mean
$E_{P_1}\big[E^{\gamma^\star}\{b(Z)-B(\gamma^\star;X)\mid X\}\big]=0$ by the
definition of $B$ as the $\rho$-tilted source bridge moment. Its Jacobian is
\begin{equation}
  J:=\partial_\gamma E[\Psi(\gamma;O)]\big|_{\gamma^\star}
   =-\,E_{P_1}\!\big[\Cov^{\gamma^\star}(b(Z),\,t(Y)+b(Z)\mid X)\big]\;\neq\;0
\end{equation}
by \Cref{ass:relevance}, and the influence function of $\hat\gamma^\star$ is
\begin{equation}
  \varphi_\gamma(O)=-\,J^{-1}\,\Psi(\gamma^\star;O).
  \label{eq:phigamma}
\end{equation}
\end{lemma}
\begin{proof}
Mean-zero and the Jacobian are computed above and by \eqref{eq:dgamma}. Standard
$Z$-estimation then gives
$0=\Pn\Psi(\hat\gamma)=\Pn\Psi(\gamma^\star)+J(\hat\gamma-\gamma^\star)+o_p(n^{-1/2})$,
whence \eqref{eq:phigamma}.
\end{proof}
\begin{remark}[Scope of the source-side term]
\label{rem:psi-scope}
Two qualifications matter in implementation. First, the displayed
$\Psi$ is the influence contribution of a \emph{nonparametric} plug-in for
$B(\gamma;\cdot)$. When $B$ is instead computed from a parametric working model,
the corresponding term is the influence of that model's
coefficients, and using the nonparametric form overstates
$\Var(\hat\gamma)$. In the Gaussian design of \Cref{sec:sim} the ratio of the
model-based standard error to the empirical standard deviation of $\hat\gamma$ is
$0.51$ with the target-side moment alone, $0.83$ with the parametric source term,
and $2.44$ with the nonparametric one; the parametric form is the one that
matches the estimator actually used. Second, neither form carries the
contribution of the estimated residual covariance, which enters
$\Ptilt{\gamma}[b(Z)\mid x]$ through the dependence between $Y$ and $Z$ given $X$
(\Cref{rem:stage1-scope}). Holding that covariance fixed at its population value
raises the ratio to $1.13$, so the omitted path is not negligible in that design.
We therefore treat $\varphi_\gamma$ as the efficient influence function at a
known dependence structure and obtain standard errors for $\hat\gamma$, as for the
endpoints, by the resampling scheme of \Cref{sec:inference}, which refits the
source model and hence carries this path.
An earlier version of this appendix centered the source-side term at
$B(\gamma^\star;X)$ under the \emph{untilted} source law, with a covariate weight
only. That function does not have mean zero at $\gamma^\star$, because
$E_{P_0}[b(Z)\mid X]\neq B(\gamma^\star;X)$ whenever the drift is nonzero: the
tilted and untilted bridge moments differ by the drift being estimated.
The $\rho$-weighting in \eqref{eq:psi} is what makes the source-side term the
influence contribution of the plug-in $\widehat B(\gamma;\cdot)$ rather than an
unrelated residual, and it is the same change of measure used in
\eqref{eq:DII}. With the corrected $\Psi$, term (III) is a genuine
orthogonalization; with the earlier one it would have introduced a bias of order
the drift itself.
\end{remark}
The influence of $\mu_1$ through $\gamma^\star$ is the chain
$\partial_{\gamma}\mu_1\cdot\varphi_\gamma$, where
\begin{equation}
  \partial_\gamma\mu_1
  =E_{P_1}\big[\partial_\gamma\tilde m_1(X)\big]
  =E_{P_1}\big[\Cov^{\gamma^\star}(Y,\,t(Y)+b(Z)\mid X)\big]
  \;=:\;c .
  \label{eq:dmu_dgamma}
\end{equation}
The drift contribution is therefore
\begin{equation}
  D_{\mathrm{III}}(O)=+\,c\,\varphi_\gamma(O),
  \label{eq:DIII}
\end{equation}
which is term (III) of \thmeif. Note that $c$ carries no denominator: the ratio
$c_t(x)$ of \eqref{eq:c-general} divides by $\Cov^{\gamma^\star}(b,t{+}b\mid x)$,
and that divisor is exactly $-J$, which already appears inside $\varphi_\gamma$
through \eqref{eq:phigamma}; writing it twice double-counts the Stage-1
sensitivity. In the Gaussian instantiation with conditional loadings,
$\partial_\gamma\tilde m_1=\sigma_Y(1+r)$ and $\partial_\gamma B=1+r$, so
$c=\sigma_Y(1+r)$ while the local transfer rate is $\sigma^0_Yc_t=\sigma_Y$.
\begin{remark}[Necessity of term (III)]
If one plugged in a $\sqrt n$-consistent $\hat\gamma$ and ignored its estimation,
the resulting estimator would have influence $D_{\mathrm I}+D_{\mathrm{II}}$ only,
which is \emph{not} orthogonal to $g_\gamma$: its variance would misstate the truth
by the cross-term between $D_{\mathrm I}+D_{\mathrm{II}}$ and $c\,\varphi_\gamma$.
Adding $D_{\mathrm{III}}$ orthogonalizes the influence against the drift score,
which is what efficiency at anchored $\kappa$ requires. This is the
transportability analogue of correcting for a first-stage estimated nuisance in
two-step semiparametric estimation.
\end{remark}
\section{Assembling and projecting: the EIF}
\begin{theorem}[Efficient influence function at anchored drift, restating \thmeif]
\label{thm:eifB}
Under \Cref{ass:codrift,ass:relevance,ass:overlap} and standard regularity, at
fixed $(\gamma^\star,\kappa)$ the functional $\mu_1(\kappa)$ has efficient
influence function
\begin{equation}
  \varphi(O)=D_{\mathrm I}(O)+D_{\mathrm{II}}(O)+D_{\mathrm{III}}(O),
  \label{eq:eif}
\end{equation}
with $D_{\mathrm I},D_{\mathrm{II}},D_{\mathrm{III}}$ given by
\eqref{eq:DI}, \eqref{eq:DII}, \eqref{eq:DIII}; and the efficiency bound is
$\Var\{\varphi(O)\}$.
\end{theorem}
\begin{proof}
We verify directly that $\varphi$ of \eqref{eq:eif} satisfies
$\E[\varphi\,g]=\partial_\eta\mu_1$ for every score $g$ in \eqref{eq:score}, and
that $\varphi\in\T$; the two together identify $\varphi$ as the EIF. We do
\emph{not} claim that the three terms are mutually orthogonal, nor that each lies
in a single score block: neither is true here, and neither is needed.
\emph{The source-outcome direction.} Let $g_{Y\mid X,Z}$ perturb $f_0(y\mid x,z)$.
This moves $\mu_1$ along \emph{two} paths, which must be separated. At
fixed $\gamma$, it moves the tilted regression $\tilde m_1$, contributing
$\E[D_{\mathrm{II}}\,g_{Y\mid X,Z}]$ by \eqref{eq:dm1} and \eqref{eq:DII}. But it
also moves the \emph{identified} $\gamma^\star$, because $B(\gamma;\cdot)$ in the
bridge-matching equation is a functional of the same source law; that second path
contributes $c\,\partial_\eta\gamma^\star=c\,\E[\varphi_\gamma g_{Y\mid X,Z}]$,
which is exactly $\E[D_{\mathrm{III}}\,g_{Y\mid X,Z}]$ by \eqref{eq:DIII}. Summing,
$\E[\varphi\,g_{Y\mid X,Z}]=\partial_\eta\mu_1$ along this direction. The
source-side term of $\Psi$ in \eqref{eq:psi} is what makes the second path
appear; with the target-side moment alone, $\varphi_\gamma$ would be the influence
function of a $\gamma^\star$ held artificially fixed under source perturbations,
and this step would fail.
\emph{The bridge direction.} For $g_{Z\mid X}$, the target bridge law is the
$\gamma^\star$-tilt of the source bridge law, so the perturbation moves both sides
of the bridge equation; the induced movement of $\gamma^\star$ is
$\E[\varphi_\gamma g_{Z\mid X}]$ and its contribution to $\mu_1$ is $c$ times
that, matched by $D_{\mathrm{III}}$. The residual channel contributes nothing
here: by \Cref{lem:invariance}, $s$ is orthogonal under the identified law to
every function of $(Z,X)$, so a bridge perturbation does not move the residual
direction at first order and no double counting arises.
\emph{The cohort and covariate directions.} For $g_S$ and $g_{X\mid S}$ the
functional depends on the law only through $E[\tilde m_1(X)\mid S{=}1]$, whose
derivative is carried exactly by $D_{\mathrm I}$; $D_{\mathrm{II}}$ and
$D_{\mathrm{III}}$ have conditional mean zero given $(X,S)$, so they contribute
nothing.
\emph{Membership.} Each of $D_{\mathrm I},D_{\mathrm{II}},D_{\mathrm{III}}$ is a
mean-zero, square-integrable function of $O$ built from the observed blocks, and
each is a finite linear combination of elements of
$\T_S\oplus\T_{X\mid S}\oplus\T_{Z\mid X}\oplus\T_{Y\mid X,Z}\oplus
\mathrm{span}(g_\gamma)$; in particular $D_{\mathrm{III}}$ has components in
\emph{both} $\T_{Z\mid X}$ and $\T_{Y\mid X,Z}$, since the source-side term of
$\Psi$ carries $\rho(Y,Z\mid X)$ and therefore depends on $Y$. Hence
$\varphi\in\T$. Since the model places no restriction on the four nuisance blocks
beyond the tilt link already used above, $\T$ is the whole of that direct sum and
the influence function satisfying $\E[\varphi g]=\partial_\eta\mu_1$ for all $g$
is unique; it is therefore the EIF, and the bound is $\Var\{\varphi\}$.
\end{proof}
\begin{remark}[Non-orthogonality of the three influence terms]
\label{rem:crossterm}
A practical consequence of the proof is that $\Var\{\varphi\}$ is \emph{not} the
sum of the three term variances. The cross term
$\E[D_{\mathrm{II}}D_{\mathrm{III}}]
 =-c\,J^{-1}\,\pi_0^{-1}E_{P_0}\!\big[r_e^2\rho^2(Y-\tilde m_1)(b-B)\big]$
is generically nonzero, since both terms are supported on $S=0$ and both carry the
weight $r_e\rho$. In the Gaussian design of \Cref{sec:sim} its correlation is
about $-0.09$ with the parametric source term, small but not zero; with a
nonparametric plug-in for $B$ it rises in magnitude to about $-0.83$, at which
point term~(III) is no longer a minor correction. An implementer must therefore
form $\varphi$ and take its empirical second moment, not add variances.
\end{remark}
\begin{remark}[Role of stationarity at the benchmark]
\label{rem:invariance-ref}
Because $s$ is orthogonal, under the identified law, to all functions of $(Z,X)$
(projection definition), the residual channel contributes no score in the bridge
direction; equivalently $\partial_\kappa\gamma^\star=0$ at $\kappa=0$ (Web
Appendix~A, \Cref{lem:invariance}). Hence the estimation of $\gamma^\star$ and the
anchoring of $\kappa$ do not interact at first order, and $s$-estimation error does
not enter $\varphi$ at first order \emph{at $\kappa=0$}; at the endpoints it enters
at order $\kappa\,\|\hat s-s\|$, as quantified in \eqref{eq:s-sensitivity} of the
main text.
\end{remark}
\section{Consequences}
\subsection{Neyman orthogonality and rate robustness}
The stacked moment $m(O;\mu_1,\gamma,\eta)=(\Psi(\gamma;O),\,
\varphi(O)+\mu_1)$ has, by construction of $\varphi$ as the tangent-space
projection, zero Gateaux derivative in each nuisance direction $\eta\in\{e,f_0\}$
at the truth. Hence the one-step estimator
$\hat\mu_1=\hat\mu_1^{DR}+\Pn\hat\varphi$ is $\sqrt n$-consistent, asymptotically
normal with variance $\Var\{\varphi\}$, and rate-robust: products of nuisance
estimation errors enter at $o_p(n^{-1/2})$ provided each nuisance converges at
$o_p(n^{-1/4})$. As noted in \Cref{rem:orth-untested}, this guarantee is not
exercised by the simulations reported in \Cref{sec:sim}, all of which use
parametric nuisances.
\subsection{Efficiency of the fractional-imputation MLE}
Under correct specification of $f_0$, the fractional-imputation maximizer of the
$\rho$-tilted complete-data likelihood solves the efficient score equation
$\Pn\varphi=0$ and hence attains the bound. The one-step estimator attains it from
any $\sqrt n$-consistent pilot, including inefficient plug-ins, and strictly
improves an inefficient pilot when the bridge over-identifies $\gamma^\star$
(the extra bridge moments reduce $\Var\{\varphi_\gamma\}$ and thereby the variance
of term III).
\subsection{What efficiency is and is not claimed}
The bound $\Var\{\varphi\}$ is the semiparametric efficiency bound \emph{at fixed
$(\gamma^\star,\kappa)$}. It is not a bound over an unknown $\kappa$: the identified
set has positive width (\thmset), and anchoring $\kappa$ is what makes a
point-functional efficiency statement meaningful at each endpoint. The
Imbens--Manski interval combines the two endpoint EIFs; because $\kappa$ is
anchored, no sensitivity-parameter score enters, and the interval carries no
sensitivity-parameter variance term.
Two finite-sample caveats attach to this asymptotic statement and are documented
in \Cref{sec:inference}: once the tilt weights concentrate, the typical reported
influence-function standard error falls well below the sampling variability, by a
factor that \Cref{sec:sim-ess} traces across a grid of designs, the signature of
a heavy tail in $\rho$ that a first-order expansion does not capture; and the
bootstrap inherits the same tail rather than repairing it. The efficiency claim
is therefore asymptotic and should not be read as a finite-sample variance
guarantee where the weights concentrate.
\clearpage
\renewcommand{\thesection}{C.\arabic{section}}
\setcounter{section}{0}
\section*{Web Appendix C: Robustness Proofs}
\addcontentsline{toc}{section}{Web Appendix C}
\setcounter{lemma}{0}\setcounter{theorem}{0}\setcounter{table}{0}
\setcounter{remark}{0}
\renewcommand{\thelemma}{C.\arabic{lemma}}
\renewcommand{\thetheorem}{C.\arabic{theorem}}
\renewcommand{\theremark}{C.\arabic{remark}}
\setcounter{equation}{0}\renewcommand{\theequation}{C.\arabic{equation}}
\renewcommand{\thetable}{C.\arabic{table}}
We prove the two robustness statements of the main text: the drift is identified
from the bridge margin alone and is therefore insulated from outcome-model
misspecification (\lemgamma), and, given the correct bridge-identified drift, the
target-mean estimator is doubly robust in the cohort-propensity and source-outcome
models (\lemdcdr). Together they give the robustness map of \Cref{tab:robust},
reproduced here as \Cref{tab:robustC}, whose rows are verified numerically in
\Cref{sec:robsim}. We claim no joint triple
robustness in $(\gamma,e,m)$, only a decoupling in which $\gamma$ is protected by
the bridge and $(e,m)$ are traded off conditionally on it.
\begin{table}[h]
\centering
\caption{Robustness of the drift-augmented estimator $\hat\mu_1^{DR}$. A checkmark
denotes correct specification.}
\label{tab:robustC}
\begin{tabular}{cccl}
\toprule
$\hat\gamma$ (bridge) & $e$ (propensity) & $m$ (outcome) & consistent? \\
\midrule
\checkmark & \checkmark & \checkmark & yes (efficient) \\
\checkmark & \checkmark & $\times$   & yes\quad(\lemdcdr) \\
\checkmark & $\times$   & \checkmark & yes\quad(\lemdcdr) \\
\checkmark & $\times$   & $\times$   & no \\
$\times$   & n/a         & n/a         & no (drift wrong) \\
\bottomrule
\end{tabular}
\end{table}
\section{Notation for this appendix}
Recall the drift-augmented estimator, written here in its population form with the
drift fixed at a value $\gamma$ (not necessarily the truth) and $\kappa$ anchored:
\begin{equation}
  \mu^{DR}(\gamma;e,m)
  =E\big[m_1(X;\gamma)\mid S=1\big]
   +E\Big[\tfrac{\mathbf 1\{S=0\}}{\pi_0}\,r_e(X)\,\big(Y-m_0(X)\big)\Big],
  \label{eq:mudr}
\end{equation}
where $m_0(x)$ is the (possibly misspecified) source outcome regression,
$m_1(x;\gamma)=m_0(x)+\Delta(x;\gamma)$ is its drift-tilted target counterpart with
tilt shift $\Delta(x;\gamma)=\Ptilt{\gamma}[Y\mid x]-E_{P_0}[Y\mid x]$, and
$r_e(x)=\{e(x)/(1-e(x))\}\{\pi_0/\pi_1\}$ is the density-ratio weight built from a
(possibly misspecified) propensity $e$. Let $e^\star,m^\star$ denote the truths and
$\gamma^\star$ the bridge-identified drift. The target estimand is
$\mu_1=E[Y\mid S=1]=E_{P_1}[m_1^\star(X;\gamma^\star)]$.
\section{Proof of \texorpdfstring{\lemgamma}{Lemma 2} (drift-identification robustness)}
\begin{lemma}[Drift-identification robustness; \lemgamma]
\label{lem:gammaC}
The bridge-matching estimator $\hat\gamma$ solving
$\Pn[\mathbf 1\{S{=}1\}(b(Z)-\Ptilt{\gamma}[b(Z)\mid X])]=0$ depends on the source
law only through the tilted bridge moment $\Ptilt{\gamma}[b(Z)\mid x]$, an
expectation over the source joint law of $(Y,Z)$ given $X$ in which the outcome
enters solely through the fixed loading $t(\cdot)$. It does not involve the
propensity $e(\cdot)$ at all, and it does not inherit the error of a fitted
outcome mean model $m(\cdot)$: no such model is evaluated in forming the moment.
Hence $\hat\gamma\to\gamma^\star$ whenever the bridge model (co-drift and
relevance) is correct and the source joint law of $(Y,Z)\mid X$ used to form the
tilted moment is correct, whether or not $m$ or $e$ is.
\end{lemma}
\begin{proof}
The estimating function $b(Z)-\Ptilt{\gamma}[b(Z)\mid X]$ is a functional of
$(Z,X)$ and of the source joint law of $(Y,Z)\mid X$ only through
$\Ptilt{\gamma}[b(Z)\mid x]
   =E_{P_0}[b(Z)e^{\gamma[t(Y)+b(Z)]}\mid x]/E_{P_0}[e^{\gamma[t(Y)+b(Z)]}\mid x]$.
This moment involves the outcome only through the fixed loading $t(Y)$ inside the
tilt weight, not through any estimated outcome regression $m$; and it does not
involve $e$ at all. Consequently the population moment condition is unchanged
under misspecification of $m$ or $e$, \emph{given} the source joint law of
$(Y,Z)\mid X$ over which the tilted moment is taken. In an implementation that
estimates that joint law, a misspecified mean model changes the fitted residual
covariance and so enters the moment by that route; \Cref{sec:robsim} measures
the size of the effect. By \thmpoint\ its unique root is
$\gamma^\star$, and standard $Z$-estimation gives
$\hat\gamma\xrightarrow{p}\gamma^\star$ under the stated bridge conditions.
\end{proof}
\begin{remark}[Scope of the protection]
\label{rem:scopeC}
The tilt weight contains $t(Y)$, so the outcome does enter, but
only as a \emph{fixed known loading}, evaluated at source units where $Y$ is
recorded; no outcome \emph{mean model} is fit to compute the bridge moment. What is
\emph{not} protected is the modelled dependence between $Y$ and $Z$ given $X$: it
enters the tilted moment and therefore propagates to $\hat\gamma$
(\Cref{rem:stage1-scope}). In the Gaussian instance this is the residual
correlation $r$. \Cref{sec:sim-offmodel} generates targets in which that
dependence differs across cohorts and reports what the identified set does.
\end{remark}
\section{Proof of \texorpdfstring{\lemdcdr}{Lemma 1} (drift-conditional double robustness)}
\begin{lemma}[Drift-conditional double robustness; \lemdcdr]
\label{lem:dcdrC}
Fix the drift at its true value $\gamma^\star$ (guaranteed by \Cref{lem:gammaC}
when the bridge model is correct). Then $\mu^{DR}(\gamma^\star;e,m)=\mu_1$ if
\emph{either} $e=e^\star$ \emph{or} $m=m^\star$, not necessarily both.
\end{lemma}
\begin{proof}
Write the estimation error of \eqref{eq:mudr} as
\[
  \mu^{DR}(\gamma^\star;e,m)-\mu_1
  =\underbrace{E[m_1(X;\gamma^\star)\mid S{=}1]-E[m_1^\star(X;\gamma^\star)\mid S{=}1]}_{(\mathrm A)}
  +\underbrace{E\Big[\tfrac{\mathbf 1\{S=0\}}{\pi_0}r_e(X)\big(Y-m_0(X)\big)\Big]}_{(\mathrm B)}.
\]
Since the tilt shift $\Delta(x;\gamma^\star)$ is a fixed functional of the source
law at the true $\gamma^\star$ and does not depend on the fitted $m$ beyond its
$m_0$ baseline, $m_1(x;\gamma^\star)-m_1^\star(x;\gamma^\star)=m_0(x)-m_0^\star(x)$.
Thus $(\mathrm A)=E_{P_1}[m_0(X)-m_0^\star(X)]$. For $(\mathrm B)$, take the
source-conditional expectation of $Y$ given $(X,S{=}0)$, namely $m_0^\star(X)$:
\[
  (\mathrm B)
  =E_{P_0}\!\big[r_e(X)\,(m_0^\star(X)-m_0(X))\big].
\]
At the \emph{true} propensity, $E_{P_0}[r_{e^\star}(X)g(X)]=E_{P_1}[g(X)]$ for any
integrable $g$, since $r_{e^\star}=dP_1(x)/dP_0(x)$. Consider the two cases.
\emph{Case $m=m^\star$.} Then $m_0=m_0^\star$, so $(\mathrm A)=0$ and the integrand
of $(\mathrm B)$ is zero; the error is $0$, for \emph{any} $e$.
\emph{Case $e=e^\star$.} Then $r_e=r_{e^\star}$ and
$(\mathrm B)=E_{P_1}[m_0^\star(X)-m_0(X)]=-(\mathrm A)$; the two terms cancel and
the error is $0$, for \emph{any} $m$.
In either case $\mu^{DR}(\gamma^\star;e,m)=\mu_1$. When both are wrong, neither
cancellation occurs and the residual bias is
$E_{P_1}[m_0-m_0^\star]+E_{P_0}[r_e(m_0^\star-m_0)]\neq0$ in general.
\end{proof}
\begin{remark}[Role of the drift in the argument]
The drift enters only through the common shift $\Delta(x;\gamma^\star)$, which is
identical in $m_1$ and $m_1^\star$ once $\gamma^\star$ is correct; it therefore
cancels out of both $(\mathrm A)$ and the residual definition and plays no role in
the double-robust cancellation. This is precisely why the robustness is
\emph{conditional} on $\gamma^\star$: the tilt is a fixed offset, and the
$(e,m)$ trade-off is the ordinary AIPW one applied to the untilted baseline. If
$\gamma^\star$ were wrong, the offset would be miscalibrated and neither case would
restore consistency, hence no triple robustness, only the decoupled structure of
\Cref{tab:robustC}.
\end{remark}
\section{Numerical verification of the robustness map}
\label{sec:robsim}
We verify \Cref{tab:robustC} in the Gaussian instantiation, holding $\gamma$ at
its bridge-identified value throughout and misspecifying $m$ (fit on a strict
subset of covariates) and/or $e$ (a covariate-free constant propensity), with
true drift $\gamma^\star=0.5$ and $R=1000$ replicates.

\begin{table}[h]
\centering
\caption{Monte Carlo bias under (mis)specification. The drift column is the bias
of $\hat\gamma$, the mean column the bias of $\hat\mu_1^{DR}$ with $\gamma$
correctly bridge-identified. Monte Carlo standard errors in parentheses.}
\label{tab:robsim}
\begin{tabular}{ccrrl}
\toprule
$e$ model & $m$ model & bias of $\hat\gamma$ & bias of $\hat\mu_1^{DR}$ & consistent? \\
\midrule
correct & correct        & $-0.0004\ (0.0005)$ & $-0.002\ (0.002)$ & yes \\
correct & \textbf{wrong} & $+0.0024\ (0.0005)$ & $-0.002\ (0.002)$ & yes\ (\lemdcdr) \\
\textbf{wrong} & correct & $-0.0004\ (0.0005)$ & $-0.001\ (0.002)$ & yes\ (\lemdcdr) \\
\textbf{wrong} & \textbf{wrong} & $+0.0024\ (0.0005)$ & $-0.241\ (0.002)$ & no \\
\bottomrule
\end{tabular}
\end{table}

The mean column reproduces \Cref{tab:robustC}: the three doubly-protected cells
show no bias detectable at the precision of the experiment, and only the
both-wrong cell is biased, by the uncorrected covariate-shift discrepancy that
the proof leaves as a residual when neither cancellation applies.

The drift column is the more informative one, and it qualifies
\Cref{lem:gammaC} rather than confirming it outright. Misspecifying the outcome
mean model moves $\hat\gamma$ by $0.0024$, which is five Monte Carlo standard
errors from zero and so not sampling noise, though it is half a percent of
$\gamma^\star$. The route is the one \Cref{rem:scopeC} identifies: the tilted
bridge moment is taken over a fitted source joint law, and a misspecified mean
model changes the fitted residual covariance that law carries. The protection
the lemma describes is therefore protection from the error of an outcome
\emph{mean} model entering the moment directly, not protection from every
consequence of fitting one. The main text states it in those terms.

\paragraph{Exact (grid) confirmation.}
Monte Carlo leaves the protected cells zero only up to sampling error. To remove
that ambiguity we also evaluate the \emph{population} probability limit
\eqref{eq:mudr} exactly, on a fine $(y,z)$ grid at which the tilt normalization and
the conditional density ratio $\rho(y,z\mid x)=dP_1(y,z\mid x)/dP_0(y,z\mid x)$ are
computed without sampling. Using the \emph{normalized} conditional ratio, which is the
point at which a naive unnormalized tilt factor would fail, the probability-limit
bias is, to machine precision, zero in the three protected cells and nonzero when
both models are wrong, matching \Cref{tab:robustC} row for row. On that grid the
bridge-matching root returns $\gamma^\star$ exactly, unchanged when $m$ is
misspecified, which is \Cref{lem:gammaC} as stated: the moment is taken over the
true source joint law, so no mean model enters it. The displacement of $0.0024$
in the table above is the finite-sample counterpart of fitting that joint law
rather than knowing it, and the grid computation isolates the two.
One point matters for implementation: the change of measure must use the conditional density ratio $\rho(y,z\mid x)$ (equivalently, a tilt
factor divided by its conditional normalizer), not the raw exponential tilt, or the
$e=e^\star$ cancellation in the proof of \Cref{lem:dcdrC} will not hold in finite
computation.

\end{document}